\documentclass[acmsmall]{acmart}

\setcopyright{cc}
\setcctype{by}
\acmDOI{10.1145/3720476}
\acmYear{2025}
\acmJournal{PACMPL}
\acmVolume{9}
\acmNumber{OOPSLA1}
\acmArticle{120}
\acmMonth{4}
\received{2024-10-16}
\received[accepted]{2025-02-18}

\usepackage{booktabs}   %
\usepackage{subcaption} %

\usepackage[utf8]{inputenc}            %
\usepackage[T1]{fontenc}               %
\usepackage{amsmath,mathtools,amsthm}  %
\usepackage{graphicx}
\usepackage{thmtools,thm-restate}      %
\usepackage{mathpartir}                %
\usepackage{stmaryrd}                  %
\usepackage{xspace}                    %

\usepackage{url}
\usepackage{xcolor}                    %
\usepackage{tikz}                      %
\usetikzlibrary{trees}
\usepackage[ruled, noend]{algorithm2e}
\usepackage[normalem]{ulem}
\usepackage{listings}
\usepackage{enumerate}                 %
\usepackage{enumitem}
\usepackage{cleveref}                  %
\usepackage{placeins}         %
\usepackage{threeparttable}
\usepackage{fontawesome}               %
\usepackage{pifont}                    %
\usepackage[scaled=0.85]{beramono}
\usepackage[old]{old-arrows}
\usepackage{trimclip}
\DeclareMathAlphabet{\mathbbm}{U}{bbm}{m}{n}
\usepackage{hyperref}
\usepackage{trimclip}

\usepackage{quiver}
\usepackage[hyphenbreaks]{breakurl}

\DeclareMathAlphabet{\mathbsf}{\encodingdefault}{\sfdefault}{bx}{n}

\newcommand{\LangName}[1]{\textsf{#1}}

\newcommand{\Met}{\textsf{\textsc{Met}}\xspace}
\newcommand{\SurfaceMet}{\textsf{\textsc{Metl}}\xspace}
\newcommand{\Metl}{\textsf{\textsc{Metl}}\xspace}

\newcommand{\CalcM}{\Met}

\newcommand{\Leff}{\ensuremath{\mathrm{F}^1_{\mathrm{eff}}}\xspace}

\newcommand{\OCaml}{\LangName{OCaml}\xspace}
\newcommand{\Scala}{\LangName{Scala}\xspace}
\newcommand{\Links}{\LangName{Links}\xspace}

\newcommand{\Eff}{\LangName{Eff}\xspace}
\newcommand{\Koka}{\LangName{Koka}\xspace}
\newcommand{\Effekt}{\LangName{Effekt}\xspace}

\newcommand{\CaptureCalculus}{\ensuremath{\LangName{CC}_{<:\BoxSym}}\xspace}

\newcommand{\Frank}{\LangName{Frank}\xspace}
\newcommand{\Helium}{\LangName{Helium}\xspace}

\definecolor{myred}{HTML}{BB0000}
\definecolor{myblue}{HTML}{003399}
\definecolor{dred}{HTML}{EB212E}
\definecolor{dblue}{HTML}{2E67F8}

\newcommand{\red}[1]{{\color{myred}{#1}}}

\newcommand{\dblue}[1]{{\color{dblue}{#1}}}
\newcommand{\dred}[1]{{\color{dred}{#1}}}
\newcommand{\gray}[1]{{\color{codegray}{#1}}}
\newcommand{\black}[1]{{\color{black}{#1}}}

\colorlet{hlcolor}{lightgray}
\newcommand{\highlightwithstyle}[2]{
  {\setlength{\fboxsep}{2pt} \pgfsetfillopacity{0.4} \colorbox{hlcolor}{\pgfsetfillopacity{1}$#1#2$}}
}
\newcommand{\hl}[1]{\mathpalette\highlightwithstyle{#1}}

\makeatletter
\newcommand{\sigb}[2]{
  \@ifmtarg{#1}{:}{:^{#1}} #2
}
\makeatother

\makeatletter
\newcommand{\ogeneric}[2][0.7]{%
  \vphantom{{\oplus}}\mathpalette\o@generic{{#1}{#2}}%
}
\newcommand{\o@generic}[2]{\o@@generic#1#2}
\newcommand{\o@@generic}[3]{%
  \begingroup
  \sbox\z@{$\m@th#1\oplus$}%
  \dimen@=\dimexpr\ht\z@+\dp\z@\relax
  \savebox\tw@[\totalheight]{$\m@th#1\bigcirc$}%
  \makebox[\wd\z@]{%
    \ooalign{%
      $#1\vcenter{\hbox{\resizebox{\dimen@}{!}{\usebox\tw@}}}$\cr
      \hidewidth
      $#1\vcenter{\hbox{\resizebox{#2\dimen@}{!}{$#1\vphantom{\oplus}{#3}$}}}$%
      \hidewidth
      \cr
    }%
  }%
  \endgroup
}
\makeatother

\newcommand{\meta}[1]{\mathsf{#1}}

\newcommand{\all}{\spadesuit}
\newcommand{\one}{\aid}

\let\BoxSym\Box
\newcommand{\squareop}[1]{%
  {\mathrel{\ooalign{\hss\raise-0.1ex\hbox{\scalebox{1.1}{$\BoxSym$}}\hss\cr%
  \kern0.3ex\raise0.12ex\hbox{\scalebox{0.65}{$#1$}}}}}
}

\makeatletter
\DeclareFontEncoding{LS2}{}{\noaccents@}
\DeclareFontSubstitution{LS2}{stix2}{m}{n}
\DeclareSymbolFont{arrows3}{LS2}{stixtt}{m}{n}
\DeclareMathSymbol{\squarelrblackbin}{\mathord}{arrows3}{"89}

\DeclareFontEncoding{LS1}{}{}
\DeclareFontSubstitution{LS1}{stix}{m}{n}
\DeclareSymbolFont{symbols2}{LS1}{stixfrak}{m}{n}
\DeclareMathSymbol{\typecolon}{\mathbin}{symbols2}{"25}
\makeatother

\makeatletter
\newcommand{\Scale}[2][1]{\scalebox{#1}{$\m@th#2$}}
\makeatother

\newcommand{\boxwith}[1]{{\dblue{#1}}}
\newcommand{\mind}[2]{#1_{#2}}
\newcommand{\updlock}[2]{{\llparenthesis #1 \rrparenthesis}_{#2}}

\newcommand{\lock}{\text{\faLock}}

\newcommand{\mlock}{\lock}
\newcommand{\lockwith}[1]{\mlock_{\dblue{#1}}}

\makeatletter
\newcommand{\varb}[2]{
  \@ifmtarg{#1}{:}{:_{\dblue{#1}}} #2
}
\newcommand{\Letm}[2]{
  \@ifmtarg{#1}{\keyw{let}}{\keyw{let}_{#1}} \; \keyw{mod}_{#2}\;
}
\newcommand{\Letmf}[2]{
  \@ifmtarg{#1}{\keyw{let}}{\keyw{let}_{#1}} \; #2\;
}

\newcommand{\Casem}[1]{
  \@ifmtarg{#1}{\keyw{case}}{\keyw{case}_{#1}} \;
}
\makeatother

\renewcommand{\Box}{\keyw{mod}}

\newcommand{\amk}[1]{{\langle {#1}{\mkern 1mu \mid\!}\rangle}}
\newcommand{\aex}[1]{{\langle {\!\mid\mkern 1mu} #1\rangle}}
\newcommand{\adj}[2]{{\langle {#1} {\mkern 1mu \mid\mkern 1mu} #2 \rangle}}

\newcommand{\effrow}[2]{{\{ {#1} {\mkern 1mu \mid\mkern 1mu} #2 \}}}
\newcommand{\geffrow}[2]{\gray{\{ \black{#1} {\mkern 1mu \mid\mkern 1mu} #2 \}}}

\newcommand{\eapp}{\#}
\newcommand{\aid}{{\langle\rangle}}
\newcommand{\aeq}[1]{[#1]}

\newcommand{\aconst}[1]{\red{dontuseme}}

\newcommand{\aremove}[1]{\red{dontuseme}}
\newcommand{\act}[2]{#1(#2)}

\makeatletter
\newcommand{\superimpose}[2]{{%
  \ooalign{%
    \hfil$\m@th#1\@firstoftwo#2$\hfil\cr
    \hfil$\m@th#1\@secondoftwo#2$\hfil\cr
  }%
}}
\makeatother

\newcommand{\Mod}{\keyw{mod}}

\newcommand{\Letmod}{\keyw{let\; mod}}
\newcommand{\locks}[1]{\meta{locks}(#1)}

\newcommand{\To}{\Rightarrow}

\newcommand{\typm}[3]{#1 \vdash #2 \;\gray{@}\, \dblue{#3}}
\newcommand{\atmode}[1]{\;\gray{@}\, \dblue{#1}}

\makeatletter
\newcommand{\oset}[3][0ex]{%
  \mathrel{\mathop{#3}\limits^{
    \vbox to#1{\kern-2\ex@
    \hbox{$\scriptstyle#2$}\vss}}}}
\makeatother

\newcommand{\typmi}[5]{#1 \vdash #2 \mathrel{\gray{\Rightarrow}} #3 \;\gray{@}\, \dblue{#4}}
\newcommand{\typmc}[4]{#1 \vdash #2 \mathrel{\gray{\Leftarrow}}    #3 \;\gray{@}\, \dblue{#4}}

\newcommand{\typmp}[5]{#1 \vdash #2 \mathrel{\gray{\Leftarrow}}
  \dred{{\renewcommand{\color}[1]{} #5}} \mathrel{\gray{\Rightarrow}} #3 \;\gray{@}\, \dblue{#4}}

\newcommand{\free}[3]{#1{-}\meta{free}(#2,#3)}

\newcommand{\marker}[2]{\gray{\vardiamondsuit^{#1}_{#2}}}

\newcommand{\typl}[3]{#1 \vdash #2 \,!\, #3}
\newcommand{\encode}[3]{#1 \,!\, #2 \transto{#3}}
\newcommand{\Yield}{{\codesize\texttt{yield}}}
\newcommand{\topmod}[1]{\meta{topmod}(#1)}
\newcommand{\transto}[1]{\ \dashrightarrow #1}
\newcommand{\fortrans}[1]{#1}
\newcommand{\stransl}[2]{\ensuremath{\llbracket #1 \rrbracket}_{#2}}
\newcommand{\across}[4]{\meta{across}(#1;#2;#3;#4)}
\newcommand{\rebox}[3]{\meta{rebox}(#1;#2;#3)}

\makeatletter
\newcommand{\inc}[3]{
  \@ifmtarg{#1}{}{#1 \mathrel{\typecolon}}
  #2 \sqsubseteq #3
}
\makeatother

\newcommand{\Effect}{\meta{Effect}}
\newcommand{\Pure}{\meta{Abs}}
\newcommand{\Any}{\meta{Any}}

\newcommand{\Abs}{-}

\newcommand{\earr}[3]{#1 \to^{#3} #2}
\newcommand{\elambda}[1]{\lambda^{#1}}

\makeatletter
\DeclareRobustCommand{\Circle}{%
  \mathbin{\mathpalette\on@ntimes\relax}%
}
\newcommand{\on@ntimes}[2]{%
  \vcenter{\hbox{%
    \sbox0{\m@th$#1\otimes$}%
    \setlength\unitlength{\wd0}%
    \begin{picture}(1,1)
    \linethickness{0.35pt}
    \put(.5,.5){\circle{.8}}
    \end{picture}%
  }}%
}
\makeatother

\newcommand{\txmark}{\text{\textcolor{black}{\ding{51}}{\textcolor{black}{\kern-0.7em\ding{55}}}}}

\newcommand{\subtype}{\leqslant}

\newcommand{\ftv}[1]{\meta{ftv}(#1)}
\newcommand{\fv}[1]{\meta{fv}(#1)}
\newcommand{\dom}[1]{\meta{dom}(#1)}

\newcommand{\ol}[1]{\overline{#1}}

\newcommand{\evar}{\varepsilon}
\newcommand{\eminus}[2]{#1\backslash #2}

\newcommand{\Fork}{\keyw{fork}}
\NewDocumentCommand{\ForkC}{ m m O{dual} }{\Fork^{#2}_{#1}} %

\newcommand{\refa}[1]{{\color{red}    \renewcommand{\color}[1]{}{#1}\ifthenelse{\equal{#1}{}}{}{\,}(1)}}
\newcommand{\refb}[1]{{\color{blue}   \renewcommand{\color}[1]{}{#1}\ifthenelse{\equal{#1}{}}{}{\,}(2)}}
\newcommand{\refc}[1]{{\color{violet} \renewcommand{\color}[1]{}{#1}\ifthenelse{\equal{#1}{}}{}{\,}(3)}}
\newcommand{\refd}[1]{{\color{purple} \renewcommand{\color}[1]{}{#1}\ifthenelse{\equal{#1}{}}{}{\,}(4)}}
\newcommand{\refe}[1]{{\color{cyan}   \renewcommand{\color}[1]{}{#1}\ifthenelse{\equal{#1}{}}{}{\,}(5)}}
\newcommand{\reff}[1]{{\color{magenta}\renewcommand{\color}[1]{}{#1}\ifthenelse{\equal{#1}{}}{}{\,}(6)}}
\newcommand{\refg}[1]{{\color{brown}  \renewcommand{\color}[1]{}{#1}\ifthenelse{\equal{#1}{}}{}{\,}(7)}}
\newcommand{\refh}[1]{{\color{orange} \renewcommand{\color}[1]{}{#1}\ifthenelse{\equal{#1}{}}{}{\,}(8)}}

\renewcommand{\descriptionlabel}[1]{\hspace{\labelsep}\textrm{\color{gray!85!black}#1}}

\newcommand{\notsmall}{}

\newcommand{\slab}[1]{\textrm{#1}}
\newcommand{\semlab}[1]{\text{\scshape{E-#1}}}

\newcommand{\tylab}[1]{\text{\scshape{T-#1}}}
\newcommand{\btylab}[1]{\text{\scshape{B-#1}}}

\newcommand{\mtylab}[1]{\text{\scshape{MT-#1}}}
\newcommand{\rowlab}[1]{\text{\scshape{R-#1}}}

\newcommand{\var}[1]{\mathit{#1}}

\newcommand{\keyw}[1]{{{\mathbsf{#1}}}}

\newcommand{\Handle}{\keyw{handle}}
\newcommand{\Adapt}{\keyw{mask}}
\newcommand{\Mask}{\keyw{mask}}

\newcommand{\With}{\;\keyw{with}\;}

\newcommand{\In}{\;\keyw{in}\;}
\newcommand{\Do}{\keyw{do}\;}
\newcommand{\Ret}{\keyw{return}\;}

\renewcommand{\Case}{\keyw{case}\;}
\newcommand{\Casey}{\keyw{case}}
\newcommand{\Of}{\;\keyw{of}\;}

\newcommand{\Inl}{\keyw{inl}\;}
\newcommand{\Inr}{\keyw{inr}\;}

\newcommand{\Pair}[2]{#1 \ast #2}

\makeatletter
\DeclareRobustCommand{\circbullet}{\mathbin{\vphantom{\circ}\text{\circbullet@}}}
\newcommand{\circbullet@}{%
  \check@mathfonts
  \m@th\ooalign{%
    \clipbox{0 0 0 {\dimexpr\height-\fontdimen22\textfont2}}{$\bullet$}\cr
    $\circ$\cr
  }%
}
\DeclareRobustCommand{\bulletcirc}{\mathbin{\text{\bulletcirc@}}}
\newcommand{\bulletcirc@}{%
  \check@mathfonts
  \m@th\ooalign{%
    \raisebox{\fontdimen22\textfont2}{\clipbox{0 {\fontdimen22\textfont2} 0 0}{$\bullet$}}\cr
    $\circ$\cr
  }%
}
\makeatother

\newcommand{\Unit}{{()}}

\newcommand{\code}[1]{{\codesize\texttt{#1}}}

\newcommand{\codeatmode}[1]{{\codesize\gray{\texttt{@}}\; {\texttt{#1}}}}

\newcommand{\TUnit}{1}
\newcommand{\Int}{\code{Int}}

\newcommand{\sto}{\twoheadrightarrow}

\newcommand{\rresidual}[2]{#2\backslash #1}
\newcommand{\join}{\vee}
\newcommand{\toSet}[1]{\lfloor #1\rfloor}

\newcommand{\reducesto}{\leadsto}

\newcommand{\ba}{\begin{array}}
\newcommand{\ea}{\end{array}}

\newcommand{\bl}{\ba[t]{@{}l@{}}}
\newcommand{\el}{\ea}

\renewenvironment{displaymath}{\notsmall\[}{\]\normalsize\ignorespacesafterend}

\newenvironment{syntax}{\begin{displaymath}\ba{@{}l@{\quad}r@{~}c@{~}l@{}}}{\ea\end{displaymath}\ignorespacesafterend}

\newenvironment{reductions}{\begin{displaymath}\ba{@{}l@{\quad}@{}r@{~}c@{~}l@{}}}{\ea\end{displaymath}\ignorespacesafterend}

\newcommand{\EC}{\mathcal{E}}

\definecolor{codegreen}{rgb}{0,0.6,0}
\definecolor{codeblue}{rgb}{0,0,0.8}
\definecolor{codegray}{rgb}{0.4,0.4,0.4}
\definecolor{codepurple}{rgb}{0.58,0,0.82}
\definecolor{codeone}{HTML}{9B26B6}
\definecolor{codetwo}{HTML}{414C87}
\definecolor{backcolour}{rgb}{0.95,0.95,0.92}

\newcommand{\codesize}{\fontsize{8.7}{10.4}}
\newcommand{\rulesize}{\fontsize{9.3}{11.16}}

\lstdefinestyle{ran}{
    commentstyle=\color{codegray},
    numberstyle=\tiny,
    stringstyle=\ttfamily,
    basicstyle=\codesize\selectfont\ttfamily,
    breakatwhitespace=false,
    breaklines=true,
    captionpos=b,
    keepspaces=true,
    numbers=none,
    numbersep=5pt,
    showspaces=false,
    showstringspaces=false,
    showtabs=false,
    tabsize=2,
    columns=fullflexible,
    aboveskip=.9\medskipamount,
    belowskip=.9\medskipamount,
    keywords=[1]{
      do, mask, handle, with, let, in, resume, return, eff, type,
      fst, snd, if, then, else, case, of, fun, raise, maska, data
    },
    keywordstyle=[1]\bfseries,
    keywords=[2]{
      Int, List, Maybe, String, Bool, Pure, Any, Proc, Abs,
      true, false, nil, cons, just, nothing, proc
    },
    keywordstyle=[2]\color{codegreen},
    keywords=[3]{
      choose, fail, get, put, ufork, fork, yield, suspend, abort, log, throw, ask, foo, leak,
      bar, baz, Yield, Gen, State, Fork, Coop, Queue, UCoop
    },
    keywordstyle=[3]\color{dblue},
    morecomment = [l]{\#},
    literate={
      {->}{{$\to$}}{2}
      {=>}{{$\Rightarrow$}}{2}
      {>>}{{$\sto$}}{2}
      {~>}{{$\mapsto$}}{2}
      {\#>}{{$\sharparrow$}}{2}
      {<=}{{$\subtype$}}{1}
      {|-}{{$\vdash$}}{1}
      {|/-}{{$\nvdash$}}{1}
      {@}{{$\gray{\texttt{@}}$}}{1}
      {forall}{{$\forall$}}{1}
      {Unit}{{{\color{codegreen}1}}}{1}
      {Gamma}{{{$\Gamma$}}}{1}
      {earr}{{$\xrightarrow{\texttt{e}}$}}{2}
      {earrg}{{$\xrightarrow{\texttt{\leff{Gen} \ltype{Int}, e}}$}}{2}
    },
    escapeinside={<@}{@>}
}

\newcommand{\sharparrow}{\mathrel{\mkern3mu\raisebox{-.1ex}{\scalebox{1}[1]{\#}}\mkern-17mu\To}}
\newcommand{\lop}[1]{{\color{dblue}#1}}
\newcommand{\leff}[1]{{\color{dblue}#1}}
\newcommand{\ltype}[1]{{\color{codegreen}#1}}
\newcommand{\larr}[1]{$\xrightarrow{\texttt{#1}}$}

\lstdefinestyle{koka}{
    commentstyle=\color{codegray},
    numberstyle=\tiny,
    stringstyle=\ttfamily\small,
    basicstyle=\codesize\selectfont\ttfamily,
    breakatwhitespace=false,
    breaklines=true,
    captionpos=b,
    keepspaces=true,
    numbers=none,
    numbersep=5pt,
    showspaces=false,
    showstringspaces=false,
    showtabs=false,
    tabsize=2,
    columns=fullflexible,
    aboveskip=.9\medskipamount,
    belowskip=.9\medskipamount,
    keywords=[1]{
      do, mask, handle, with, let, in, resume, return, effect,
      fst, snd, if, then, else, case, of,
    },
    keywordstyle=[1]\bfseries,
    keywords=[2]{
      List, Int, Unit, int, Tuple2
    },
    keywordstyle=[2]\color{codegreen},
    keywords=[3]{
      choose, fail, get, put, ufork, fork, yield, suspend, abort, log, throw, ask, leak,
      bar, baz, Yield, Gen, State, Fork, Coop, gen
    },
    keywordstyle=[3]\color{dblue},
    morecomment = [l]{\#},
    literate={
      {=>}{{$\Rightarrow$}}{2}
      {->}{{$\to$}}{2}
    }
}

\begin{document}

\title{Modal Effect Types}   %

\author{Wenhao Tang}
\orcid{0009-0000-6589-3821}
\email{wenhao.tang@ed.ac.uk}
\affiliation{%
  \institution{The University of Edinburgh}
  \country{United Kingdom}}

\author{Leo White}
\orcid{0009-0003-7046-3035}
\email{lwhite@janestreet.com}
\affiliation{%
  \institution{Jane Street}
  \country{United Kingdom}
}

\author{Stephen Dolan}
\orcid{0000-0002-4609-9101}
\email{sdolan@janestreet.com}
\affiliation{%
  \institution{Jane Street}
  \country{United Kingdom}
}

\author{Daniel Hillerström}
\orcid{0000-0003-4730-9315}
\email{daniel.hillerstrom@ed.ac.uk}
\affiliation{%
  \institution{The University of Edinburgh}
  \country{United Kingdom}}

\author{Sam Lindley}
\orcid{0000-0002-1360-4714}
\email{sam.lindley@ed.ac.uk}
\affiliation{%
  \institution{The University of Edinburgh}
  \country{United Kingdom}}

\author{Anton Lorenzen}
\orcid{0000-0003-3538-9688}
\email{anton.lorenzen@ed.ac.uk}
\affiliation{%
  \institution{The University of Edinburgh}
  \country{United Kingdom}}

\begin{abstract}
  Effect handlers are a powerful abstraction for defining,
  customising, and composing computational effects.
  Statically ensuring that all effect operations are handled requires
  some form of effect system, but using a traditional effect system
  would require adding extensive effect annotations to the millions of
  lines of existing code in these languages.
  Recent proposals seek to address this problem by removing the need
  for explicit effect polymorphism. However, they typically rely on
  fragile syntactic mechanisms or on introducing a separate notion of
  second-class function.
  We introduce a novel approach based on modal effect types.

\end{abstract}




\maketitle

\section{Introduction}
\label{sec:introduction}

Effect handlers~\cite{PlotkinP13} allow programmers to define,
customise, and compose a range of computational effects including
concurrency, exceptions, state, backtracking, and probability, in
direct-style inside the programming language.
Following their pioneering use in languages such
as \Eff~\citep{BauerP15},
\Effekt~\citep{BrachthauserSO20,BrachthauserSLB22},
\Frank~\citep{frank,ConventLMM20}, \Koka~\citep{koka}, and
\Links~\citep{linksrow}, they are now increasingly being adopted in
production languages and systems such as OCaml~\citep{multicore},
Scala~\citep{BoruchGruszeckiOLLB23}, and
WebAssembly~\citep{Phipps-CostinRGLHSPL23}.

In a statically typed programming language with effect handlers some
form of effect system to track effectful operations is necessary in
order to ensure that a given program handles all of its effects.
However, traditional effect systems require extensive effect
annotations even for code that does not use effects. Consider the
standard \lstinline{map} function:
\begin{lstlisting}
  map : forall a b . (a -> b) -> List a -> List b
\end{lstlisting}
This type is a statement about the values that map accepts and
returns, but is silent about which effects may occur during its
evaluation. In the effect system of \Koka, for instance, this
\lstinline{map} function is thus presumed to be a total function that
takes a function which cannot perform any effects and itself does not
perform any effects.

However, this would prevent programmers from passing any effectful
function to \lstinline{map}. To use \lstinline{map} in effectful code,
in \Koka we must give it a more permissive type such as:
\begin{lstlisting}
  map' : forall a b e . (a earr b) earr List a earr List b
\end{lstlisting}
This type uses \emph{effect polymorphism}, quantifying over
an \emph{effect variable} \lstinline{e}, which occurs on every arrow.
Such effect annotations pollute the type signature of \lstinline{map}
to convey the obvious: \lstinline{map'} is polymorphic in its effects,
that is, the effects of
\lstinline{map' f xs}
depend on the effects of the function argument \lstinline{f}.
Effect annotations impose a mild burden to authors of new code, but
pose a significant problem when extending an existing language with
effectful features.

Type signatures of existing library code must be rewritten to support
effect polymorphism~\citep{BoruchGruszeckiOLLB23,NewGL23}, even in
legacy libraries that do not use effects, making it challenging to
retrofit such an effect system onto an existing language in a
backwards-compatible way without causing friction for existing
codebases.
However, if we can eliminate the need to annotate effect polymorphism,
then retrofitting an effect system ought to become a tractable
problem. Our goal is to design a principled effect system, where
effect polymorphism silence is a virtue.

An important step towards that goal was taken by the
language \Frank~\citep{frank,ConventLMM20}.
\Frank gives \lstinline{map} its original unannotated type, whilst
still allowing it to be passed effectful functions. The key idea is
that expressions are typed assuming an unknown set of possible
effects---the \emph{ambient effects}---will be provided by the context
in which the expression occurs. Rather than assuming unannotated
function types perform no effects, they are assumed to perform the
ambient effects.

\Frank still uses effect variables behind the scenes, implicitly
inserting effect variables for passing the ambient effects around.
For instance, \Frank simply treats the type
signature of \lstinline{map} above as syntactic sugar for
\lstinline{map'}
(\Frank has certain other syntactic idiosyncrasies, so in order to
ease readability, we render \Frank code in similar syntax to \Metl, which
we introduce in Section~\ref{sec:overview}).
This syntactic mechanism is fragile. For instance, effect variables
can appear in error messages as in the following example in which we
use a \lstinline{yield} effect to write a function that yields all
values in a list.
\begin{lstlisting}
  gen : List Int <@\larr{\lop{yield}}@> Unit
  gen xs = map (fun x -> do yield x) xs; ()
\end{lstlisting}
If the user forgets the \lstinline{yield} annotation, Frank
complains:
\begin{lstlisting}
  cannot unify effects e and yield, <@£@>
\end{lstlisting}
Here \code{£} and \lstinline{e} are the underlying effect variables
inserted by the Frank compiler. They do not appear in the source program
and in larger programs it can be unclear how to fix such errors.

\Effekt~\citep{BrachthauserSO20} and
\Scala~\citep{BoruchGruszeckiOLLB23} also make use of ambient effects
to avoid effect polymorphism by tracking effects as capabilities.
However, they either restrict functions to be second-class or require
having capability variables in types for certain use cases, as we
discuss further in \Cref{sec:capability-based-effect-systems}.

We build on the insight that ambient effect contexts can substantially
reduce the annotation burden.
Instead of relying on desugaring to traditional effect polymorphism
like \Frank, we develop \Met (Modal Effect Types), a novel effect
system with a theoretical foundation based on modal types.
We follow multimodal type theory (MTT)~\citep{Gratzer23,GratzerKNB20}
in tracking \emph{modes} for types and terms and
consider \emph{modalities} as the transitions between such modes.
We treat each possible ambient effect context as a mode, and each
possible transition between effect contexts as a modality.
Our system provides two kinds of modalities: 1) \emph{absolute}
modalities, which override the ambient effect context; and 2)
\emph{relative} modalities, which describe a local change to the
ambient effect context, as exemplified by effect handlers which
handles certain effects and forwards others unchanged.

\Met precludes hidden effect variables in
error messages as there are no hidden effect variables.
Moreover, \Met works smoothly with pure first-class higher-order
functions, which require neither hidden effect variables nor extra
annotations, and can be applied to effectful arguments.
Both \Frank and \Met strive to capture the essence of modular
programming with effects.
\Frank relies on a fragile syntactic characterisation based on
polymorphic types.
In contrast, \Met provides a more robust characterisation based on
simple types and modal types.

The main contributions of this paper are as follows.
\begin{itemize}
  \item We give a high-level overview of the key ideas of modal effect
  types: effect contexts, absolute and relative modalities.
  We provide a series of practical examples to show how modal effect
  types enable us to write modular effectful programs without effect
  polymorphism %
  (\Cref{sec:overview}).
  \item We briefly recall the design of multimodal type theory (MTT),
  the basis of modal effect types, and outline why MTT works well for
  designing an effect system (\Cref{sec:lock-and-key}).
  \item We introduce \Met, a simply-typed core calculus with effect
  handlers and modal effect types (\Cref{sec:core-calculus}). We prove
  its type soundness and effect safety.
  \item Intuitively, \Met can type check all functions that can be
  written in traditional effect systems using a single effect
  variable, which is the most common case in practice. We formally
  prove this intuition by presenting a calculus for row-based effect
  systems with a single effect variable and encoding it in \Met
  (\Cref{sec:encodings}).
  \item We extend \Met with data types and
  polymorphism for value types. To recover the full power of traditional effect
  systems, we also extend \Met with effect polymorphism which can be
  seamlessly used alongside modal effect types to express
  effectful programs that use higher-order effects modularly
  (\Cref{sec:extensions}).
  \item We outline and prototype a surface language \SurfaceMet which
  uses bidirectional type checking to infer the introduction
  and elimination of modalities (\Cref{sec:surface-lang}).
\item We present a direct comparison of type signatures in
  \SurfaceMet, \Koka, and \Effekt in order to highlight the practical
  appeal of modal effect types (\Cref{sec:comparison}).
\end{itemize}
\Cref{sec:related-work} also discusses related and future work.
The full specifications, proofs, and appendices can be found in the
extended version of the paper~\citep{TangWDHLL25}.

\FloatBarrier
\section{Programming with Modal Effect Types}
\label{sec:overview}

In this section we illustrate the main ideas of modal effect types
through a series of examples. We demonstrate how modal effect types
support modular composition of higher-order functions and effect
handlers without effect polymorphism.
The examples are written in \SurfaceMet, which translates to the core
calculus \Met via a simple type-directed elaboration.
In order to elucidate the core idea that modal effect types support
modular effectful programming without polymorphism we begin with
examples in the simply-typed fragment of \SurfaceMet.

\subsection{From Function Arrows to Effect Contexts}
\label{sec:effect-context}

Traditional effect systems annotate a function type with the effects
that the function may perform when invoked.
For instance, consider the following typing judgement for the
\lstinline{app} function specialised to take a pair of a function from
integers to the unit type and an integer.

\begin{lstlisting}
  |-  fun (f, x) -> f x  :  (Int <@\larr{E}@> Unit, Int) <@\larr{E}@> Unit
\end{lstlisting}
The effect annotation \lstinline{E} is a row of typed operations that
\lstinline{f} may perform.
For instance, if \lstinline{E} is
\lstinline{get : Unit >> Int, put : Int >> Unit} then \lstinline{f}
may perform a \lstinline{get} operation which takes a unit value and
returns an integer and a \lstinline{put} operation which takes an
integer and returns a unit value.
As in Frank~\citep{ConventLMM20} and Koka~\citep{koka}, rows are
\emph{scoped}~\citep{Leijen05} meaning that they allow duplicate
operations with the same name (but possibly different types).
The order of duplicates matters, but the relative order of distinct
operations does not.

Since \lstinline{app} invokes its argument, \lstinline{E} also denotes
the operations that invoking \lstinline{app} might perform.
As we saw in the introduction, the standard way to support modularity
is to be polymorphic in \lstinline{E}.
But this introduces an annotation burden for all higher-order
functions, including those (like \lstinline{app}) which do not
themselves perform effects.

In the spirit of \Frank, \Met decouples effects from function
types and tracks \emph{effect contexts} in typing judgements.
All components of the term and type share the same effect context
(unless manipulated by modalities as we will see in
\Cref{sec:absolute-modalities,sec:relative-modalities}).
For instance, we have the following typing judgement for the
same \lstinline{app} function as above.
\newsavebox\boxIntToUnit \savebox\boxIntToUnit{\lstinline{Int -> Unit}}
\newsavebox\boxToUnit \savebox\boxToUnit{\lstinline{ -> Unit}}
\newsavebox\boxAppBody \savebox\boxAppBody{\lstinline{ -> f x}}
\newsavebox\boxF \savebox\boxF{\lstinline{f}}
\newcommand{\underanno}[2]{$\gray{\underbrace{\usebox{#1}}_{\codeatmode{#2}}}$}
\begin{lstlisting}
  |-  fun (<@\underanno{\boxF}{E}@>, x) <@\underanno{\boxAppBody}{E}@>  :  (<@\underanno{\boxIntToUnit}{E}@>, Int) <@\underanno{\boxToUnit}{E}@>  @  E
\end{lstlisting}

As a visual aid, we use braces to explicitly annotate the effect
contexts for the argument \lstinline{f} and the whole function in the
term and type.
The $\codeatmode{E}$ annotation belongs to the judgement and
indicates the effect context \lstinline{E}.
This is the \emph{ambient effect context} for the whole term and type
of this typing judgement.
An effect context specifies which operations may be performed.
In this example, the effect contexts are all the same as the ambient
effect context.
We know that \lstinline{app} can perform the same effects as its
argument \lstinline{f} as they share the same effect context.

\subsection{Overriding the Ambient Effect Context with Absolute Modalities}
\label{sec:absolute-modalities}

An \emph{absolute modality} \lstinline{[E]} defines a new effect
context \lstinline{E} that overrides the ambient effect context.
For instance, the following function invokes the operation
\lstinline{yield} via the \lstinline{do} keyword.
The \lstinline{yield} operation takes an integer and returns a unit
value.
\newsavebox\boxDoYieldX \savebox\boxDoYieldX{\lstinline{do yield x}}
\newsavebox\boxFunDoYieldX \savebox\boxFunDoYieldX{\lstinline{fun x -> do yield x}}
\begin{lstlisting}
  |-  fun x -> <@\underanno{\boxDoYieldX}{\lop{yield} : \ltype{Int} $\sto$ \ltype{1}}@>  :  [yield : Int >> Unit](<@\underanno{\boxIntToUnit}{\lop{yield} : \ltype{Int} $\sto$ \ltype{1}}@>)  @  .
\end{lstlisting}
The absolute modality \lstinline{[yield : Int >> Unit]} specifies a
singleton effect context in which the \lstinline{yield} operation with
type \lstinline{Int >> Unit} may be performed.
Here, it overrides the empty ambient effect context (\lstinline{.}),
allowing \lstinline{yield} to be performed in the function body.

Effect contexts percolate through the structure of a type.
For example, a function of type \lstinline{[E](A -> B)} may perform
effects \lstinline{E} when invoked, and a list of type
\lstinline{[E](List (A -> B))} may perform effects \lstinline{E} when
its components are invoked.
For brevity, we define an effect context abbreviation.
\begin{lstlisting}
  eff Gen a = yield : a >> Unit
\end{lstlisting}
Such abbreviations are merely macros, such that,
for instance, \lstinline{[Gen Int]} denotes the modality
\lstinline{[yield : Int >> Unit]} and \lstinline{[Gen Int, E]} denotes
the modality
\lstinline{[yield : Int >> Unit, E]}.

For higher-order functions like \lstinline{map} and \lstinline{app}
which do not directly perform any effects, we use the empty absolute
modality \lstinline{[]}.
For instance, in \SurfaceMet, the curried first-class higher-order
\lstinline{iter} function, specialised to iterate over a list of
integers, is defined as follows.
\begin{lstlisting}
  iter : []((Int -> Unit) -> List Int -> Unit)
  iter f nil         = ()
  iter f (cons x xs) = f x; iter f xs
\end{lstlisting}
The empty absolute modality \lstinline{[]} specifies an empty effect
context in which the function is defined.
However, due to subeffecting, \lstinline{iter} is not limited to only
the empty effect context.
For instance, we can apply \lstinline{iter} to the previous function
which uses \lstinline{yield}.
\begin{lstlisting}
  |-  iter (fun x -> do yield x)  :  List Int -> Unit  @  Gen Int
\end{lstlisting}
\SurfaceMet allows us to use \lstinline{iter} here directly even
though its type contains an absolute modality. This is allowed since
an elaboration step implicitly eliminates the modality of
\lstinline{iter} before it is applied to the function that invokes
\lstinline{yield}.
Following the literature on modal types, we refer to introduction of
modalities as \emph{boxing} and elimination as \emph{unboxing}.
Though the modality \lstinline{[]} requires us to use \lstinline{iter}
under the empty effect context, subeffecting then upcasts it to the
singleton one \lstinline{Gen Int}.

To achieve the same flexibility of applying \lstinline{iter} to any
effectful arguments in a traditional row-based effect system, we would
need effect polymorphism:
\begin{lstlisting}
  iter : forall e . (Int earr Unit) earr List Int earr Unit
\end{lstlisting}

\subsection{Transforming the Ambient Effect Context with Relative Modalities}
\label{sec:relative-modalities}

So far we have only seen examples that are either pure or just perform
effects.
Absolute modalities suffice for modular programming with such examples
without requiring any use of effect polymorphism.
However, the situation becomes more interesting when we introduce
constructs that manipulate effect contexts non-trivially, such as
effect handlers.
Effect handlers provide a way of interpreting effects inside the
language itself.
For instance, we can use an effect handler to interpret operations
\lstinline{yield : Int >> Unit} by collecting their arguments into a
list of integers.
\begin{lstlisting}
  asList f = handle f () with
    return ()              ~> nil
    (yield : Int >> Unit) x r ~> cons x (r ())
\end{lstlisting}
The body of \lstinline$asList$ invokes the function \lstinline$f$
inside a \emph{handler}.
The handler has two clauses that account for two cases: 1) what
happens when \lstinline$f$ returns; and 2) what happens when
\lstinline$f$ performs \lstinline$yield$. In the first case, it
returns the list \lstinline$nil$. In the second case, it prepends the
integer \lstinline$x$ onto the head of the list returned by the
application of \lstinline$r$.
Here \lstinline$r$ is bound to the continuation of the
\lstinline$yield$ invocation inside \lstinline$f$.
The argument type of \lstinline$r$ is determined by the return type of
the operation being handled (unit in the case of \lstinline$yield$)
and its return type is determined by the return type of the
handler. Thus \lstinline$r : Unit -> List Int$.
When the continuation \lstinline{r} is invoked, the handler is
reinstalled around it to handle the residual effects of \lstinline{f}
(this kind of handler is known as \emph{deep} in the
literature~\citep{KammarLO13}).
We write \lstinline{H} for the handler clauses in \lstinline{asList}.

What type should \lstinline{asList} have?
Naively, we might simply expect to ignore the handler:
\begin{lstlisting}
  asList : []((Unit -> Unit) -> List Int)
\end{lstlisting}
This would be unsound as it would allow us to write:
\begin{lstlisting}
  crash : [Gen String](String -> List Int)
  crash s = asList (fun () -> do yield s)
\end{lstlisting}
The function passed to \lstinline{asList} yields a string. This is
then accidentally handled by the handler in \lstinline|asList|, which
expects an integer.

A possible fix is to box the argument of \lstinline{asList} with an
absolute modality \lstinline{[Gen Int]}:
\begin{lstlisting}
  asList : []([Gen Int](Unit -> Unit) -> List Int)
\end{lstlisting}
To see what happens here, consider the following typing judgement for
the inlined function body of \lstinline{asList} under some effect
context \lstinline{E}.
\newsavebox\boxM \savebox\boxM{\lstinline{f}}
\newsavebox\boxAppM \savebox\boxAppM{~\;\lstinline|f ()|~\;}
\newsavebox\boxUnitUnit \savebox\boxUnitUnit{\lstinline|Unit -> Unit|}
\begin{lstlisting}
  |-  fun <@\underanno{\boxM}{\leff{Gen} \ltype{Int}}@> -> handle<@\underanno{\boxAppM}{\leff{Gen} \ltype{Int}, E}@>with H  :  [Gen Int](<@\underanno{\boxUnitUnit}{\leff{Gen} \ltype{Int}}@>) -> List Int  @  E
\end{lstlisting}
The effect handler extends the ambient effect context \lstinline{E} with
a \lstinline{yield} operation to give an effect context of
\lstinline{Gen Int, E}. Meanwhile, the argument \lstinline{f} has
the effect context \lstinline{Gen Int} specified by the absolute
modality \lstinline{[Gen Int]}. This is sound, because it is safe
to invoke a function which can only use \lstinline{Gen Int} under the
effect context \lstinline{Gen Int, E}.

However, the restriction that the argument can only use
\lstinline{Gen Int} severely hinders reusability. We would like to apply
\lstinline{asList} to arguments that may perform other operations in
addition to \lstinline{yield}.
To this end, we introduce \emph{relative modalities} which enable us
to describe the relative change that a handler makes to the effect
context. For instance, consider:
\begin{lstlisting}
  asList : [](<Gen Int>(Unit -> Unit) -> List Int)
\end{lstlisting}
The relative modality \lstinline{<Gen Int>} is part of the argument
type and extends the ambient effect context with \lstinline{Gen Int} for
the inner function \lstinline{Unit -> Unit}.
The typing judgement becomes:
\begin{lstlisting}
  |-  fun <@\underanno{\boxM}{\leff{Gen} \ltype{Int}, E}@> -> handle<@\underanno{\boxAppM}{\leff{Gen} \ltype{Int}, E}@>with H  :  <Gen Int>(<@\underanno{\boxUnitUnit}{\leff{Gen} \ltype{Int}, E}@>) -> List Int  @  E
\end{lstlisting}
Now, the effect context for the function of argument \lstinline{f} is
also \lstinline{Gen Int, E}, matching the effect context at its
invocation.
This allows the argument \lstinline$f$ to perform other effects from
the ambient effect context \lstinline{E} (which will be forwarded to
outer handlers).

In practice relative modalities often appear in an argument position
and specify which effects of an argument will be handled in the
function body.
A higher-order function that handles effects \lstinline{D} of its
argument typically has a type of the form %
\lstinline{<D>(1 -> A) -> B}.

In a traditional row-based effect system, in order to be able to use
\lstinline{asList} across different effect contexts, we would
typically require effect polymorphism~\cite{HillerstromL16,koka}.
\begin{lstlisting}
  asList : forall e . (Unit earrg Unit) earr List Int
\end{lstlisting}

\subsection{Coercions Between Modalities}
\label{sec:overview-effect-coercions}

The implicit unboxing and boxing performed by \SurfaceMet allows
values to be coerced between different modalities. For instance, we
can extend an absolute modality:
\begin{lstlisting}
  |- fun f -> f : [Gen Int](<@\underanno{\boxUnitUnit}{\leff{Gen} \ltype{Int}}@>) -> [Gen Int, Gen String](<@\underanno{\boxUnitUnit}{\leff{Gen} \ltype{Int}, \leff{Gen} \ltype{String}}@>) @ E
\end{lstlisting}
Not all modalities can be coerced to one another. For example, we
cannot extend a relative modality
\begin{lstlisting}
  |/- fun f -> f : <>(<@\underanno{\boxUnitUnit}{E}@>) -> <Gen Int>(<@\underanno{\boxUnitUnit}{\leff{Gen} \ltype{Int}, E}@>) @ E  # Ill-typed
\end{lstlisting}
as doing so would insert a fresh \lstinline{yield : Int >> Unit}
operation which may shadow other \lstinline{yield} operations in
\lstinline{E}, consequently permitting bad programs like
\lstinline{crash} in \Cref{sec:relative-modalities}.

An absolute modality can be coerced into the corresponding relative
modality:
\begin{lstlisting}
  |- fun f -> f : [Gen Int](<@\underanno{\boxUnitUnit}{\leff{Gen} \ltype{Int}}@>) -> <Gen Int>(<@\underanno{\boxUnitUnit}{\leff{Gen} \ltype{Int}, E}@>) @ E
\end{lstlisting}
However, the converse is not permitted
\begin{lstlisting}
  |/- fun f -> f : <Gen Int>(<@\underanno{\boxUnitUnit}{\leff{Gen} \ltype{Int}, E}@>) -> [Gen Int](<@\underanno{\boxUnitUnit}{\leff{Gen} \ltype{Int}}@>) @ E  # Ill-typed
\end{lstlisting}
because the argument \lstinline{f} may also use effects from the
ambient effect context \lstinline{E}.

Similarly, the following typing judgement is invalid
\begin{lstlisting}
  |/- fun <@\underanno{\boxM}{\leff{Gen} \ltype{Int}, E}@> -> <@\underanno{\boxAppM}{E}@> : <Gen Int>(<@\underanno{\boxUnitUnit}{\leff{Gen} \ltype{Int}, E}@>) -> Unit @ E  # Ill-typed
\end{lstlisting}
because the argument \lstinline{f} may use \lstinline{Gen Int} in
addition to the ambient effect context \lstinline{E}.

\subsection{Composing Handlers}
\label{sec:composing-handlers}

We can compose handlers modularly.
For example, consider state operations \lstinline{get} and
\lstinline{put}.
\begin{lstlisting}
  eff State s = get : Unit >> s, put : s >> Unit
\end{lstlisting}

Specialising the state for integers, we can give the standard
state-passing interpretation of %
\lstinline{State Int} as follows~\cite{PlotkinP13}.

\begin{lstlisting}
  state : [](<State Int>(Unit -> Unit) -> Int -> Unit)
  state m = handle m () with
    return x               ~> fun s -> x
    (get : Unit >> Int) () r  ~> fun s -> r s s
    (put : Int >> Unit) s' r  ~> fun s -> r () s'
\end{lstlisting}

Using integer state we can write a generator which yields the prefix
sum of a list.
\begin{lstlisting}
  prefixSum : [Gen Int, State Int](List Int -> Unit)
  prefixSum xs = iter (fun x -> do put (do get () + x); do yield (do get ())) xs
\end{lstlisting}
The absolute modality \lstinline{[Gen Int, State Int]}
aggregates all effects performed in %
\lstinline{prefixSum}.

We can now handle \lstinline{prefixSum} by composing two handlers in
sequence.
\begin{lstlisting}
> asList (fun () -> state (fun () -> prefixSum [3,1,4,1,5,9]) 0)
# [3,4,8,9,14,23] : List Int
\end{lstlisting}
The type signature of \lstinline{state} mentions only
\lstinline{State Int} even though it is applied to a computation which
invokes \lstinline{prefixSum}, which also uses \lstinline{Gen Int}.
In contrast, to achieve the same modularity, conventional row-based
effect systems would ascribe the following type to
\lstinline{state}:
\begin{lstlisting}
  state : forall e . (Unit <@$\xrightarrow{\texttt{\leff{State} \ltype{Int}, e}}$@> Unit) earr Int earr Unit
\end{lstlisting}

\subsection{Storing Effectful Functions in Data Types}
\label{sec:unix-fork}

We show how modal effect types allow us to smoothly store effectful
functions into data types. We consider a richer effect handler example
that implements cooperative concurrency using a UNIX-style fork
operation~\cite{RitchieT74,Hillerstrom22}.
A \lstinline{Coop} effect context includes two operations.
\begin{lstlisting}
  eff Coop = ufork : Unit >> Bool, suspend : Unit >> Unit
\end{lstlisting}
The \lstinline{ufork} operation returns a boolean.
As we shall see, concurrency can be implemented by a handler that
invokes the continuation twice.
The idea is that passing true to the continuation defines the
behaviour of the parent, whereas passing false defines the behaviour
of the child.
The \lstinline{suspend} operation suspends the current process
allowing another process to run.

We model a process as a data type that embeds a continuation function
which takes a list of suspended processes and returns unit. In
addition, we define auxiliary functions \lstinline$push$ to append a
process onto the end of the list and \lstinline$next$ to remove and
then run the process at the head of the list.

\vspace{-.6\baselineskip}
\begin{minipage}[t]{0.55\textwidth}
\begin{lstlisting}
  data Proc = proc (List Proc -> 1)

  push : [](Proc -> List Proc -> List Proc)
  push x xs = xs ++ cons x nil
\end{lstlisting}
\end{minipage}
\begin{minipage}[t]{0.45\textwidth}
\begin{lstlisting}
  next : [](List Proc -> 1)
  next ps = case ps of
    nil               -> ()
    cons (proc p) ps' -> p ps'
\end{lstlisting}
\end{minipage}

The following handler implements a scheduler parameterised by a list
of suspended processes.
\begin{lstlisting}
  schedule : [](<Coop>(Unit -> Unit) -> List Proc -> Unit)
  schedule m = handle m () with
    return ()              ~> fun ps -> next ps
    (suspend : Unit>>Unit)  () r ~> fun ps -> next (push (proc (fun ps' -> r () ps')) ps)
    (ufork : Unit>>Bool) () r ~> fun ps -> r true (push (proc (fun ps' -> r false ps')) ps)
\end{lstlisting}
The \lstinline$return$-case is triggered when a process finishes, and
runs the next available process.
The \lstinline$suspend$-case pushes the continuation onto the end of
the list, before running the next available process.
The \lstinline$ufork$-case implements the process duplication
behaviour of UNIX fork by first pushing one application of the
continuation onto the end of the list, and then immediately applying
the other.
Observe that the above code seamlessly stores continuation functions
in \lstinline{Proc} and then puts \lstinline{Proc} in \lstinline{List}
without even mentioning any effects. These functions are not
restricted to be pure; they may use any effects from the ambient
effect context.

The \lstinline{schedule} function allows processes to use any other
effects.
To achieve this flexibility, a traditional row-based effect system
requires effect polymorphism and a parameterised data type.
\begin{lstlisting}
  data Proc e = proc (List Proc <@\larr{e}@> Unit)
  schedule : forall e . (1 <@\larr{\leff{Coop}, e}@> 1) <@\larr{e}@> List (Proc e) <@\larr{e}@> 1
\end{lstlisting}

\subsection{Masking}
\label{sec:overview-mask}

Whereas handlers extend the effect context, masking restricts the
effect context~\citep{BiernackiPPS18}. Masking is a useful device to
conceal private implementation details~\cite{LucassenG88}.
We illustrate masking by using a generator to implement a function to
find an integer satisfying a predicate.
\begin{lstlisting}
  findWrong : []((Int -> Bool) -> List Int -> Maybe Int)  # ill-typed
  findWrong p xs = handle (iter (fun x -> if p x then do yield x) xs) with
    return _               ~> nothing
    (yield : Int >> Unit) x _ ~> just x
\end{lstlisting}
The \lstinline{findWrong} program is ill-typed because it is unsound
to invoke predicate \lstinline{p} inside the handler, as this would
accidentally handle any \lstinline{yield} operations performed by
\lstinline{p}.
\newsavebox\boxPX \savebox\boxPX{\lstinline{(p x)}}
\begin{lstlisting}
  |- ... handle (iter (fun x -> if<@\underanno{\boxPX}{Gen Int, E}@>then do yield x) xs) with ...  :  _  @  E
\end{lstlisting}
Changing the type of \lstinline{p} from \lstinline{Int -> Bool} to
\lstinline{<Gen Int>(Int -> Bool)} would fix the type error but leak
the implementation detail that \lstinline{findWrong} uses
\lstinline{yield}.
A better solution is to mask \lstinline{yield} for the argument
\lstinline{p} and rewrite the handled expression as follows.
\begin{lstlisting}
  |- ... handle (iter (fun x -> if mask<yield><@\underanno{\boxPX}{E}@> ... ) with ...  :  _  @  E
\end{lstlisting}
The term \lstinline{mask<yield>(M)} masks the operation
\lstinline{yield} from the ambient context of the subterm
\lstinline{M}. In doing so, it conceals any \lstinline{yield}
invocations of \lstinline{M} from its immediate enclosing handler,
thus deferring handling of those operations to the second-nearest
dynamically enclosing handler.
Now, the effect context for \lstinline{p} is equivalent to the ambient
context \lstinline{E}, since the transformations of extending with \lstinline{yield}
(performed by the immediate enclosing handler) followed by masking
with \lstinline{yield} (performed by the mask) cancel each other out.
Like a handler, a mask wraps its return value in a relative modality.
The term \lstinline{mask<yield>(p x)} initially returns a value of
type \lstinline{<yield|>Bool} instead of \lstinline{Bool}, where
\lstinline{<yield|>} is a relative modality masking \lstinline{yield}
from the ambient context.
\SurfaceMet automatically unboxes the relative modality,
because pure values (e.g. booleans) are oblivious of the effect
context.

In general, relative modalities have the form \lstinline{<L|D>} which
specifies a local transformation on the effect context. Here,
\lstinline{L} represents the row of effect labels that are being
removed from the context, whilst \lstinline{D} represents the row of
effects being added to it.
We write \lstinline{<D>} as a shorthand for \lstinline{<|D>}.

\subsection{Kinds}
\label{sec:leaving-handlers}
\label{sec:kinds}

A handler extends the effect context with those effects it handles.
When a value leaves the scope of a handler, its effect context
changes, and we must keep track of this change.
This requires the introduction of a simple kinding system.

For example, let us consider \lstinline{state'}, a variation of the
\lstinline{state} function defined in \Cref{sec:composing-handlers} in
which the return type of the handled computation is changed from
\lstinline{Unit} to \lstinline{Unit -> Unit}.
The body of \lstinline{state'} is exactly the same as that of
\lstinline{state}.
We might naively expect its type signature to be the following.
\begin{lstlisting}
  state' : [](<State Int>(Unit -> (Unit -> Unit)) -> Int -> (Unit -> Unit))
\end{lstlisting}
However, this typing is not sound. Suppose we apply \lstinline{state'}
as follows.
\begin{lstlisting}
  state' (fun () -> fun () -> do put (do get () + 42)) 0  :  Unit -> Unit
\end{lstlisting}
The function \lstinline{fun () -> do put (do get () + 42)} is returned by
the return clause of \lstinline{state'}, escaping the scope of their
handler.
To guarantee effect safety, we must capture the fact that the returned
function is a thunk which might perform \lstinline{get} and
\lstinline{put} when invoked.
The following typing is sound.
\begin{lstlisting}
  state' : [](<State Int>(Unit -> (Unit -> Unit)) -> Int -> <State Int>(Unit -> Unit))
\end{lstlisting}

Let us contrast the types of \lstinline{state : [](<State Int>(Unit -> Unit) -> Int -> Unit}
and \lstinline{state'}.
The crucial difference is that the former cannot leak the state effect
as the handled computation has unit type, whereas the latter can as
the handled computation is a function.

In practice, it is useful to allow a value of base type or an
algebraic data type that contains only base types or types boxed with
absolute modalities to appear anywhere, including escaping the scope
of a handler.
Such values can never depend on the effect context in which they are
used.
We introduce a kind system in which the
\lstinline{Abs} kind classifies such \emph{absolute types}, whereas
the \lstinline{Any} kind classifies unrestricted types.
Subkinding allows absolute types to be treated as unrestricted.

\subsection{Polymorphism for Value Types}
\label{sec:poly-kinds}

Now that we have explored the simply-typed fragment of modal effect
types, we briefly outline its extension with polymorphism for value
types.
For simplicity, \SurfaceMet requires explicit type abstraction and
type application.
We write explicit type abstractions and applications using braces.
For instance, we can define the polymorphic iterate function as follows.
\begin{lstlisting}
  iter : forall a . []((a -> Unit) -> List a -> Unit)
  iter {a} f nil         = ()
  iter {a} f (cons x xs) = f x; iter {a} f xs
\end{lstlisting}

The extension is mostly routine, however, we must respect kinds. The
\lstinline{state} and \lstinline{state'} examples in \Cref{sec:kinds}
illustrate a non-uniformity that we must account for.
We may generalise them such that the former allows any absolute return
type and the latter allows any return type at all.
\begin{lstlisting}
  state  : forall [a] . [](<State Int>(Unit -> a) -> Int -> a)
  state' : forall  a  . [](<State Int>(Unit -> a) -> Int -> <State Int>a)
\end{lstlisting}
The syntax \lstinline{forall [a]} ascribes kind \lstinline{Abs} to
\lstinline{a}, allowing values of type \lstinline{a} to escape the
handler.
The syntax \lstinline{forall a} ascribes kind \lstinline{Any} to
\lstinline{a}, not allowing values of type \lstinline{a} to escape the
handler.
Though in practice it is usually desirable for return types of
computations inside handler scopes to be absolute. The latter type
signature is the more general in that simply by $\eta$-expanding we
can coerce it to the former.
\begin{lstlisting}
  |- fun {a} m s -> state' {a} m s : forall [a] . [](<State Int>(Unit -> a) -> Int -> a) @ .
\end{lstlisting}

\subsection{Effect Polymorphsim}
\label{sec:higher-order-fork}

Though modal effect types alone suffice for writing a remarkably rich
class of modular effectful programs, occasionally effect variables are
still useful.
In particular, they are required for the implementation of
higher-order effects~\citep{scope14,BauerP15,BergS23,YangW23}, which
take closures as arguments.

Modal effect types restrict the parameter type and result
type of an operation to be absolute.
This is because effect handlers provide non-trivial manipulation of
control-flow, which allows the parameter and result of an
operation to jump between different effect contexts.
For example, if we were to allow an operation
\lstinline{leak : (Unit->Unit) >> Unit},
then we could write the following unsafe program.
\begin{lstlisting}
  handle asList (fun () -> do leak (fun () -> do yield 42)) with
    return                 _   ~> fun () -> 37
    (leak : (Unit -> Unit) >> Unit) p _ ~> p
\end{lstlisting}
The \lstinline{asList} handler extends the ambient effect context with
\lstinline{yield}.
However, the \lstinline{leak} handler binds the closure
\lstinline{(fun () -> do yield 42)} to \lstinline{p} and returns this
closure, leaking the \lstinline{yield} operation.

Consider then a higher-order fork operation whose parameter is a thunk
(this operation is realisable with one-shot continuations; whereas the
fork operation of \Cref{sec:unix-fork} requires multi-shot
continuations).
We may define a recursive effect context for cooperative processes as
follows.
\begin{lstlisting}
  eff Coop = fork : [Coop](Unit -> Unit) >> Unit, suspend : Unit >> Unit
\end{lstlisting}
This is sound because the parameter of \lstinline{fork} is under an
absolute modality. However, this stops the forked process from using
any effect other than the two concurrency primitives \lstinline{fork}
and \lstinline{suspend}. To allow it to use other effects, we
re-introduce effect variables and polymorphism.
With an effect variable \lstinline{e}, we can define the following
higher-order \lstinline{Coop} parameterised over effect context
\lstinline{e}.
\begin{lstlisting}
  eff Coop e = fork : [Coop e, e](Unit -> Unit) >> Unit, suspend : Unit >> Unit
\end{lstlisting}
In \Cref{sec:meet-polymorphism} we show that the extension of effect
variables is sound and backward compatible.
Nonetheless, effect variables are only necessary for use-cases such as
higher-order effects in which a computation needs to be stored for use
in an effect context different from the ambient one.

\section{A Tale of Locks and Keys: Elaborating \SurfaceMet into \Met}
\label{sec:lock-and-key}

So far, we have presented a series of examples in \SurfaceMet
illustrating the core ideas of modal effect types.
While \SurfaceMet is an easy-to-use surface language, it contains
too many implicit coercions to be a suitable basis for a core calculus.
Instead, our core calculus \Met is a more explicit language based
on (simply-typed) multimodal type theory (MTT)~\citep{GratzerKNB20,Gratzer23}.
In this section, we motivate the design of \Met and introduce
core concepts of MTT.
For a more detailed account of the simply-typed fragment of MTT,
we refer the reader to the work of \citet{KavvosG23}.

MTT extends type theory with the notions of \emph{modes} and
\emph{modalities} and is parametric in them.
A typing judgement has the form $\typm{\Gamma}{M:A}{E}$,
which means that term $M$ has type $A$ under context $\Gamma$ at mode
$E$.
In this work we consider \emph{effect contexts} as modes, and use
\emph{absolute} modalities \lstinline{[E]} and \emph{relative}
modalities \lstinline{<L|D>} to move between them.
In MTT, most type constructors are mode-local: the components have the same
mode as the whole type.
For example, if a function type $A \to B$ is at mode $E$, then both
$A$ and $B$ are at mode $E$.
This is exactly the behaviour needed for effect contexts in \Met
(\Cref{sec:effect-context}), and is a primary motivation for treating
effect contexts as modes in \Met.

In MTT, modes can be transformed by modalities.
A modality $\mu : E \to F$ transforms types and terms from mode $E$ to
mode $F$.
While this is implicit in \SurfaceMet, \Met requires explicit terms
for modality introduction and elimination. We write $\Mod_\mu$ for
introducing the modality $\mu$.
For example, the function \lstinline{fun x -> do yield x} of type
\lstinline{[yield : Int >> Unit](Int -> Unit)} from
\Cref{sec:absolute-modalities} is elaborated to the following term in
\Met which explicitly introduces the absolute modality.
\[\bl
  \Mod_{\aeq{\Yield : \Int \sto \TUnit}}\;(\lambda x^{\Int} . \Do\Yield\;x)
\el\]

In order to invoke this function, we need to eliminate its modality first.
In \SurfaceMet, we need not manually unbox variables; elaboration does
so for us.
\Met adopts let-style unboxing, which requires binding a
term in order to unbox it.
For example, consider that we bind the above function to the variable
\lstinline{gen} and then apply it to $42$.
The binding and application are elaborated to \Met as follows.
\[\bl
  \Letm{}{\aeq{\Yield : \Int \sto \TUnit}} \code{gen} =
  \Mod_{\aeq{\Yield : \Int \sto \TUnit}}\;(\lambda x^{\Int} . \Do\Yield\;x)
\In
\code{gen}\;42
\el\]
The let-binding eliminates the absolute modality.
Whenever \lstinline{gen} is used, the type system ensures that the
effect context contains at least the operation ${\Yield : \Int \sto
\TUnit}$.
Both boxing and unboxing interact with the typing context non-trivially.
Their typing rules in MTT are as follows%
\footnote{As we will see in \Cref{sec:core-calculus}, in \Met
modalities on bindings and locks have indexes and the $\Letm{}{\mu}\!\!$
syntax also has an additional annotation. We opt for a simplified
version here to convey the core intuition.}.
\begin{mathpar}
\inferrule*
{
  \mu : E\to F \\
  \typm{\Gamma,\hl{\lockwith{\mu}}}{M:A}{E}
}
{\typm{\Gamma}{\Mod_\mu\;M : \boxwith{\mu}A}{F}}

\inferrule*
{
  \typm{\Gamma}{M : \boxwith{\mu} A}{E} \\
  \typm{\Gamma,x\hl{:_{\mu}}{A}}{N:B}{E}
}
{\typm{\Gamma}{\Letm{}{\mu} x = M \In N : B}{E}}
\end{mathpar}

Contexts $\Gamma$ are ordered.
Boxing is mediated by a locked context.
Unboxing binds a variable annotated with the corresponding modality.
These locks and annotations are crucial for controlling variable
access.
For instance, consider the following typing judgement
\[\bl
x:{\boxwith{\mu}} A \vdash
\Letm{}{\mu} x' = x \In \Mod_{\nu}\;x'
:
\boxwith{\nu}A
\atmode{F}
\el\]
which unboxes a variable of type $\boxwith{\mu}A$ and re-boxes it with
modality $\nu$.
If this term is typable, then its typing derivation would contain the
following judgement for variable $x'$:
\[\bl
x:{\boxwith{\mu}A},
x'\hl{:_{\mu}}A,
\hl{\lockwith{\nu}}
\vdash
x' : A \atmode{E}
\quad\text{where } \nu : E \to F
\el\]
Whether this usage of $x'$ is valid or not depends on the mode theory.
Beyond modes and modalities, a mode theory also specifies a set of
modality transformations $\alpha:\mu\To\nu$. (For those readers
familiar with category theory, the structure of modes, modalities, and
transformations forms a 2-category.)
We can use a variable $x':_\mu A$ across a lock $\lockwith{\nu}$ only
if there exists a transformation $\alpha:\mu\To\nu$.

In fact, the \Met kind system relaxes this constraint slightly by
allowing variables of absolute kind to cross locks even when such a
transformation does not exist.
Nonetheless, constraints on the form of modality transformations in
\Met (\Cref{sec:modalities}) are crucial for disallowing unsound
coercions, between certain relative modalities, for instance.
A case in point is the second example of
\Cref{sec:overview-effect-coercions}, where we saw that \Metl
disallows coercing a function of type
\lstinline{<>(Unit -> Unit)}
into a function of type \lstinline{<Gen Int>(Unit -> Unit)}.
This example is elaborated into \Met as follows:
\[\bl
\lambda f^{\boxwith{\aid{}}(\TUnit \to \TUnit)}. \Letm{}{\aid{}} \hat{f} = f \In \Mod_{\aex{\code{Gen Int}}}\;\hat{f} : \boxwith{\aex{\code{Gen Int}}}(\TUnit \to \TUnit) \atmode{E}
\el\]
First, $f$ is unboxed to obtain $\hat{f}\varb{\aid{}}{(\TUnit \to
  \TUnit)}$ in the context. Then, it is boxed again with the
$\aex{\code{Gen Int}}$ modality.  However, just as in \Metl, this
example is not well-typed in \Met. The reason for this is that the
term $\Mod_{\boxwith{\aex{\code{Gen Int}}}}$ introduces a lock
$\lockwith{\aex{\code{Gen Int}}}$ in the context. The variable
$\hat{f}$ can only be used under the lock if there is a modality
transformation of form $\aid{}\To\aex{\code{Gen Int}}$.  But as
explained in \Cref{sec:overview-effect-coercions}, such a
transformation would break type safety and is thus not permitted.

Locks are introduced in the context whenever a typing rule changes the effect context.
This happens not only during boxing but also in the rules for handlers and masks.
For example, the definition of \lstinline{asList : [](<Gen Int>(Unit -> Unit) -> List Int)}
from \Cref{sec:relative-modalities} is elaborated to
\[\bl
\Mod_{\aeq{}}\;(
\lambda {f}^{\boxwith{\aex{\code{Gen Int}}}(\TUnit \to \TUnit)} .  \Letm{}{\aex{\code{Gen Int}}} \hat{f} = f \In \Handle\; \hat{f}\;\Unit \With \hat{H})
\el\]
where $\hat{H}$ is the elaboration of handler clauses $H$.
The handler for the \lstinline{Gen} effect introduces
a lock $\lockwith{\aex{\code{Gen Int}}}$ in the context.
We can use the variable $\hat{f}$ under the lock if there is a modality transformation
$\alpha:\aex{\code{Gen Int}}\To\aex{\code{Gen Int}}$. This identity transformation always
exists.

\FloatBarrier
\section{A Multimodal Core Calculus with Effect Handlers}
\label{sec:core-calculus}

In this section we introduce \Met, a simply-typed call-by-value
calculus with effect handlers and modal effect types. We present its
static and dynamic semantics as well as its meta theory. We aim at a
minimal core calculus here and defer extensions such as data types,
alternative forms of handlers, and polymorphism (including both for
values and effects) to \Cref{sec:extensions}. %

\subsection{Syntax}
\label{sec:CalcM-syntax}

The syntax of \Met is as follows.

\vspace{-.5\baselineskip}
\noindent
  \begin{minipage}[t]{0.4\textwidth}
  \begin{syntax}
  \slab{Types}    &A,B  &::= & \TUnit \mid A\to B \mid \hl{\boxwith{\mu} A} \\
  \slab{Masks}          &L   &::= & \cdot \mid \ell,L \\
  \slab{Extensions}\hspace{-1em}      &D   &::= & \cdot \mid \ell : P, D \\
  \slab{Effect Contexts}\hspace{-1ex} &E,F &::= & \cdot \mid \ell : P, E \\
  \slab{Presence}\hspace{-1em}      &P   &::= & A\sto B \mid \hl{\Abs\vphantom{A}} \\
  \slab{Modalities}\hspace{-1em}      &\mu,\nu &::= & \hl{\aeq{E}}
                                    \mid \hl{\adj{L}{D}} \\
  \slab{Kinds}          &K &::= &  \hl{\Pure} \mid \Any \\ %
  \end{syntax}
  \end{minipage}
  \hspace{1.4em}
  \begin{minipage}[t]{0.5\textwidth}
    \vspace{-.2em}
  \begin{syntax}
  \slab{Contexts}\hspace{-3em}       &\Gamma &::=& \cdot \mid \hl{\Gamma, x\varb{\mu_F}{A}}
                                            \mid \hl{\Gamma,\lockwith{\mind{\mu}{F}}}
                                            \\
  \slab{Terms}\hspace{-.5em}   &M,N  &::= & \Unit \mid x \mid \lambda x^A.M \mid M\,N \mid \hl{\Mod_\mu\,V} \\
                        &     &\mid&  \hl{\Letm{\nu}{\mu} x = V\In M} \\
                        &     &\mid& \Do\ell\; M \mid \Mask_L\,M \\
                        &     &\mid& \Handle\;M\With H \\
  \slab{Values}\hspace{-.5em}  &V,W  &::= & \Unit \mid x \mid \lambda x^A.M \mid {\Mod_\mu\, V} \\
  \slab{Handlers}\hspace{-3em}       &H    &::= & \{ \Ret x \mapsto M \}
                              \mid  \{ \ell \; p \; r \mapsto M \} \uplus H
  \end{syntax}
  \end{minipage}
\vspace{.3\baselineskip}

\Met extends a simply-typed $\lambda$-calculus with standard
constructs for effects and handlers as well as the main novelty of
this work: modal effect types. We highlight the novel parts in
grey.

We have provided a brief introduction to MTT in \Cref{sec:lock-and-key}.
We present \Met without assuming deep familiarity with MTT and
discuss further in \Cref{sec:related-work-mtt}.
MTT provides us with the flexibility to define our own mode theory. In
the following, we first illustrate the structures of modes,
modalities, and modality transformations for \Met before presenting
the typing rules.

\subsection{Effect Contexts as Modes}
\label{sec:modes}

The \emph{modes} of \Met are effect contexts $E$.
Each type and term is at some effect context $E$, specifying the
available effects from the context.

Effect contexts $E$ are defined as scoped rows of effect
labels~\citep{Leijen05}.
Each label denotes an effectful operation.
An effect context may contain the same label multiple times.
Each label has a presence type $P$~\citep{remy1994type}.
A presence type $P$ can be an operation arrow of the form $A \sto B$,
which indicates that the operation takes an argument of type $A$ and
returns a value of type $B$, or absent $\Abs$, which indicates that
the operation of this label cannot be invoked.

Following \citet{remy1994type} and \citet{Leijen05}, we identify
effects up to reordering of distinct labels, and allow absent labels
to be freely added to or removed from the right of effect contexts.
For instance, $\ell:P,\ell':\Abs$ is equivalent to $\ell:P$.
We can think of an effect context as denoting a map from labels to
infinite sequences of presence types where a cofinite tail of each
sequence contains only $\Abs$.

Extensions $D$ and masks $L$ are used respectively to extend effect
contexts with more labels or removes some labels from them.
Extensions are like effect contexts except that we do not ignore
labels with absent types in their equivalence relation,
so $\ell:P,\ell':\Abs$ and $\ell:P$ are distinct.

We define a sub-effecting relation on effect contexts $E \subtype E'$
if we can replace the absent types in $E$ with proper operation arrows
to obtain $E'$. %
We also have a subtyping relation on extensions $D\subtype D'$. Like
sub-effecting on effect contexts, it requires $D$ and $D'$ to contain
the same row of labels, but it allows absent types in $D$ to be
replaced by concrete signatures in $D'$.
We give the full rules for type equivalence and sub-effecting in
\Cref{app:rules-meet}.

Masks $L$ are simply multisets of labels without presence types; we
only need labels when removing them from effect contexts.
We define three operations $D+E$, $E-L$, and $L\bowtie D$ as follows.
\[\ba{r@{~}c@{~}l}
D+E &=& D,E \\
\cdot - L &=& \cdot \\
(\ell:P,E) - L &=&
\begin{cases}
  E - L'         &\text{if } L\equiv \ell,L' \\
  \ell:P,(E-L)   &\text{otherwise}
\end{cases} \\
\ea
\quad
\ba{r@{~}c@{~}l}
L \bowtie \cdot &=& (L,\cdot) \\
L \bowtie (\ell:P,D) &=&
  \begin{cases}
    L'\bowtie D      &\text{if } L\equiv\ell,L' \\
    (L',(\ell:P,D'))  &\text{otherwise} \\
    \span\text{where } (L',D') = L\bowtie D \\
  \end{cases}
\ea\]
The operation $D+E$ extends $E$ with $D$.
The operation $E-L$ removes the labels in $L$ from $E$.
The operation $L\bowtie D = (L',D')$ gives the difference between $L$
and $D$. The $L'$ are those labels in $L$ not appearing in the domain
of $D$, and the $D'$ are those entries in $D$ with labels not in $L$.

\subsection{Modalities Manipulating Effect Contexts}
\label{sec:modalities}

Components of types and terms may have different effect contexts from
the ambient one.
We use \emph{modalities} to manipulate effect contexts.
For the modal type $\boxwith{\mu} A$, the effect context for $A$ is
derived from the ambient effect context manipulated by the modality
$\mu$ as follows.
\[
  \act{\aeq{E}}{F} = E \qquad\qquad\qquad \act{\adj{L}{D}}{F} = D+(F-L)
\]
The absolute modality $\aeq{E}$ completely replaces the effect context
$F$ with $E$, similar to effect annotations on function types in
traditional effect systems.
The relative modality $\adj{L}{D}$ is the key novelty of \Met.
It specifies a transformation on the input effect context. It masks
the labels $L$ in $F$ before extending the resulting context with $D$.
We call $\adj{}{}$ the identity modality and write $\one$ for it.
Modalities are monotone total functions on effect contexts.
If $E\subtype F$, we have $\act{\mu}{E}\subtype\act{\mu}{F}$.

We write $\mu_F$ for the pair of $\mu$ and $F$ where $F$ is the effect
context that $\mu$ acts on.
We refer to such a pair as a concrete modality.
We write $\mu_F : E\to F$ if $\act{\mu}{F} = E$.
The arrow goes from $E$ to $F$ instead of the other direction to be
consistent with MTT.
Note that our terminology here differs slightly from that of MTT
introduced in \Cref{sec:lock-and-key}.
Concrete modalities $\mu_F$ correspond to the notion of modalities in
MTT and our modalities $\mu$ are actually indexed families of
modalities in MTT.

\paragraph{Modality Composition}

We can compose the actions of modalities in the intuitive way.
\[\ba{rclcll}
\mu&\circ&{\aeq{E}} &=& \aeq{E}
\\
\aeq{E}&\circ&{\adj{L}{D}} &=& \aeq{D+(E-L)}
\\
\adj{L_1}{D_1}&\circ&\adj{L_2}{D_2} &=&
  \adj{L_1+L}{D_2+D}
  &\text{ where } (L,D) = L_2 \bowtie D_1
\ea\]
To keep close to MTT, our composition reads from left to right.
First, an absolute modality completely specifies the new effect
context, thus shadowing any other modality $\mu$.
Second, replacing the effect context with $E$ and then masking $L$ and
extending with $D$ is equivalent to just replacing with $D+(E-L)$.
Third, sequential masking and extending can be combined into one by
using $L_2\bowtie D_1$ to cancel the overlapping part of $L_2$ and
$D_1$.
For instance, we have $\aex{\Yield:\Int \sto \TUnit}\circ\amk{\Yield}
= \one$.

Composition is well-defined since
composing followed by applying is equivalent to sequentially applying
$\act{(\mu\circ\nu)}{E} = \act{\nu}{\act{\mu}{E}}$.
We also have associativity $(\mu\circ\nu)\circ\xi =
\mu\circ(\nu\circ\xi)$ and identity $\one$.
The definition of composition naturally generalises to concrete
modalities $\mu_F$.
We can compose $\mu_F : E\to F$ and $\nu_E : E'\to E$ to get
$\mu_F\circ\nu_E : E'\to F$ which is defined as $(\mu\circ\nu)_F$.

\paragraph{Modality Transformations}

Just as modalities allow us to manipulate effect contexts, we need
a transformation relation that tells us when we can change
modalities.

In \Met, there could only be at most one transformation between any
two modalities.
As a result, we do not need to give names to transformations.
We write $\mu_F \To \nu_{F}$ for a transformation between concrete
modalities $\mu_F : E \to F$ and $\nu_{F} : E' \to F$.
Intuitively, such a transformation indicates that under ambient effect
context $F$, the action of $\mu$ can be replaced by the action of
$\nu$.
This relation is used to control variable access as we have
demonstrated in \Cref{sec:lock-and-key}.
For instance, supposing we have a variable of type $\boxwith{\mu}
(\TUnit\to \TUnit)$ under ambient effect context $F$, we can rewrap it
to a function of type $\boxwith{\nu} (\TUnit\to \TUnit)$ if $\mu_F \To
\nu_F$.

Intuitively, $\mu_F \To \nu_F$ is safe when $\nu(F)$ is larger than
$\mu(F)$ so that we have not lost any operations. Moreover, subeffecting
should not break the safety guarantee of transformations. That is,
$\mu(F') \subtype \nu(F')$ should hold for any effect context $F'$ with
$F\subtype F'$.
We formally define $\mu_F\To\nu_F$ by the transitive closure of the
following four rules.
\begin{mathpar}
  \inferrule*[Lab=\mtylab{Abs}]
  {
    {\mu}_{F} : E' \to F \\\\
    E\subtype E' \\
  }
  {\aeq{E}_F\To {\mu}_{F}}
\ \ \
  \inferrule*[Lab=\mtylab{Upcast}]
  {
    D \subtype D'
  }
  {\adj{L}{D}_F \To \adj{L}{D'}_F}
\ \ \
  \inferrule*[Lab=\mtylab{Expand}]
  {
    (F-L) \equiv \ell:A\sto B,E
  }
  {\adj{L}{D}_F \Rightarrow \adj{\ell,L}{D,\ell:A\sto B}_F}
\ \ \
  \inferrule*[Lab=\mtylab{Shrink}]
  {
    (F-L) \equiv \ell:P,E
  }
  {\adj{\ell,L}{D,\ell:P}_F \Rightarrow \adj{L}{D}_F}
\end{mathpar}

\mtylab{Abs} allows us to transform an absolute modality to any other
modality as long as no effect leaks.
\mtylab{Upcast} allow us to upcast a label with an absent type in $D$
to an arbitrary presence type, since the corresponding operation is
unused.
Recall that the subtyping relation between extensions only upcasts
presence types.
\mtylab{Expand} allows us to simultaneously mask and extend some
present operations given that these operations exist in the ambient
effect context $F$.
\mtylab{Shrink} allows us to do the reverse for any operations
regardless of their presence.

The following lemma shows that the syntactic definition of
transformation matches our intuition.
The proof is in \Cref{app:lemmas-modes}.
\begin{restatable}[Semantics of modality transformation]{lemma}{semanticModTrans}
  \label{lemma:semantic-modtrans}
  We have $\mu_F\To\nu_F$ if and only if
  $\act{\mu}{F'}\subtype\act{\nu}{F'}$ for all $F'$ with $F\subtype F'$.
\end{restatable}

Let us give some examples. First, $\aeq{}_E \To \mu_E$ always
holds, consistent with the intuition that pure values can be used
anywhere safely.
Second, $\aex{\ell:\Abs}_E \To \aex{\ell:P}_E$ always holds.
Third, we have $\adj{\ell}{\ell:A\sto B}_{\ell:A\sto B,E} \Leftrightarrow
\one_{\ell:A\sto B,E}$ in both directions.
Last, $\one_E \To \aex{\ell:P}_E$ does not hold for any
$E$.

\subsection{Kinds and Contexts}
\label{sec:met-kinds-contexts}

\begin{figure}[htbp]
\raggedright
\boxed{\Gamma\vdash A : K\vphantom{\mu}}
\boxed{\Gamma\vdash P\vphantom{\mu}}
\boxed{\Gamma\vdash (\mu,A)\To\nu \atmode{F}}
\hfill
\begin{mathpar}
\inferrule*
{ }
{\Gamma \vdash \TUnit : \Pure}

\inferrule*
{
  \Gamma \vdash A : \Pure
}
{\Gamma \vdash A : \Any}

\inferrule*
{
  \Gamma \vdash \aeq{E} \\
  \Gamma \vdash A : \Any \\
}
{\Gamma \vdash \boxwith{\aeq{E}} A : \Pure}

\inferrule*
{
  \Gamma \vdash \adj{L}{D} \\
  \Gamma \vdash A : K \\
}
{\Gamma \vdash \boxwith{\adj{L}{D}} A : K}

\inferrule*
{
  \Gamma \vdash A : \Any \\\\
  \Gamma \vdash B : \Any
}
{\Gamma \vdash A\to B : \Any}

\inferrule*
{
  \Gamma\vdash A : \Pure \\\\
  \Gamma\vdash B : \Pure
}
{\Gamma \vdash A\sto B}

\inferrule*
{
  \Gamma \vdash A:\Pure
}
{\Gamma\vdash (\mu,A)\To\nu \atmode{F}}

\inferrule*
{
  \mu_F\To\nu_F
}
{\Gamma\vdash (\mu,A)\To\nu \atmode{F}}
\end{mathpar}

\raggedright
\boxed{\Gamma \atmode{E}}
\hfill
\vspace{-.5\baselineskip}
\begin{mathpar}
\inferrule*
{ }
{\cdot\atmode{E}}

\inferrule*
{
  \Gamma\atmode{F} \\
  \mu_F : E\to F \\
  \Gamma \vdash A : K \\
}
{\Gamma, x\varb{\mu_F}{A} \atmode{F}}

\inferrule*
{
  \Gamma\atmode{F} \\
  \mind{\mu}{F} : E\to F \\
}
{\Gamma, \lockwith{\mind{\mu}{F}} \atmode{E}}
\end{mathpar}
\caption{Representative kinding, well-formedness, and auxiliary rules for \Met.}
\label{fig:kinding-modalities}
\end{figure}

As illustrated in \Cref{sec:leaving-handlers}, we have two kinds
$\Pure$ and $\Any$.
The $\Pure$ kind is a sub-kind of the kind of all types $\Any$, and
denotes types of values that are guaranteed not to use operations from the ambient
effect context.
We show the kinding and well-formedness rules for types and presence
types in \Cref{fig:kinding-modalities}, relying on the well-formedness
of modalities $\Gamma\vdash\mu$ and effect contexts $\Gamma\vdash E$,
which is standard and defined in \Cref{app:rules-meet}.
Function arrows have kind $\Any$ due to the possibility of using
operations from the ambient effect context. A modal type
$\boxwith{\aeq{E}}A$ is absolute as it cannot depend on the ambient
effect context.
We restrict the kind of the argument and return value of effects to be
$\Pure$ in order to prevent effect leakage as discussed in
\Cref{sec:higher-order-fork}.

Contexts are ordered.
Each term variable binding $x\varb{\mu_F}{A}$ in contexts is tagged
with an concrete modality $\mu_F$. %
We omit this annotation when $\mu$ is identity.
Contexts contain locks $\lockwith{\mu_F}$ carrying concrete modalities
$\mu_F$. %
As shown in \Cref{sec:lock-and-key}, they track the introduction and
elimination of modalities and play an important role in controlling
variable access.
We omitted indexes of modalities on bindings and locks in
\Cref{sec:lock-and-key} for brevity; they are obvious from the
context.

We define the relation $\Gamma\atmode{E}$ that
context $\Gamma$ is well-formed at effect context $E$ in
\Cref{fig:kinding-modalities}.
For instance, given some modalities $\mu_F : E_1\to F, \nu_F : E_2\to
F,$ and $\xi_E : E_3 \to E$, the following context is well-formed at
effect context $E$.
Reading from left to right, the lock $\lockwith{\aeq{E}_F}$ switches
the effect context from $F$ to $E$ as $\aeq{E}(F) = E$.
\[
x\varb{\mu_F}{A_1},y\varb{\nu_F}{A_2},
{\lockwith{\aeq{E}_F}},
z\varb{\xi_E}{A_3}\atmode{E}
\]

Following MTT, we define $\locks{-}$ to compose all the modalities on
the locks in a context.
\[\bl
\locks{\cdot} = {\one} \qquad
\locks{\Gamma,\lockwith{\mind{\mu}{F}}} = \locks{\Gamma}\circ\mind{\mu}{F} \qquad
\locks{\Gamma,x\varb{\mu_F}{A}} = \locks{\Gamma} \\
\el\]

Following MTT, we identify contexts up to the following two equations.
\[\bl
\Gamma,\lockwith{{\one}_E}\atmode{E} = \Gamma\atmode{E} \qquad\qquad
\Gamma,\lockwith{{\mu}_F},\lockwith{{\nu}_{F'}} \atmode{E} = \Gamma,\lockwith{{\mu}_F\circ{\nu}_{F'}}\atmode{E} \\
\el\]

\subsection{Typing}
\label{sec:CalcM-typing}

The typing rules of \Met are shown in
\Cref{fig:typing-met}.
The typing judgement $\typm{\Gamma}{M:A}{E}$ means that the term $M$
has type $A$ under context $\Gamma$ and effect context $E$.
As usual, we require $\Gamma\atmode{E}$,
$\Gamma\vdash E$, $\Gamma\vdash A:K$ for some
$K$, and well-formedness for type annotations as well-formedness conditions.
We explain the interesting rules, which are highlighted in grey; the
other rules are standard.

\begin{figure}[htbp]

\raggedright
\boxed{\typm{\Gamma}{M : A}{E}}
\hfill

\begin{mathpar}
\hl{
\inferrule*[Lab=\tylab{Var}]
{
  {\nu}_{F}= \locks{\Gamma'} : E\to F \\\\
  \Gamma\vdash (\mu,A)\To\nu \atmode{F}
}
{\typm{\Gamma,x\varb{\mu_F}{A},\Gamma'}{x:A}{E}}
}

\hl{
\inferrule*[Lab=\tylab{Mod}]
{
  \mind{\mu}{F} : E \to F \\\\
  \typm{\Gamma,\lockwith{\mind{\mu}{F}}}{V:A}{E}
}
{\typm{\Gamma}{\Box_\mu\,V : \boxwith{\mu} A}{F}}
}

\hl{
\inferrule*[Lab=\tylab{Letmod}]
{
  \nu_F : E\to F \\
  \typm{\Gamma,\lockwith{\nu_F}}{V : \boxwith{\mu} A}{E} \\\\
  \typm{\Gamma,x\varb{\nu_F\circ\mu_E}{A}}{M:B}{F}
}
{\typm{\Gamma}{\Letm{\nu}{\mu} x = V \In M : B}{F}}
}
\\

\inferrule*[Lab=\tylab{Abs}]
{
  \typm{\Gamma, x : A}{M : B}{E}
}
{\typm{\Gamma}{\lambda x^A .M : A \to B}{E}}

\inferrule*[Lab=\tylab{App}]
{
  \typm{\Gamma}{M : A \to B}{E} \\\\
  \typm{\Gamma}{N : A}{E}
}
{\typm{\Gamma}{M\; N: B}{E}}

\inferrule*[Lab=\tylab{Do}]
{
  E = \ell:A\sto B,F \\\\
  \typm{\Gamma}{N : A}{E} \\
}
{\typm{\Gamma}{\Do \ell \; N : B}{E}}

\hl{
\inferrule*[Lab=\tylab{Mask}]
{
  \typm{\Gamma,\lockwith{\mind{\amk{L}}{F}}}{M: A}{F-L} \\
}
{\typm{\Gamma}{\Mask_{L}\; M: \boxwith{\amk{L}} A}{F}}
}

\hl{
\inferrule*[Lab=\tylab{Handler}]
{
  H = \{\Ret x \mapsto N\} \uplus \{ \ell_i\;p_i\;r_i \mapsto N_i \}_i \\\\
  \typm{\Gamma, \lockwith{\mind{\aex{D}}{F}}}{M : A}{D+F} \\
  \typm{\Gamma, x : \boxwith{\aex{D}} A}{N : B}{F} \\\\
  {D = \{\ell_i : A_i \sto B_i\}_i} \\
  [\typm{\Gamma, p_i : A_i, {r_i}: B_i \to B}{N_i : B}{F}]_i
}
{\typm{\Gamma}{\Handle\;M\With
  H : B}{F}}
}
\end{mathpar}

\caption{Typing rules for \Met.}
\label{fig:typing-met}
\end{figure}

\paragraph{Modality Introduction and Elimination}

Modalities are introduced by \tylab{Mod} and eliminated by
\tylab{Letmod}.
The term $\Box_\mu\,V$ introduces modality $\mu$ to the type of the
conclusion and lock $\lockwith{\mu_F}$ into the context of the
premise, and requires the value $V$ to be well-typed under the new
effect context $E$ manipulated by $\mu$.
The lock $\lockwith{\mu_F}$ tracks the change to the effect context.
Specialising the modality $\mu$ to either absolute or relative
modalities, we get the following two rules.

\noindent\vspace{-.5\baselineskip}
\begin{mathpar}
  \inferrule*%
  {
    \typm{\Gamma,\lockwith{{\aeq{E}}_{F}}}{V:A}{E}
  }
  {\typm{\Gamma}{\Box_{\aeq{E}}\,V : \boxwith{\aeq{E}} A}{F}}

  \inferrule*%
  {
    \typm{\Gamma,\lockwith{{\adj{L}{D}}_{F}}}{V:A}{D+(F-L)}
  }
  {\typm{\Gamma}{\Box_{\adj{L}{D}}\,V : \boxwith{\adj{L}{D}} A}{F}}
\end{mathpar}

Note that we use modality $\mu$ instead of concrete modality $\mu_F$ in
types and terms, because the index can always be inferred from the
effect context.
We restrict $\Box$ to values to avoid effect leakage~\citep{LorenzenWDEL24,Ahman23}.
Otherwise, a term such as $\Box_{\aex{\ell:P}}\,(\Do\ell\,V)$ would type
check under the empty effect context but get stuck due to the unhandled
operation $\ell$.

The term $\Letm{\nu}{\mu} x = V\In M$ moves the modality $\mu$ from
the type of $V$ to the binding of $x$.
As with boxing, unboxing is restricted to values.
Following MTT, we use let-style modality elimination
which takes another modality $\nu$ in addition to the modality $\mu$
that is eliminated from $V$.
This is crucial for sequential unboxing.
For instance, the following term sequentially unboxes $x :
\boxwith{\nu}\boxwith{\mu} A$. The variables $y$ and $z$ are bound as
$y\varb{\nu}{\boxwith{\mu} A}$ and $z\varb{\nu\circ\mu}{A}$,
respectively.
\[
\Letm{}{\nu} y = x \In
\Letm{\nu}{\mu} z = y \In M
\]

\paragraph{Masking and Handling}

Masking and handling also introduce relative modalities.
Unlike $\Mod$, these constructs can apply to computations as they
perform masking and handling semantically.
In \tylab{Mask}, the mask $\Mask_L\;M$ removes effects $L$ from the
ambient effect context for $M$.
For the return value of $M$, we need to box it with $\amk{L}$ to
reconcile the mismatch between $F-L$ and $F$.
In \tylab{Handler}, the handler $\Handle\; M\With H$ extends the ambient
effect context with effects $D$ for $M$.
For the return value of $M$ which is bound as $x$ in the return
clause, we need to box it with $\aex{D}$ to reconcile the mismatch
between $D+F$ and $F$.
The other parts of the handler rule are standard.

\paragraph{Accessing Variables}

The \tylab{Var} rule uses the auxiliary judgement
$\Gamma\vdash(\mu,A)\To\nu\atmode{F}$ defined in
\Cref{fig:kinding-modalities}.
Variables of absolute types can always be used as they do not depend
on the effect context.
For a non-absolute term variable binding $x\varb{\mu_F}{A}$ from
context $\Gamma,x\varb{\mu_F}{A},\Gamma'$, we must guarantee that it
is safe to use $x$ in the current effect context.
The term bound to $x$ is defined inside $\mu$ under the effect context
$F$.
As we track all transformations on effect contexts up to the binding
of $x$ as locks in $\Gamma'$, the current effect context $E$ is
obtained by applying $\locks{\Gamma'}$ to $F$.
Thus, we need the transformation $\mu_F \To \locks{\Gamma'}_F$ to hold
for effect safety.

\paragraph{Subeffecting}

Subeffecting is incorporated into the \tylab{Var} rule within the
transformation relation $\mu_F\To\nu_F$.
We have seen how subeffecting works in
\Cref{sec:overview-effect-coercions}.
We give another example here which upcasts the empty effect context to
$E$.
It is well-typed because $\aeq{}_\cdot \To \aeq{E}_\cdot$ holds.
\[
\lambda x^{\boxwith{\aeq{}}(\Int \to \Int)} .
  \Letm{}{\aeq{}} y = x \In
  \Box_{\aeq{E}}\, y
: {\boxwith{\aeq{}}(\Int \to
\Int)} \to {\boxwith{\aeq{E}}(\Int \to \Int)}
\]

\subsection{Operational Semantics}
\label{sec:semantics}

The operational semantics for \Met is quite
standard~\citep{HillerstromLA20}.
We first define evaluation contexts $\EC$:
\begin{syntax}

  \slab{Evaluation contexts} &  \EC &::= & [~]
    \mid \EC\; N \mid V\;\EC \mid
    \Do\ell\;\EC \mid \Mask_L\;\EC \mid \Handle\; \EC \With H \\
  \end{syntax}

\noindent The reduction rules are as follows.
  \begin{reductions}
  \semlab{App}   & (\lambda x^A.M)\,V &\reducesto& M[V/x] \\
  \semlab{Letmod}  & \Letm{\nu}{\mu} x = \Box_\mu\,V \In M &\reducesto& M[V/x] \\
  \semlab{Mask} & \Mask_L\, V &\reducesto& \Box_{\amk{L}}\, V \\
  \semlab{Ret} &
    \Handle\; V \With H &\reducesto& N[(\Box_{\aex{D}}\,V)/x], \hfill
    \text{where } (\Ret x \mapsto N) \in H
  \\
  \semlab{Op} &
    \Handle\; \EC[\Do\ell \; V] \With H
      &\reducesto& N[V/p, (\lambda y.\Handle\; \EC[y] \With H)/r],\\
  \multicolumn{4}{@{}r@{}}{
        \text{ where } \free{0}{\ell}{\EC} \text{ and } (\ell \; p \; r \mapsto N) \in H
  } \\
  \semlab{Lift} &
    \EC[M] &\reducesto& \EC[N],  \hfill\text{if } M \reducesto N \\
  \end{reductions}
The only slightly non-standard aspect of the rules is the boxing of
values escaping masks and handlers in \semlab{Mask} and \semlab{Ret}.
They coincide with the typing rules for masks and handlers.
In \semlab{Ret}, we assume handlers are decorated with the operations
$D$ that they handle as in \Cref{sec:overview}.

Following \citet{BiernackiPPS18}, the predicate $\free{n}{\ell}{\EC}$
is defined inductively on evaluation contexts as follows.
We have $\free{n}{\ell}{\EC}$ if there are $n$ masks for
$\ell$ in $\EC$ without corresponding handlers.
The \semlab{Op} rule requires $\free{0}{\ell}{\EC}$ to guarantee that
the current handler is not masked by masks in $\EC$.
The meta function $\meta{count}(\ell;L)$ yields the number of $\ell$
labels in $L$.
We omit the inductive cases that do not change $n$.

\vspace{-.5\baselineskip}
\begin{mathpar}
  \inferrule*
  { }
  {\free{0}{\ell}{[~]}}

  \inferrule*
  {\free{n}{\ell}{\EC} }
  {\free{n}{\ell}{\Do\ell'\,\EC}}

  \inferrule*
  {\free{n}{\ell}{\EC} \\ \meta{count}(\ell;L) = m}
  {\free{(n+m)}{\ell}{\Mask_L\,\EC}}
  \\

  \inferrule*
  {\free{(n+1)}{\ell}{\EC} \\ \ell \in \dom{H}}
  {\free{n}{\ell}{\Handle\;\EC\With H}}

  \inferrule*
  {\free{n}{\ell}{\EC} \\ \ell \notin \dom{H}}
  {\free{n}{\ell}{\Handle\;\EC\With H}}
\end{mathpar}

\subsection{Type Soundness and Effect Safety}
\label{sec:metatheory}

We prove type soundness and effect safety for \Met. Our proofs cover
the extensions in \Cref{sec:extensions}.

\Met enjoys substitution properties along the lines of
\citet{KavvosG23}.
For example, we have the following rule for substituting values with
modalities into terms.
\begin{mathpar}
  \inferrule*
  {
    \typm{\Gamma,\lockwith{\mu_F}}{V:A}{F'} \\
    \typm{\Gamma,x\varb{\mu_F}{A},\Gamma'}{M:B}{E} \\
  }
  {\typm{\Gamma,\Gamma'}{M[V/x]:B}{E}}
\end{mathpar}
We state and prove the relevant properties in \Cref{app:CalcM-lemmas}.

To state syntactic type soundness, we first define normal forms.

\begin{definition}[Normal Forms]
  \label{def:normal-forms}
  We say a term $M$ is in a normal form with respect to effect type
  $E$, if it is either
  a value $V$
  or of the form $M =
  \EC[\Do\ell\;V]$ for $\ell\in E$ and $\free{n}{\ell}{\EC}$.
\end{definition}

The following together give type soundness and effect safety (proofs
in \Cref{app:progress,app:subject-reduction}).
\begin{restatable}[Progress]{theorem}{progress}
  \label{thm:progress}
  If $\typm{\,}{M:A}{E}$, then either there exists $N$ such that
  $M\reducesto N$ or $M$ is in a normal form with respect to $E$.
\end{restatable}
\begin{restatable}[Subject Reduction]{theorem}{subjectReduction}
  \label{thm:subjrect-reduction}
  If\, $\typm{\Gamma}{M:A}{E}$ and $M \reducesto N$, then
  $\typm{\Gamma}{N:A}{E}$.
\end{restatable}

\section{Encoding Effect Polymorphism in \Met}
\label{sec:encodings}

Even without effect variables, \Met is sufficiently expressive to
encode programs from conventional row-based effect systems provided
effect variables on function arrows always refer to the lexically
closest one.
This is an important special case, since most functions in practice
use at most one effect variable.
For example, as of July 2024, the \Koka repository contains 520
effectful functions across 112 files but only 86 functions across 5
files use more than one effect variable, almost all of them internal
primitives for handlers not exposed to programmers.
Moreover, almost all programs in the \Frank repository make no mention
of effect variables, relying on syntactic sugar to hide the single
effect variable.
We formally characterise and prove this intuition on the
expressiveness of \Met.

\subsection{Row Effect Types with a Single Effect Variable}
\label{sec:Leff}

We first define \Leff, a core calculus with row-based effect types in
the style of Koka~\citep{koka}, but where each scope can only refer to
the lexically closest effect variable.

\vspace{-.5\baselineskip}

\noindent\hspace{-1em}
\begin{minipage}[t]{0.5\textwidth}
\begin{syntax}
  \slab{Types}\hspace{-2.5em} &A,B  &::= & \TUnit \mid \earr{A}{B}{\geffrow{E}{\varepsilon}} \mid \forall \gray{\varepsilon}.A \\
  \slab{Effects}\hspace{-2.5em}  &L,D,E,F &::= & \cdot \mid \ell, E \\
  \slab{Contexts}\hspace{-1.5em} &\Gamma  &::= & \cdot
                  \mid \Gamma,x:_\varepsilon A
                  \mid \Gamma,\marker{}{E}
                  \mid \Gamma,\marker{\Lambda}{E}
                  \\
  \slab{Values}\hspace{-2.5em} &V,W  &::= & x \mid \elambda{\geffrow{E}{\varepsilon}} x^A . M \mid \Lambda \gray{\varepsilon} . V \\
\end{syntax}
\end{minipage}
\hspace{-1.8em}
\begin{minipage}[t]{0.45\textwidth}
\begin{syntax}
  \slab{Terms}\hspace{-2.5em}    &M,N  &::= & \Unit \mid x \mid \elambda{\geffrow{E}{\varepsilon}} x^A . M \mid M\,N \\
    & &\mid & \Lambda \gray{\varepsilon}. V \mid M\,\eapp\gray{\effrow{E}{\varepsilon}} \mid \Do \ell \; M  \\
    &     & \mid & \Mask_{L}\; M \mid \Handle\;M\With H \\
  \slab{Handlers}\hspace{-1.5em} &H &::= & \{\Return \; x \mapsto M\} \mid \{\ell\,p\,r \mapsto M\} \uplus H \\
\end{syntax}
\end{minipage}
\vspace{.5\baselineskip}

In types we include units, effectful functions, and effect abstraction
$\forall \gray{\varepsilon}. A$. As we consider only one effect
variable at a time, we need not track effect variables on function
types and effect abstraction.
Nonetheless, we include them in grey font for easier comparison with
existing calculi.
In $\Gamma$, each term variable is annotated with the effect variable
$\varepsilon$ that was referred to at the time of its introduction.
Further, we add markers $\marker{}{E}$ and $\marker{\Lambda}{E}$ to
the context, which track the change of effects due to functions,
masks, handlers, and effect abstraction.
These markers are not needed by the typing rules but help with the
encoding. As with \Met, we require contexts to be ordered.
For simplicity we assume operations are always present and defined by
a global context $\Sigma = \{\ol{\ell:A\sto B}\}$, thus unifying
extensions $D$, masks $L$, and effect contexts $E$ of \Met into one
syntactic category.
Mirroring our kind restriction for operation arrows in \Met, we
assume that these $A$ and $B$ are not function arrows, but they can be
effect abstractions (which may themselves contain function arrows).

\begin{figure}[t]
\raggedright
\boxed{\typl{\Gamma}{M:A}{\effrow{E}{\varepsilon}}}
\hfill
\begin{mathpar}
  \inferrule*[Lab=\rowlab{Var}]
  {
    \varepsilon = \varepsilon'
    \text{ or } \\\\
    A = \forall \gray{\varepsilon''}. A'
    \text{ or } A = \TUnit
  }
  {
    \typl{\Gamma_1,x:_{\varepsilon'} A,\Gamma_2}{x:A}{\effrow{E}{\varepsilon}}
  }
\quad\!
  \inferrule*[Lab=\rowlab{Abs}]
  {
    \typl{\Gamma, \marker{}{E}, x:_\varepsilon A}{M:B}{\effrow{F}{\varepsilon}}
  }
  {
    \typl{\Gamma}{\elambda{\geffrow{F}{\varepsilon}} x^A . M : \earr{A}{B}{\geffrow{F}{\varepsilon}}}{\effrow{E}{\varepsilon}}
  }
\quad\!
  \inferrule*[Lab=\rowlab{App}]
  {
    \typl{\Gamma}{M : \earr{A}{B}{\geffrow{E}{\varepsilon}}}{\effrow{E}{\varepsilon}} \\\\
    \typl{\Gamma}{N : A}{\effrow{E}{\varepsilon}} \\
  }
  {\typl{\Gamma}{M\, N : B}{\effrow{E}{\varepsilon}}}

  \inferrule*[Lab=\rowlab{EAbs}]
  {
    \varepsilon' \notin \ftv{\Gamma} \\\\
    \typl{\Gamma,\marker{\Lambda}{E}}{V : A}{\effrow{\cdot}{\varepsilon'}} \\
  }
  {
    \typl{\Gamma}{\Lambda \gray{\varepsilon'}. V : \forall \gray{\varepsilon'}.A}{\effrow{E}{\varepsilon}}
  }
\quad\ \
  \inferrule*[Lab=\rowlab{EApp}]
  {
    \typl{\Gamma}{M : \forall \gray{\varepsilon'}.A}{\effrow{E}{\varepsilon}}
  }
  {
    \typl{\Gamma}{M\,\eapp\gray{\effrow{E}{\varepsilon}} : A[\geffrow{E}{\varepsilon}/]}{\effrow{E}{\varepsilon}}
  }
\quad\ \
  \inferrule*[Lab=\rowlab{Mask}]
  {
    \typl{\Gamma,\marker{}{L+E}}{M: A}{\effrow{E}{\varepsilon}}
  }
  {
    \typl{\Gamma}{\Mask_{L}\; M: A}{\effrow{L+E}{\varepsilon}}
  }

  \inferrule*[Lab=\rowlab{Do}]
  {
    (\ell:A\sto B) \in \Sigma \\\\
    \typl{\Gamma}{M : A}{\effrow{\ell,E}{\varepsilon}}
  }
  {
    \typl{\Gamma}{\Do \ell \; M : B}{\effrow{\ell,E}{\varepsilon}}
  }
\quad\
  \inferrule*[Lab=\rowlab{Handler}]
  {
    \typl{\Gamma,\marker{}{E}}{M : A}{\effrow{\ol{\ell_i},E}{\varepsilon}} \\
    \typl{\Gamma, x :_\varepsilon A}{N : B}{\effrow{E}{\varepsilon}} \\\\
    \{\ell_i:A_i\sto B_i\} \subseteq \Sigma \\
    [\typl{\Gamma, p_i :_\varepsilon A_i, {r_i}:_\varepsilon B_i \to^{\geffrow{E}{\evar}} B}{N_i : B}{\effrow{E}{\varepsilon}}]_i \\
  }
  {
    \typl{\Gamma}{\Handle\;M\With \{\Ret x \mapsto N\} \uplus \{ \ell_i\;p_i\;r_i \mapsto N_i \}_i : B}{\effrow{E}{\varepsilon}}
  }
\end{mathpar}
\caption{Typing rules of \Leff.}
\label{fig:Feff}
\end{figure}

\Cref{fig:Feff} gives the typing rules of \Leff.
The judgement $\typl{\Gamma}{M : A}{\effrow{E}{\varepsilon}}$ states
that in context $\Gamma$, the term $M$ has type $A$ and might use
concrete effects $E$ extended with effect variable $\varepsilon$.
The typing rules are mostly standard for row-based effect systems.
The $\rowlab{Var}$ rule ensures that either the current effect
variable matches the effect variable at which the variable was
introduced or that the value is an effect abstraction or unit.
These constraints guarantee that programs can only refer to one effect
variable in one scope.
The \rowlab{App}, \rowlab{Do}, $\rowlab{Mask}$, and $\rowlab{Handler}$
rules are standard.
The $\rowlab{Abs}$ rule is standard except for requiring the effect
variable to remain unchanged.
The \rowlab{EAbs} rule introduces a new effect variable $\evar'$ and
the $\rowlab{EApp}$ rule instantiates an effect abstraction.
While conventional systems allow instantiating with any effect row,
\rowlab{EApp} only allows instantiation with the ambient effects $E$
and effect variable $\evar$.
The instantiation operator $[\geffrow{E}{\varepsilon}/]$ implements
standard type substitution for the single effect variable as follows.

\noindent\vspace{-.5\baselineskip}
\[
\ba{c}
\TUnit[\geffrow{E}{\varepsilon}/] \,=\, \TUnit \qquad\qquad
(\forall \gray{\varepsilon'}.A)[\geffrow{E}{\varepsilon}/] \,=\, \forall \gray{\varepsilon'}.A \\
(\earr{A}{B}{\geffrow{F}{\varepsilon'}})[\geffrow{E}{\varepsilon}/] \,=\, \earr{A[\geffrow{E}{\varepsilon}/]}{B[\geffrow{E}{\varepsilon}/]}{\geffrow{F,E}{\varepsilon}} \\
\ea
\]

For example, the following \Leff function (grey parts omitted) sums up
all yielded integers.
\[\ba{r@{~~}c@{~~}l}
  \text{\lstinline{asSum}} &:& \forall . \earr{(\earr{\TUnit}{\TUnit}{{\code{Gen Int}}{}})}{\Int}{} \\
  \text{\lstinline{asSum}}
  &=& \Lambda. \lambda f. \Handle\;f\;()\With \{
    \Ret x \mapsto 0,
    \Yield\; x\;r \mapsto x + r\;\Unit\}
\ea\]

\subsection{Encoding}
\label{sec:encoding-Leff}

We now give compositional translations for types and contexts of \Leff
into \Met.
We transform \Leff types at effect context $E$ to modal types in
\Met by the translation $\stransl{-}{E}$.

\begin{minipage}[t]{0.45\textwidth}
  \[\ba[t]{r@{\ \ }c@{\ \ }l}
    \stransl{\TUnit}{E} &=& \boxwith{\one} \TUnit \\
    \stransl{\earr{A}{B}{F}}{E} &=& \boxwith{\adj{E-F}{F-E}}(\stransl{A}{F}\to\stransl{B}{F}) \\
    \stransl{\forall.A}{E} &=& \boxwith{\aeq{}} \stransl{A}{\cdot} \\
  \ea
  \]
\end{minipage}
\begin{minipage}[t]{0.55\textwidth}
  \vspace{-.3em}
  \[\ba[t]{r@{\ \ }c@{\ \ }l}
    \stransl{\cdot}{E} &=& \cdot \\
    \stransl{\Gamma,x:A}{E} &=& \stransl{\Gamma}{E}, x\varb{\mu_E}{A'} \text{ for } \boxwith{\mu} A' = \stransl{A}{E} \\
    \stransl{\Gamma,\marker{}{F}}{E} &=& \stransl{\Gamma}{F}, \lockwith{\adj{F-E}{E-F}} \\
    \stransl{\Gamma,\marker{\Lambda}{F}}{\cdot} &=& \stransl{\Gamma}{F}, \lockwith{\aeq{}} \\
  \ea
  \]
  \\
\end{minipage}

For the unit type, we insert the identity modality for the uniformity
of the translation.
For a function arrow $\earr{A}{B}{F}$, we use a relative modality
$\adj{E-F}{F-E}$ to wrap the function arrow.
This modality transforms the outside effect context $E$ to $F$ which
is specified by the original function arrow $\earr{A}{B}{F}$.
For effect abstraction, we use an empty absolute modality to wrap the
type, simulating entering a new scope with different effect variables.
We translate contexts by translating each type and moving top-level
modalities to their bindings. For each marker, we insert a corresponding
lock to reflect the changes of effect context. %

As an example of type translation, the type of \lstinline{asSum} in
\Leff from \Cref{sec:Leff} is translated to
\[\ba{rcl}
\text{\lstinline{asSum}} &:& \boxwith{\aeq{}}(\boxwith{\aex{\code{Gen Int}}}(\boxwith{\one}\TUnit \to \boxwith{\one}\TUnit) \to \boxwith{\one}\Int). \\
\ea\]

Observe that not every valid typing judgement in \Leff can be
transformed to a valid typing judgement in \Met, because the translation
depends on markers in contexts, while the typing of \Leff does not.
We define well-scoped typing judgements, which characterise the typing
judgements for which our encoding is well-defined, as follows.

\begin{restatable}[Well-scoped]{definition}{LeffWellScoped}
  A typing judgement $\typl{\Gamma_1, x :_\varepsilon A, \Gamma_2}{M :
  B}{\effrow{E}{\evar}}$ is \textit{well-scoped for $x$} if either $x
  \notin \fv{M}$ or $\marker{\Lambda}{F} \notin \Gamma_2$ or $A =
  \forall. A'$ or $A = \TUnit$. A typing judgement $\typl{\Gamma}{M :
  A}{\effrow{E}{\evar}}$ is \textit{well-scoped} if it is well-scoped for all $x \in
  \Gamma$.
\end{restatable}

In particular, if the judgement at the leaf of a derivation tree is
well-scoped, then every judgement in the derivation tree is
well-scoped.
Also note that contexts with no markers are always well-scoped.
Consequently, restricting judgements to well-scoped ones in \Leff does
not reduce expressive power, as we can always ensure that contexts at
the leaves have no markers.

We write $\encode{M:A}{E}{M'}$ for the translation of a \Leff term $M$
of type $A$ with effect context $E$ to a \Met term $M'$.
We have the following type preservation theorem.
The translation of terms and proof of type preservation can be found
in \Cref{app:encoding}.
\begin{restatable}[Type preservation of encoding]{lemma}{LeffToMet}
  If $\typl{\Gamma}{M:A}{\effrow{E}{\varepsilon}}$ is well-scoped, then $\encode{M:A}{E}{M'}$
  and $\typm{\stransl{\Gamma}{E}}{M':\stransl{A}{E}}{E}$.
\end{restatable}
Intuitively, the term translation inserts boxing and unboxing
constructs in appropriate places in order to realise transition
between effect contexts as embodied by the type translation.
Specifically, for variables we unbox them immediately after they are
bound, and re-box them when they are used.
As an example of term translation, we can translate the
\lstinline{asSum} handler in \Leff defined in \Cref{sec:Leff} as
follows into \Met, omitting boxing and unboxing of the identity
modality $\one$.
\[\ba{r@{~~}c@{~~}l}
  \text{\lstinline{asSum}} &=&
  \bl
    \Box_{\aeq{}} (\lambda f. \Letm{}{\aex{\code{Gen Int}}} f = f \In
    \Handle\;f\;()\With \\
     \{
      \Ret x \mapsto \Letm{}{\aex{\code{Gen Int}}} x = x \In 0,
      \Yield\; x\;r \mapsto x + r\;\Unit\})
  \el
\ea\]
The explicit boxing and unboxing, as well as the identity modalities,
here are only generated to keep the encoding systematic.
We need not write them in practice as illustrated in
\Cref{sec:overview}.

Our encoding focuses on presenting the core idea and does not consider
advanced features including data types and polymorphism.
In \Cref{sec:extensions} we show how to extend \Met with these
features and in \Cref{app:extensible-encoding} we briefly outline how
to extend the encoding to cover them.

\section{Extensions}
\label{sec:extensions}

In this section we demonstrate that \Met scales to support data types
and polymorphism including both value and effect polymorphism.
Effect polymorphism helps deal with situations in which it is
useful to refer to one or more effect contexts that differ from the
ambient one (such as the higher-order fork operation
in~\Cref{sec:higher-order-fork}), recovering the full expressive power
of row-based effect systems.
We only discuss the key ideas of extensions here; their full
specification as well as more extensions including shallow
handlers~\citep{KammarLO13,HillerstromL18} are given in
\Cref{app:meet}.
We prove type soundness and effect safety for all extensions in
\Cref{app:CalcM}.

\subsection{Making Data Types Crisp}
\label{sec:data-types-crisp}

We demonstrate the extensibility of \Met with data types by extending
it with pair and sum types.
We expect no extra challenge to extend \Met with algebraic data types.
The syntax and typing rules are shown as follows.
\begin{mathpar}
\inferrule*[Lab=\tylab{Pair}]
{
  \typm{\Gamma}{M:A}{E} \\
  \typm{\Gamma}{N:B}{E} \\
}
{\typm{\Gamma}{(M,N):\Pair{A}{B}}{E}}

\inferrule*[Lab=\tylab{Inl}]
{
  \typm{\Gamma}{M:A}{E} \\
}
{\typm{\Gamma}{\Inl M:A+B}{E}}

\inferrule*[Lab=\tylab{Inr}]
{
  \typm{\Gamma}{M:B}{E} \\
}
{\typm{\Gamma}{\Inr M:A+B}{E}}

\hl{
\inferrule*[Lab=\tylab{CrispPair}]
{
  \nu_F : E\to F \\
  \typm{\Gamma, \lockwith{\nu_F}}{V:\Pair{A}{B}}{E} \\\\
  \typm{\Gamma, x\varb{\nu_F}{A}, y\varb{\nu_F}{B}}{M:A'}{F}
}
{\typm{\Gamma}{\Casem{\nu} V\Of (x,y)\mapsto M : A'}{F}}
}
\quad
\hl{
\inferrule*[Lab=\tylab{CrispSum}]
{
  \nu_F : E\to F \\
  \typm{\Gamma, \lockwith{\nu_F}}{V:A+B}{E} \\\\
  \typm{\Gamma, x\varb{\nu_F}{A}}{M_1:A'}{F} \\
  \typm{\Gamma, y\varb{\nu_F}{B}}{M_2:A'}{F}
}
{\typm{\Gamma}{\Casem{\nu}V\Of \{\Inl x \mapsto M_1, \Inr y\mapsto M_2\} : A'}{F}}
}
\end{mathpar}

The \tylab{Pair}, \tylab{Inl}, and \tylab{Inr} are standard
introduction rules.
The elimination rules \tylab{CrispPair} and \tylab{CrispSum} are more
interesting. In addition to normal pattern matching, they interpret
the value $V$ under the effect context transformed by certain
modalities $\nu$, which can then be tagged to the variable bindings in
case clauses.
They follow the crisp induction principles of multimodal type
theory~\citep{GratzerKNB20,Shulman18,GratzerKNB21}.
These crisp elimination rules provide extra expressiveness. For
example, we can write the following function which transforms a sum
of type $\boxwith{\mu}(A+B)$ to another sum of type $(\boxwith{\mu}A
+ \boxwith{\mu}B)$.
This function is not expressible without crisp elimination rules.
\[
\lambda x^{\boxwith{\mu}(A+B)} . \Letm{}{\mu} y = x \In \Casem{\mu} y \Of
\{
  \Inl x_1 \mapsto \Inl (\Box_\mu\;x_1),
  \Inr x_2 \mapsto \Inr (\Box_\mu\;x_2)
\}
\]

\subsection{Polymorphism for Values and Effects}
\label{sec:meet-polymorphism}

The extensions to syntax and typing rules with polymorphism
are as follows.

\noindent
\begin{minipage}[t]{0.45\textwidth}
\begin{syntax}
\slab{Types}    &A,B  &::= & \cdots \mid \forall\alpha^K.A\\
\slab{Effects} &E    &::= & \cdots
                      \mid {\evar \vphantom{\backslash}}
                      \mid {\eminus{E}{L}} \\
\slab{Kinds}          &K &::= &  \cdots \mid \Effect \\
\end{syntax}
\end{minipage}
\begin{minipage}[t]{0.5\textwidth}
  \vspace{-.5em}
\begin{syntax}
\slab{Contexts}      &\Gamma &::=& \cdots \mid \Gamma,\alpha:K \\
\slab{Terms}   &M,N  &::= & \cdots \mid \Lambda\alpha^K.V\mid M\,A \\
\slab{Values}         &V,W  &::= & \cdots \mid \Lambda \alpha^K.V \mid V\,A \\
\end{syntax}
\end{minipage}

\begin{mathpar}
\inferrule*[Lab=\tylab{TAbs}]
{
  \typm{\Gamma,\alpha:K}{V:A}{E}
}
{\typm{\Gamma}{\Lambda\alpha^K.V:\forall\alpha^K.A}{E}}

\inferrule*[Lab=\tylab{TApp}]
{
  \typm{\Gamma}{M:\forall\alpha^K.B}{E} \\
  \Gamma\vdash A : K \\
}
{\typm{\Gamma}{M\,A : B[A/\alpha]}{E}}
\end{mathpar}

It may appear surprising that we treat type application $V\,A$ as
values. This is useful in practice to allow instantiation inside
boxes. We also extend the semantics to allow reduction in values.

To support effect polymorphism, we extend the syntax of effect
contexts $E$ with effect variables $\evar$ and introduce a new kind
$\Effect$ for them.
As is typical for row polymorphism, we restrict each effect type to
contain at most one effect variable.
We also extend the syntax with effect masking $\eminus{E}{L}$, which
means the effect types given by masking $L$ from $E$.
The latter is needed to keep the syntax of effect contexts closed
under the masking operation $E-L$; otherwise we cannot define
$\evar-L$.
In other words, the syntax of effects is the free algebra generated
from extending $D,E$ and masking $\eminus{E}{L}$ with base elements
$\cdot$ and $\evar$.

The effect equivalence and subeffecting rules are extended in a
relatively standard way.
\begin{mathpar}
\inferrule*
{ }
{\eminus{E}{\cdot} \equiv E}

\inferrule*
{ }
{\eminus{\cdot}{L} \equiv \cdot}

\inferrule*
{ }
{\eminus{(\ell:P,E)}{(\ell,L)} \equiv \eminus{E}{L}}

\inferrule*
{ \ell \notin L }
{\eminus{(\ell:P,E)}{L} \equiv \ell:P,\eminus{E}{L}}

\inferrule*
{ }
{\eminus{(\eminus{\evar}{L})}{L'} \equiv \eminus{\evar}{(L,L')}}

\inferrule*
{ }
{\eminus{\evar}{L} \equiv \eminus{\evar}{L}}

\inferrule*
{ }
{\cdot \subtype \eminus{\evar}{L}}

\inferrule*
{ }
{\eminus{\evar}{L} \subtype \eminus{\evar}{L}}
\end{mathpar}
We do not allow non-trivial equivalence or subtyping between different
effect variables.
We always identify effects up to the equivalence relation. That is, we
can directly treat syntax of effects as the free algebra quotiented by
the equivalence relation $E\equiv F$.
Observe that using the equivalence relation, all open effect types
with effect variable $\evar$ can be simplified to an equivalent normal
form $D,\eminus{\evar}{L}$.
We assume the operation $E-L$ is defined for effects $E$ in normal
form and extend it with one case for effect variables as
$\eminus{\evar}{L} - L' = \eminus{\evar}{(L,L')}$.

\FloatBarrier
\section{Simple Bidirectional Type Checking and Elaboration}
\label{sec:surface-lang}

{
\renewcommand{\transto}[1]{\ \hl{\dashrightarrow #1}}
\renewcommand{\fortrans}[1]{\hl{#1}}

In this section we outline the design of \SurfaceMet, a basic surface
language on top of \Met, which uses a simple bidirectional typing
strategy to infer all boxing and unboxing~\citep{PierceT00,
DunfieldK21}.

The bidirectional typing rules for simply-typed $\lambda$-calculus and
modalities of \SurfaceMet and its type-directed elaboration to \Met
are shown in \Cref{fig:selected-typing-maetel}.
We show the full typing and elaboration rules and discuss our
implementation in \Cref{app:surface-lang}.

\begin{figure}[htbp]
\raggedright
\boxed{\typmi{\Gamma}{M}{A}{E}{}\transto{N}}
\boxed{\typmc{\Gamma}{M}{A}{E}\transto{N}}
\hfill

\vspace{-\baselineskip}
\begin{mathpar}
\inferrule*[Lab=\btylab{Var}]
{
  \nu_F = \locks{\Gamma'} : E\to F\\
  \boxwith{\ol{\zeta}}G = \across{\Gamma}{A}{\nu}{F}
}
{
  \typmi{\Gamma,x:\boxwith{\ol{\mu}}G,\Gamma'}{x}{\boxwith{\ol{\zeta}}G}{E}{}
  \transto{\Mod_{\ol\zeta}\;x}
}

\inferrule*[Lab=\btylab{Mod}]
{
  \mu_F : E \to F \\\\
  \typmc{\Gamma,\lockwith{\mu_F}}{V}{A}{E} \transto{V'} \\
}
{
  \typmc{\Gamma}{V}{\boxwith{\mu} A}{F}
  \transto{\Mod_\mu\;V'}
}

\inferrule*[Lab=\btylab{Abs}]
{
  \typmc{\Gamma, x : \boxwith{\ol{\mu}}G}{M}{B}{E} \transto{M'} \\
}
{
  \typmc{\Gamma}{\lambda x . M}{\boxwith{\ol{\mu}}G \to B}{E} \\\\
  \transto{\lambda x . \Letm{}{\ol{\mu}} \hat{x} = x \In M'}
}

\inferrule*[Lab=\btylab{App}]
{
  \typmi{\Gamma}{M}{\boxwith{\ol{\mu}}(A \to B)}{E}{} \transto{M'} \\\\
  {\ol{\mu}}_E \To \one_E \\
  \typmc{\Gamma}{N}{A}{E} \transto{N'} \\
}
{
  \typmi{\Gamma}{M\; N}{B}{E}{}
  \transto{
    (\Letm{}{\ol{\mu}} x = M' \In x)
  \;N'}
}

\inferrule*[Lab=\btylab{Annotation}]
{
  \typmc{\Gamma}{V}{A}{E} \transto{V'} \\
}
{\typmi{\Gamma}{V:A}{A}{E}{} \transto{V'}}

\inferrule*[Lab=\btylab{Switch}]
{
  \typmi{\Gamma}{M}{\boxwith{\ol{\mu}} G}{E}{} \transto{M'} \\
  \Gamma\vdash ({\ol{\mu}}, G) \To {\ol{\nu}} \atmode{E} \\
}
{
  \typmc{\Gamma}{M}{\boxwith{\ol{\nu}} G}{E}
  \transto{\Letm{}{\ol{\mu}} x = M' \In \Mod_{\ol{\nu}}\; x}
}
\end{mathpar}
\caption{Representative bidirectional typing and elaboration rules for \SurfaceMet.}
\label{fig:selected-typing-maetel}
\end{figure}

As with usual bidirectional typing, we have inference mode
$\typmi{\Gamma}{M}{A}{E}{}$ and checking mode
$\typmc{\Gamma}{M}{A}{E}$.
They both additionally take the effect context $E$ as an input.
The highlighted part $\transto{N}$ gives the elaborated term $N$ in
\Met.

We write $G$ for guarded types which do not have top-level modalities,
and $\ol{\mu}$ for a sequence of modalities which can be empty.
As a result, every type is in form $\ol{\mu} G$.
We use the following syntactic sugar for boxing and unboxing modality
sequences to simplify the elaboration.
\[\ba{rcl}
\Mod_{\ol{\mu}}\;V &\doteq& \ol{\Mod_{\mu}}\;V \\
\Letm{}{\ol{\mu}} x = M \In N &\doteq& (\lambda x . \ol{\Letm{}{\mu} x = x \In} N)\; M \\
\ea\]

Our bidirectional typing rules are mostly simple and standard.
The most novel part is the usage of an auxiliary function
$\meta{across}$ in \btylab{Var} defined as follows.
\[\ba{rcl}
\meta{across}(\Gamma,A,\nu,F) &=&
\begin{cases}
  A, &\text{if } \Gamma\vdash A : \Pure \\
  \boxwith{\zeta} G, &\text{otherwise, where } A = \boxwith{\ol{\mu}}{G}
    \text{ and } \rresidual{\ol{\mu}_F}{\nu_F} = \zeta_E
\end{cases}
\ea\]
When $A$ is absolute, we can always access the variable.
Otherwise, in order to know how far we should unbox the modalities
$\ol{\mu}$ of the variable, we define a right residual operation
$\rresidual{\mu_F}{\nu_F}$ for the modality transformation relation.
Given $\mu_F : E\to F$ and $\nu_F : F' \to F$, the partial operation
$\rresidual{\mu_F}{\nu_F}$ fails if there does not exist $\zeta_{F'}$
such that $\mu_F \To \nu_F \circ \zeta_{F'}$. Otherwise, it gives an
indexed modality such that $\mu_F \To \nu_F \circ
(\rresidual{\mu_F}{\nu_F})$ and for any $\zeta_{F'}$ with $\mu_F
\To \nu_F \circ \zeta_{F'}$, we have $\rresidual{\mu_F}{\nu_F} \To
\zeta_{F'}$.
Intuitively, $\rresidual{\mu_F}{\nu_F}$ gives the best solution
$\zeta_{F'}$ for the transformation $\mu_F \To \nu_F\circ\zeta_{F'}$
to hold.
The concrete definition of $\rresidual{\mu_F}{\nu_F}$ is given in
\Cref{app:surface-lang-spec}.

\btylab{App} unboxes $M$ when it has top-level modalities and inserts
explicit unboxing in the elaborated term.
\btylab{Mod} introduces a lock into the context and inserts
explicit boxing in the elaborated term.
\btylab{Annotation} is standard for bidirectional typing.
\btylab{Switch} not only switches the direction from checking to
inference, but also transforms the top-level modalities when there is
a mismatch by inserting unboxing and re-boxing.
It uses the judgement $\Gamma\vdash(\mu,A)\To\nu\atmode{E}$ defined in
\Cref{sec:met-kinds-contexts}.

Though incorporating polymorphic type inference is beyond the scope of
this paper, we are confident that modal effect types are compatible
with it.
The key observation here is that in the presence of polymorphism, the
problem of automatically boxing and unboxing is closely related to
that of inferring first-class polymorphism.
Modality introduction is analogous to type abstraction (which type
inference algorithms realise through generalisation).
Modality elimination is analogous to type application (which type
inference algorithms realise through instantiation).
As such, one can adapt any of the myriad techniques for combining
first-class polymorphism with Hindley-Milner type inference.
As we adopt a bidirectional type system, we could simply follow the
literature on extending bidirectional typing with sound and complete
inference for higher-rank polymorphism~\citep{DunfieldK13},
first-class polymorphism~\citep{ZhaoO22} and bounded
quantification~\citep{CuiJO23}.

In the future, we plan to further explore type inference for modal
effect types and in particular design an extension to \OCaml, building
on and complementing recent work on modal types for
\OCaml~\citep{LorenzenWDEL24} and making use of existing techniques
for supporting first-class polymorphism.

}

\FloatBarrier
\section{Related and Future Work}
\label{sec:related-work}

\subsection{Capability-based Effect Systems}
\label{sec:capability-based-effect-systems}

Capability-based effect systems such as
\Effekt~\citep{BrachthauserSO20,BrachthauserSLB22} and
\CaptureCalculus~\citep{BoruchGruszeckiOLLB23} interpret effects as
capabilities and offer a form of effect polymorphism through
capability passing.

For instance, in \Effekt the \lstinline{asList} handler in
\Cref{sec:relative-modalities} has the following type:
\begin{lstlisting}[style=koka]
  def asList{ f: Unit => List[Int] / { Gen[Int] } }: List[Int] / {}
\end{lstlisting}
The special \emph{block parameter} (or \emph{capability})
\lstinline|f| can use the effect \lstinline{Gen[Int]} in addition to
those from the context. The annotation \lstinline${ Gen[Int] }$ is
similar to our relative modalities.

A key difference between \Effekt and \Met is that \Effekt requires
blocks to be second-class, whereas \Met supports first-class functions
by default.
For instance, consider the standard composition function:
\lstinline{compose f g x = g (f x)}. We cannot write this function
naively in \Effekt as it relies on first-class functions.
One solution is to uncurry it.
\begin{lstlisting}
  def composeUncurried[A, B, C](x: A){ f: A => B }{ g: B => C }: C / {}
\end{lstlisting}
Note that we move \lstinline{x} to be the first argument, as \Effekt
requires value parameters to appear before block parameters. We cannot
partially apply \lstinline{composeUncurried}.

\citet{BrachthauserSLB22} recover first-class functions in \Effekt by
boxing blocks.
However, such boxed blocks can only use those capabilities specified
in the box types, similarly to our absolute modalities.
With boxes, we can write the curried \lstinline{compose} with the
following type signature
\begin{lstlisting}
  def compose[A, B, C]{ f: A => B }{ g: B => C }: A => C at {f, g} / {}
\end{lstlisting}
which returns a value of type \lstinline|A => C at {f, g}| --- a
first-class boxed block.
The annotation \lstinline|{f, g}| indicates that this block captures
the capabilities \lstinline{f} and \lstinline{g}.
This kind of annotation is reminiscent of effect variables, and indeed
such examples illustrate why \Met without effect polymorphism is not
expressive enough to encode all of \Effekt.
If we extend \Met with effect variables (\Cref{sec:meet-polymorphism})
then it is possible to encode capability variables like \lstinline{f}
and \lstinline{g}.

Another key difference between \Met and \Effekt is that \Effekt
uses named handlers~\citep{BiernackiPPS20,XieCIL22,ZhangM19}, in which
operations are dispatched to a specific named handler, whereas \Met
uses \citet{PlotkinP13}-style handlers, in which operations are
dispatched to the first matching handler in the dynamic context.
Named handlers also provide a form of effect generativity. In future
it would be interesting to explore variants of modal effect types with
named handlers and generative effects~\citep{VilhenaP23}.

\CaptureCalculus~\citep{BoruchGruszeckiOLLB23}, the basis for capture
tracking in \Scala 3, also provides succinct types for uncurried
higher-order functions like \lstinline{composeUncurried}.
As in \Effekt, the curried version requires the result function to be
explicitly annotated with its capture set \lstinline|{f,g}|.
Though \CaptureCalculus is a good fit for \Scala 3, it does rely on
existing advanced features like path-dependent types (especially the
ability of using term variables in types) and implicit parameters.
Modal effect types do not require the language to support such
advanced features.

\subsection{Direct Comparison between \SurfaceMet, \Koka, and \Effekt}
\label{sec:comparison}

In \Cref{sec:overview}, we compare the types of \lstinline{iter},
\lstinline{asList}, \lstinline{state}, and \lstinline{schedule} in
\Met with their counterparts in a row-based effect system similar to
\Koka.
In \Cref{sec:capability-based-effect-systems}, we discuss the
differences between modal effect types and capability-based effect
systems such as \Effekt.
Here we present a more direct comparison of the type signatures of
typical programs in \SurfaceMet, \Koka, and \Effekt, in order to
demonstrate the practicality of modal effect types.

\paragraph{Invoking Effects}
\label{sec:invoking-effects}

An effect system tracks which effects are invoked.
Consider a function \lstinline{foo x =do yield (x + 42)} which uses
the \lstinline{yield} operation from \Cref{sec:absolute-modalities}.
Its types in \Metl, \Koka, and \Effekt are as follows.

\begin{lstlisting}[style=koka]
  foo : [Gen Int](Int -> 1)                      # METL
  foo : forall<e>. (x : int) -> <gen<int>|e> ()  # Koka
  def foo(x : Int) : Unit / { Gen[Int] }         # Effekt
\end{lstlisting}

Both \Metl and \Effekt support effect subtyping which enables us to
apply \lstinline{foo} with other effects.
\Koka does not support effect subtyping and uses an effect variable
\lstinline{e} for modularity.
\Koka supports a special typing rule which implicitly inserts effect
variables on function types in covariant position.
Consequently, we can in fact simply write the following type in \Koka.
\begin{lstlisting}
  foo : (x : int) -> (gen<int>) ()               # Koka
\end{lstlisting}

\paragraph{Handling Effects}

An effect system tracks which effects are handled.
Recall the \lstinline{asList} handler from \Cref{sec:relative-modalities}.
Its types in \Metl, \Koka, and \Effekt are as follows.

\begin{lstlisting}[style=koka]
  asList : [](<Gen Int>(1 -> 1) -> List Int)                           # METL
  asList : forall<e>. (action : () -> <gen<int>|e> ()) -> e list<int>  # Koka
  def asList{ f: Unit => List[Int] / { Gen[Int] } }: List[Int] / {}    # Effekt
\end{lstlisting}

Both \Metl and \Effekt allow the argument of \lstinline{asList} to use
the ambient effects in addition to \lstinline{Gen Int}.
In \Koka, we must make the argument polymorphic over other effects.

\paragraph{Functions in Data Types}

An effect system should be compatible with algebraic data types.
Consider a pair of \lstinline|foo| functions as defined in
\Cref{sec:invoking-effects}.

In \SurfaceMet, we can choose to either share a single absolute
modality between both components, or to have a separate absolute
modality for each component.
\begin{lstlisting}
  pair1 : [Gen Int](Int -> 1, Int -> 1)
  pair2 : ([Gen Int](Int -> 1), [Gen Int](Int -> 1))
\end{lstlisting}
The values \lstinline|pair1| and \lstinline|pair2| can be easily
converted between one another.

In \Koka, we must make both functions effect-polymorphic.
We can choose either to use the same or different effect variables:

\begin{lstlisting}[style=koka]
  pair3 : forall<e,e1>. ((x : int) -> <gen<int>|e> (), (y : int) -> <gen<int>|e1> ())
  pair4 : forall<e>.    ((x : int) -> <gen<int>|e> (), (y : int) -> <gen<int>|e> ())
\end{lstlisting}

In \Effekt, as discussed in
\Cref{sec:capability-based-effect-systems}, we cannot express this
pair of functions directly, as functions are second-class. We must box
the functions to make them first-class.

\begin{lstlisting}[style=koka]
  pair5 : Tuple2[Int => Unit / { Gen[Int] } at {}, Int => Unit / { Gen[Int] } at {}]
\end{lstlisting}

The syntax \lstinline$at {}$ means that preceding function type is
boxed with no captured capability.

\paragraph{Higher-Order Functions}

An effect system should be compatible with higher-order functions.
Consider the standard sequential composition function
\lstinline{compose f g x = g (f x)}.
Its type in \Metl and \Koka are as follows.

\begin{lstlisting}[style=koka]
  compose : forall a b c . []((a -> b) -> (b -> c) -> (a -> c))                   # METL
  compose : forall<a,b,c,e>. (f : (a) ->e b) -> (g : (b) ->e c) -> ((x : a) ->e c)# Koka
\end{lstlisting}

In \Effekt, as we discussed in
\Cref{sec:capability-based-effect-systems}, we must either switch to
an uncurried version or box the result with capability variables.

\begin{lstlisting}[style=koka]
  def composeUncurried[A, B, C](x: A){ f: A => B }{ g: B => C }: C / {}  # Effekt
  def compose[A, B, C]{ f: A => B }{ g: B => C }: A => C at {f, g} / {}  # Effekt
\end{lstlisting}

In summary, compared to \Koka, \SurfaceMet provides more succinct
types without effect variables for a rich class of programs.
Compared to \Effekt, \SurfaceMet supports first-class functions
smoothly with no extra restrictions or capability variables in types.

\subsection{Frank}

Our absolute and relative modalities are inspired by \emph{abilities}
and \emph{adjustments} in \Frank~\citep{frank,ConventLMM20}.
Absolute modalities and abilities specify the whole effect context
required to run some computation. Relative modalities and adjustments
specify changes to the ambient effect context.
The key difference is that \Frank is based on a traditional row-based
effect system and implicitly inserts effect variables into
higher-order programs.
This is a fragile syntactic abstraction as discussed in
\Cref{sec:introduction}.
In contrast, \Met exploits modal types to robustly capture the essence
of modular effect programming without effect polymorphism.
As demonstrated in \Cref{sec:encodings}, a core \Frank-like calculus
with implicit effect variables is expressible in \Met.
\Frank's \emph{adaptors} are richer than \Met's masking, though we
expect relative modalities to extend readily to encompass the full
power of adaptors.

Unlike adjustments in \Frank, modal types are first-class types just
like data types and can appear anywhere.
For instance, we can put two functions with modal types in a pair.
\begin{lstlisting}
  handleTwo : []((<Gen Int>(Unit -> Unit), <State Int>(Unit -> Unit)) -> (List Int, Unit))
  handleTwo (f, g) = (asList f, state g 42)
\end{lstlisting}

\subsection{Relationship Between \Met and Multimodal Type Theory}
\label{sec:related-work-mtt}

The literature on multimodal type theory organises the structure of
modes (objects), modalities (morphisms between objects), and their
transformations (2-cells between morphisms) in a
\emph{2-category}~\citep{GratzerKNB20,Gratzer23,KavvosG23}
(or, in the case of a single mode, a semiring~\citep{AbelB20,ChoudhuryEEW21,OrchardLE19,PetricekOM14}).
In \Met, modes are effect contexts $E$, modalities are of the form
$\mu_F : E\to F$, and transformations are of the form $\mu_F \To
\nu_F$.
However, 2-categories are insufficient in a system that also includes
submoding.
The extra structure can be captured by moving to \emph{double
categories}, which have an additional kind of vertical morphism
between objects (in \Met, vertical morphisms are given by the
subeffecting relation $E\subtype F$), as also proposed
by~\citet{Katsumata18}.
Consequently, the transformations do not strictly require the two
modalities to have the same sources and targets, enabling us to have
$\aeq{}_F \To \aeq{E}_F$ in \Met.
The relationship between \Met and MTT is explained in more detail in
\Cref{app:double-category-effects}.

\subsection{Other Related Work}

We discuss other related work on effect systems and modal types.

\paragraph{Row-based Effect Systems}
Row polymorphism is one popular approach to implementing effect
systems for effect handlers.
\Links \citep{linksrow} uses R{\'e}my-style row polymorphism with
presence types \citep{remy1994type}, whereas \Koka~\citep{koka} and
\Frank~\citep{frank} use scoped rows \citep{Leijen05} which allow
duplicated labels.
\citet{rose} propose a general framework for comparing different
styles of row types, and \citet{YoshiokaSI24} propose a similar
framework focusing on comparing effect rows.
\Met adopts Leijen-style scoped rows, but also allows operation labels
to be absent in the spirit of R{\'e}my-style presence types.

\paragraph{Subtyping-based Effect Systems}
\Eff~\citep{bauer13, pretnar14} is equipped with an effect system with
both effect variables and sub-effecting based on the type inference
and elaboration described in~\citet{KarachaliasPSVS20}.
The effect system of \Helium~\citep{BiernackiPPS20} is based on finite
sets, offering a natural sub-effecting relation corresponding to
set-inclusion. As such, their system aligns closely with
\citet{LucassenG88}-style effect systems.
\citet{SHL} propose a calculus for effect handlers with effect
polymorphism and sub-effecting via qualified
types~\citep{jones94,rose}.

\paragraph{Modal Types and Effects.}
\citet{ChoudhuryK20} propose to use the necessity modality to recover
purity from an effectful calculus, which is similar to our empty
absolute modality.
\citet{ZyuzinN21} extend
contextual modal types~\citep{NanevskiPP08} to algebraic effects and
handlers.
Their system lacks anything like our relative modality and thus cannot
benefit from ambient effect contexts due to strict syntactic
restrictions.
Consequently, they cannot provide concise modular types for
higher-order functions and handlers as \Met does.

\subsection{Future Work}

Future work includes: implementing our system as an extension to
\OCaml; exploring extensions of modal effect types with Fitch-style
unboxing, named handlers, and capabilities; combining modal effect
types with control-flow linearity~\citep{SHL}; and developing a
denotational semantics.

\begin{acks}
Sam Lindley was supported by UKRI Future Leaders Fellowship ``Effect
Handler Oriented Programming'' (MR/T043830/1 and MR/Z000351/1).
\end{acks}

\section*{Data-Availability Statement}

The prototype of our surface language \SurfaceMet is
available on Zenodo~\citep{MetArtifact}.
This prototype implements and extends the bidirectional type checker
outlined in \Cref{sec:surface-lang,app:surface-lang}.
It type checks all examples in the paper.

\bibliography{reference}

\appendix
\section{Full Specification of \Met with Extensions}
\label{app:meet}

We provide the full specification of \Met including all
extensions.

\subsection{More Extensions}
\label{app:more-extensions}

We first present three more extensions of modal effect types: richer
forms of handlers, boxing pure computations, and commuting
modalities with type abstraction.
We discuss the key ideas of these extensions below and show their full
specification in the following sub-sections.

\paragraph{Absolute and Shallow Handlers}
\label{sec:shallow-handlers}

Up to now we have considered only \emph{deep} handlers of the form
$\Handle\;M\With H$ where $M$ depends on the ambient effect contexts.
Deep handlers automatically wrap the handler around the body of the
continuation $r$ captured in a handler clause, and thus $r$ depends on
the ambient effect context.
Though this usually suffices in practice, in some cases we may want
the computation $M$ or the continuation to be absolute, i.e.,
independent from the ambient effect context.
We call such handlers \emph{absolute} handlers.
This situation is more prevalent with effect polymorphism.

To support absolute handlers, we extend the handler syntax and typing
rules as follows.
\begin{mathpar}
{
\inferrule*[Lab=\tylab{Handler${}^\all$}]
{
  {D = \{\ell_i : A_i \sto B_i\}_i} \\
  \typm{\Gamma,\lockwith{\aeq{D+E}_F}}{M: A}{D+E} \\\\
  \typm{\Gamma, \lockwith{\aeq{E}_F}, x : \boxwith{\aeq{D+E}} A}{N : B}{E}\\
  [\typm{\Gamma, \lockwith{\aeq{E}_F}, p_i : A_i, r_i :
    \boxwith{\aeq{E}}
    (B_i \to B)}{N_i : B}{E}]_i
  \\
  [E]_F \To \one_F
}
{\typm{\Gamma}{\Handle^\all\;M\With
  \{\Ret x \mapsto N\} \uplus \{ \ell_i\;p_i\;r_i \mapsto N_i \}_i : B}{F}}
}
\end{mathpar}
The \tylab{Handler${}^\all$} rule extends the context with an absolute
lock $\lockwith{\aeq{D+E}_F}$ specifying the effect context for $M$,
and boxes the continuation $r$ with the absolute modality $\aeq{E}$,
where $E$ exactly gives the effect context after handling.
We put the lock $\lockwith{\aeq{E}_F}$ in handler clauses as deep
handlers capture themselves into continuations.
We also extend the handler syntax with \emph{shallow} handlers
$\Handle^\dagger\,M\With H$, in which the handler is not automatically
wrapped around the body of continuations, and \emph{absolute shallow}
handlers
$\Handle^{\all\dagger}\,M\With H$~\citep{KammarLO13,HillerstromL18}.

\paragraph{Boxing Computations under Empty Effect Contexts}
\label{sec:axiom-k}

We have restricted boxes to values in order to guarantee effect
safety.
This restriction is not essential for $\boxwith{\aeq{}}$.
For example, suppose we have
$f\varb{\boxwith{\aeq{}_F}}{(A\to B)}$ and
$x\varb{\boxwith{\aeq{}_F}}{A}$,
it is sound to treat $\Box_{\aeq{}}\,(f\,x)$ as a computation which
returns a value of type $\boxwith{\aeq{}}B$.
As $f\,x$ is evaluated under the empty effect context, we can
guarantee that it cannot get stuck on unhandled operations.

We extend the introduction rule for the empty absolute modality to
allow non-value terms with the following typing rule.
\begin{mathpar}
\inferrule*[Lab=\tylab{ModAbs}]
{\typm{\Gamma,\lockwith{\aeq{}_F}}{M:A}{\cdot}}
{\typm{\Gamma}{\Box_{\aeq{}}\,M : \boxwith{\aeq{}} A}{F}}
\end{mathpar}
As an example, we can write the following $\var{app}$ function.
\[\ba{rcl}
\var{app} &:& \forall \alpha.\forall\beta.\boxwith{\aeq{}}(\alpha\to \beta) \to \boxwith{\aeq{}} \alpha \to \boxwith{\aeq{}} \beta \\
\var{app} &=& \Lambda\alpha.\Lambda\beta.\lambda f.\lambda x.
  \Letm{}{\aeq{}} f = f \In
  \Letm{}{\aeq{}} x = x \In
  \Box_{\aeq{}}\,(f\,x) \\
\ea\]
The formula corresponding to the type of this function is commonly
referred to as Axiom K in modal logic and is also satisfied by other
similar modalities such as the safety modality of~\citet{ChoudhuryK20}.

\paragraph{Commuting Modalities and Type Abstraction}
\label{sec:commuting-mod-poly}

Crisp elimination rules in \Cref{sec:data-types-crisp} allow us to
commute modalities and data types.
Similarly, it is also sound and useful to commute type abstractions
and modalities.
However, the current modality elimination rule cannot do so, for a
similar reason to why it is not possible to transform
$\forall\alpha.A+B$ to $(\forall\alpha.A) + (\forall\alpha.B)$ in
System F.
We extend modality elimination to the form
$\Letm{\nu}{\mu}\Lambda\ol{\alpha^K} x = V \In M$ which allows $V$ to
use additional type variables in $\ol{\alpha^K}$ which are abstracted
when bound to $x$.
The extended typing and reduction rules are as follows.
\begin{mathpar}
\inferrule*[Lab=\tylab{Letmod'}]
{
  \nu_F : E\to F \\
  \typm{\Gamma,\lockwith{\nu_F},\ol{\alpha:K}}{V : \boxwith{\mu} A}{E} \\
  \typm{\Gamma,x\varb{\nu_F\circ\mu_E}{\forall\ol{\alpha^K}.A}}{M:B}{F}
}
{\typm{\Gamma}{\Letm{\nu}{\mu} \Lambda\ol{\alpha^K} . x = V \In M : B}{F}}
\end{mathpar}
\begin{reductions}
  \semlab{Letmod'}  & \Letm{\nu}{\mu} \Lambda\ol{\alpha^K} . x = \Box_\mu\,U \In M
    &\reducesto& M[(\Lambda\ol{\alpha^K}.U)/x] \\
\end{reductions}
For instance, we can now write a function of type
${\forall\alpha^K.\boxwith{\mu}A} \to \boxwith{\mu}(\forall\alpha.A)$
where $\alpha\notin\ftv{\mu}$ as follows.
\[
\lambda x^{\forall\alpha^K.\boxwith{\mu}A} .
\Letm{}{\mu} \Lambda\alpha^K . y = x\;\alpha \In \Box_\mu\,y
\]

\subsection{Syntax}
\label{app:syntax-meet}

\Cref{fig:syntax-meet} gives the syntax of \Met with all extensions
including data types, polymorphism, and enriched handlers. We highlight
the syntax not present in core \Met.

\begin{figure}[htbp]
\begin{syntax}
\slab{Types}    &A,B  &::= & \TUnit \mid A\to B \mid {\boxwith{\mu} A} \mid
                        \hl{\alpha \mid \forall\alpha^K.A
                        \mid \Pair{A}{B} \mid A+B}
                          \\
\slab{Masks}          &L   &::= & \cdot \mid \ell,L \\
\slab{Extensions}\hspace{-1em}      &D   &::= & \cdot \mid \ell : P, D \\
\slab{Effect Contexts} &E,F &::= & \cdot \mid \ell:P,E
                      \mid \hl{\evar \vphantom{\backslash}
                      \mid \eminus{E}{L}} \\
\slab{Presence}\hspace{-1em}      &P   &::= & A\sto B \mid \Abs \\
\slab{Modalities}\hspace{-1em}      &\mu &::= & {\aeq{E}}
                                  \mid {\adj{L}{D}} \\
\slab{Kinds}          &K &::= &  {\Pure} \mid \Any \mid \hl{\Effect} \\
\slab{Contexts}\hspace{-2em}       &\Gamma &::=& \cdot
                                          \mid {\Gamma, x\varb{\mu_F}{A}}
                                          \mid {\Gamma,\lockwith{\mind{\mu}{F}}}
                                          \mid \hl{\Gamma,\alpha:K}
                                          \\
\slab{Terms}   &M,N  &::= & x \mid  \lambda x^A.M \mid M\,N \mid \hl{\Lambda\alpha^K.V \mid M\,A} \\
                      &     &\mid& {\Box_\mu\,V} \mid {\Letm{\nu}{\mu} x = V\In M} \\
                      &     &\mid& \Do\ell\; M \mid \Mask_L\,M \mid \Handle^{\hl{\delta}}\;M\With H \\
                      &     &\mid&  \hl{(M, N) \mid \Casem{\nu} V \Of (x, y) \mapsto M
                      } \\
                      &     &\mid&
                      \hl{
                        \Inl M \mid \Inr M \mid
                        \Casem{\nu} V \Of \{ \Inl x \mapsto M, \Inr y \mapsto N \}
                      } \\
                      &     &\mid& \hl{ {\Letm{\nu}{\mu} \Lambda\ol{\alpha^K} . x = V\In M}} \\
\slab{Values}         &V,W  &::= & x \mid \lambda x^A.M \mid {\Box_\mu\, V} \mid \hl{\Lambda \alpha^K.V \mid V\,A
   \mid (V, W) \mid \Inl V \mid \Inr V} \\
\slab{Handlers}\hspace{-2em}       &H    &::= & \{ \Ret x \mapsto M \}
                            \mid  \{ \ell \; p \; r \mapsto M \} \uplus H \\
\slab{Decorations}         &\delta& ::= & \hl{\cdot \mid \all \mid \dagger \mid \all\dagger} \\
\end{syntax}
\caption{Syntax of \Met with all extensions.}
\label{fig:syntax-meet}
\end{figure}

\subsection{Kinding, Well-Formedness, Type Equivalence and Sub-effecting}
\label{app:rules-meet}

The full kinding and well-formedness rules for \Met are shown in
\Cref{fig:kinding-meet}.
The type equivalence and sub-effecting rules are shown in
\Cref{fig:equiv-sub-meet}.
We highlight the special rule that allows us to add or remove absent
labels from the right of effect contexts.

\subsection{Auxiliary Operations}

Since we extend the syntax of effect contexts $E$, we also need to
define a new case for the operation $E - L$ as follows. The
definitions of other operations $D+E$ and $L\bowtie D$ remain
unchanged from those in \Cref{sec:modes}.
We include the full definition here for easy reference.

\boxed{D + E}
\boxed{E - L}
\boxed{L \bowtie D}
\hfill
\[\ba{r@{~}c@{~}l}
D+E &=& D,E \\
\cdot - L &=& \cdot \\
(\ell:P,E) - L &=&
\begin{cases}
  E - L'         &\text{if } L\equiv \ell,L' \\
  \ell:P,(E-L)   &\text{otherwise}
\end{cases} \\
\ea
\quad
\ba{r@{~}c@{~}l}
L \bowtie \cdot &=& (L,\cdot) \\
L \bowtie (\ell:P,D) &=&
  \begin{cases}
    L'\bowtie D      &\text{if } L\equiv\ell,L' \\
    (L',(\ell:P,D'))  &\text{otherwise} \\
    \span\text{where } (L',D') = L\bowtie D \\
  \end{cases}
\ea\]

Since we extend the syntax of contexts, we need to extend
$\locks{\Gamma}$ with one extra trivial case.

\[\ba{r@{~}c@{~}l}
\locks{\cdot} &=& {\one} \\
\locks{\Gamma,\lockwith{\mind{\mu}{F}}} &=& \locks{\Gamma}\circ\mind{\mu}{F} \\
\ea
\qquad
\ba{r@{~}c@{~}l}
\locks{\Gamma,x\varb{\mu_F}{A}} &=& \locks{\Gamma} \\
\locks{\Gamma,\alpha:K} &=& \locks{\Gamma} \\
\ea\]

\begin{figure}[htbp]
\raggedright
\boxed{\Gamma\vdash A : K\vphantom{\mu}}
\hfill
\begin{mathpar}
\inferrule*
{
  \Gamma \ni \alpha : K
}
{\Gamma \vdash \alpha : K}

\inferrule*
{
  \Gamma \vdash A : \Pure
}
{\Gamma \vdash A : \Any}

\inferrule*
{
  \Gamma \vdash \aeq{E} \\
  \Gamma \vdash A : \Any \\
}
{\Gamma \vdash \boxwith{\aeq{E}} A : \Pure}

\inferrule*
{
  \Gamma \vdash \adj{L}{D} \\
  \Gamma \vdash A : K \\
}
{\Gamma \vdash \boxwith{\adj{L}{D}} A : K}

\inferrule*
{
  \Gamma \vdash A : \Any \\\\
  \Gamma \vdash B : \Any
}
{\Gamma \vdash A\to B : \Any}

\inferrule*
{
  \Gamma, \alpha:K \vdash A : K'
}
{\Gamma \vdash \forall\alpha^{K}.A : K'}

\inferrule*
{ }
{\Gamma \vdash \TUnit : \Pure}

\inferrule*
{
  \Gamma \vdash A : K \\\\
  \Gamma \vdash B : K \\
}
{\Gamma \vdash \Pair{A}{B} : K}

\inferrule*
{
  \Gamma \vdash A : K \\\\
  \Gamma \vdash B : K \\
}
{\Gamma \vdash A + B : K}
\end{mathpar}
\raggedright
\boxed{\Gamma\vdash \mu}
\boxed{\Gamma\vdash E : K \vphantom{\mu}}
\boxed{\Gamma\vdash L\vphantom{\mu}}
\boxed{\Gamma\vdash D\vphantom{\mu}}
\boxed{\Gamma\vdash P\vphantom{\mu}}
\hfill
\begin{mathpar}
\inferrule*
{
  \Gamma \vdash L \\
  \Gamma \vdash D \\
}
{\Gamma \vdash \adj{L}{D}}

\inferrule*
{
  \Gamma \vdash E : \Effect
}
{\Gamma \vdash \aeq{E}}

\inferrule*
{ }
{\Gamma \vdash \cdot : \Effect }

\inferrule*
{
  \Gamma \vdash P \\
  \Gamma \vdash E : \Effect \\
}
{\Gamma \vdash \ell: P,E : \Effect}

\inferrule*
{
  \Gamma \vdash E : \Effect \\
  \Gamma \vdash L \\
}
{\Gamma \vdash \eminus{E}{L} : \Effect}

\inferrule*
{ }
{\Gamma \vdash L}

\inferrule*
{ }
{\Gamma \vdash \cdot} %

\inferrule*
{
  \Gamma \vdash P \\
  \Gamma \vdash D \\
}
{\Gamma \vdash \ell: P,D}
\\

\inferrule*
{ }
{\Gamma \vdash \Abs}

\inferrule*
{
  \Gamma\vdash A : \Pure \\
  \Gamma\vdash B : \Pure
}
{\Gamma \vdash A\sto B}
\end{mathpar}
\raggedright
\boxed{\Gamma\vdash (\mu,A)\To\nu \atmode{F}}
\hfill
\begin{mathpar}  
\inferrule*
{
  \Gamma \vdash A:\Pure
}
{\Gamma\vdash (\mu,A)\To\nu \atmode{F}}

\inferrule*
{
  \mu_F\To\nu_F
}
{\Gamma\vdash (\mu,A)\To\nu \atmode{F}}
\end{mathpar}
\raggedright
\boxed{\Gamma \atmode{E}}
\hfill
\begin{mathpar}
\inferrule*
{ }
{\cdot\atmode{E}}

\inferrule*
{
  \Gamma\atmode{F} \\
  \mu_F : E\to F \\
  \Gamma \vdash A : K \\
}
{\Gamma, x\varb{\mu_F}{A} \atmode{F}}

\inferrule*
{
  \Gamma\atmode{E} \\
}
{\Gamma, \alpha:K \atmode{E}}

\inferrule*
{
  \Gamma\atmode{F} \\
  \mind{\mu}{F} : E\to F \\
}
{\Gamma, \lockwith{\mind{\mu}{F}} \atmode{E}}
\end{mathpar}
\caption{Kinding, well-formedness, and auxiliary rules for \Met.}
\label{fig:kinding-meet}
\end{figure}

\begin{figure}[htbp]
\raggedright
\boxed{L\equiv L'\vphantom{\subtype P'\mu}}
\boxed{D\equiv D'\vphantom{\subtype P'\mu}}
\hfill
\begin{mathpar}
\inferrule*
{ }
{\cdot \equiv \cdot}

\inferrule*
{
  L_1\equiv L_2 \\
  L_2\equiv L_3
}
{L_1\equiv L_3}

\inferrule*
{
  L \equiv L'
}
{\ell,L \equiv \ell,L'}

\inferrule*
{
  \ell \neq \ell' \\
  L \equiv L'
}
{\ell,\ell',L \equiv \ell',\ell,L}
\\

\inferrule*
{ }
{ \cdot \equiv \cdot}

\inferrule*
{
   D_1 \equiv D_2 \\
   D_2 \equiv D_3
}
{ D_1 \equiv D_3}

\inferrule*
{
  P \equiv P' \\
  D \equiv D'
}
{ \ell:P,D \equiv \ell:P',D'}

\inferrule*
{
  \ell \neq \ell'
}
{ \ell:P,\ell':P',D \equiv \ell':P',\ell:P,D}
\end{mathpar}
\raggedright
\boxed{E\equiv F\vphantom{\subtype P'\mu}}
\hfill
\begin{mathpar}
\inferrule*
{ }
{ \cdot \equiv \cdot}

\inferrule*
{
   E_1 \equiv E_2 \\
   E_2 \equiv E_3
}
{ E_1 \equiv E_3}

\inferrule*
{
  P \equiv P' \\
  E \equiv E'
}
{ \ell:P,E \equiv \ell:P',E'}

\inferrule*
{
  \ell\neq\ell'
}
{ \ell:P,\ell':P',E \equiv \ell':P',\ell:P,E}
\\

\hl{
\inferrule*
{
}
{E,\ell:\Abs \equiv E}
}

\inferrule*
{ }
{\eminus{\evar}{L} \equiv \eminus{\evar}{L}} %

\inferrule*
{
  E \equiv E' \\
  L \equiv L'
}
{\eminus{E}{L} \equiv \eminus{E'}{L'}}

\inferrule*
{ }
{\eminus{E}{\cdot} \equiv E}

\inferrule*
{ }
{\eminus{\cdot}{L} \equiv \cdot}

\inferrule*
{ }
{\eminus{(\ell:P,E)}{(\ell,L)} \equiv \eminus{E}{L}}

\inferrule*
{ \ell \notin L }
{\eminus{(\ell:P,E)}{L} \equiv \ell:P,\eminus{E}{L}}

\inferrule*
{ }
{\eminus{(\eminus{\evar}{L})}{L'} \equiv \eminus{\evar}{(L,L')}}
\end{mathpar}
\raggedright
\boxed{P\equiv P'\vphantom{\subtype P'\mu}}
\boxed{\mu\equiv \nu\vphantom{\subtype P'}}
\hfill
\begin{mathpar}
\inferrule*
{
  A \equiv A' \\
  B \equiv B'
}
{A\sto B \equiv A'\sto B'}

\inferrule*
{ }
{\Abs \equiv \Abs}

\inferrule*
{
  E \equiv F
}
{\aeq{E} \equiv \aeq{F}}

\inferrule*
{
  L \equiv L' \\ %
  D \equiv D'
}
{\adj{L}{D} \equiv \adj{L'}{D'}}
\end{mathpar}
\raggedright
\boxed{A\equiv B\vphantom{\subtype P'\mu}}
\hfill
\begin{mathpar}
\inferrule*
{ }
{\alpha \equiv \alpha}

\inferrule*
{
  \mu \equiv \nu \\
  A \equiv B
}
{\boxwith{\mu}A \equiv \boxwith{\nu}B}
  
\inferrule*
{
  A \equiv A' \\
  B \equiv B'
}
{A\to B \equiv A'\to B'}

\inferrule*
{
  A \equiv B
}
{\forall\alpha^K.A \equiv \forall\alpha^K.B}

\inferrule*
{
  A \equiv A' \\
  B \equiv B'
}
{\Pair{A}{B} \equiv \Pair{A'}{B'}}

\inferrule*
{
  A \equiv A' \\
  B \equiv B'
}
{A + B \equiv A' + B'}
\end{mathpar}
\raggedright
\boxed{E\subtype F\vphantom{P'}}
\hfill
\begin{mathpar}
\inferrule*
{ }
{ \cdot \subtype E }

\inferrule*
{ E \equiv F }
{ E \subtype F }

\inferrule*
{
  E_1 \equiv \ell:P_1, E_1' \\
  E_2 \equiv \ell:P_2, E_2' \\\\
  P_1 \subtype P_2 \\
  E_1' \subtype E_2'
}
{ E_1 \subtype E_2 }
\end{mathpar}
\raggedright
\boxed{P\subtype P'}
\boxed{D\subtype D'}
\hfill
\begin{mathpar}
\inferrule*
{
}
{P \subtype P}

\inferrule*
{ }
{\Abs \subtype P}

\inferrule*
{ }
{\cdot \subtype \cdot}

\inferrule*
{
  D_1 \equiv \ell:P_1,D_1' \\
  D_2 \equiv \ell:P_2,D_2' \\\\
  P_1 \subtype P_2 \\
  D_1' \subtype D_2'
}
{D_1 \subtype D_2}
\end{mathpar}
\caption{Type equivalence and sub-effecting for \Met.}
\label{fig:equiv-sub-meet}
\end{figure}

\subsection{Typing Rules}
\label{app:typing-meet}

\Cref{fig:typing-meet} gives the typing rules of \Met. We only show
the extended rules with respect to the typing rules of core \Met in
\Cref{fig:typing-met}.

\begin{figure}[htbp]

\raggedright
\boxed{\typm{\Gamma}{M : A}{E}}
\hfill

\begin{mathpar}
\inferrule*[Lab=\tylab{TAbs}]
{
  \typm{\Gamma,\alpha:K}{V:A}{E}
}
{\typm{\Gamma}{\Lambda\alpha^K.V:\forall\alpha^K.A}{E}}

\inferrule*[Lab=\tylab{TApp}]
{
  \typm{\Gamma}{M:\forall\alpha^K.B}{E} \\
  \Gamma\vdash A : K \\
}
{\typm{\Gamma}{M\,A : B[A/\alpha]}{E}}

\inferrule*[Lab=\tylab{Pair}]
{
  \typm{\Gamma}{M:A}{E} \\
  \typm{\Gamma}{N:B}{E} \\
}
{\typm{\Gamma}{(M,N):\Pair{A}{B}}{E}}

\inferrule*[Lab=\tylab{Inl}]
{
  \typm{\Gamma}{M:A}{E} \\
}
{\typm{\Gamma}{\Inl M:A+B}{E}}

\inferrule*[Lab=\tylab{Inr}]
{
  \typm{\Gamma}{M:B}{E} \\
}
{\typm{\Gamma}{\Inr M:A+B}{E}}

\inferrule*[Lab=\tylab{CrispPair}]
{
  \nu_F : E\to F \\
  \typm{\Gamma, \lockwith{\nu_F}}{V:\Pair{A}{B}}{E} \\\\
  \typm{\Gamma, x\varb{\nu_F}{A}, y\varb{\nu_F}{B}}{M:A'}{F}
}
{\typm{\Gamma}{\Casem{\nu} V\Of (x,y)\mapsto M : A'}{F}}
\quad
\inferrule*[Lab=\tylab{CrispSum}]
{
  \nu_F : E\to F \\
  \typm{\Gamma, \lockwith{\nu_F}}{V:A+B}{E} \\\\
  \typm{\Gamma, x\varb{\nu_F}{A}}{M_1:A'}{F} \\
  \typm{\Gamma, y\varb{\nu_F}{B}}{M_2:A'}{F}
}
{\typm{\Gamma}{\Casem{\nu}V\Of \{\Inl x \mapsto M_1, \Inr y\mapsto M_2\} : A'}{F}}

{
\inferrule*[Lab=\tylab{Handler${}^\all$}]
{
  {D = \{\ell_i : A_i \sto B_i\}_i} \\
  \typm{\Gamma,\lockwith{\aeq{D+E}_F}}{M: A}{D+E} \\
  \typm{\Gamma, \lockwith{\aeq{E}_F}, x : \boxwith{\aeq{D+E}} A}{N : B}{E}\\
  [\typm{\Gamma, \lockwith{\aeq{E}_F}, p_i : A_i, r_i :
    \boxwith{\aeq{E}}
    (B_i \to B)}{N_i : B}{E}]_i
  \\
  [E]_F \To \one_F
}
{\typm{\Gamma}{\Handle^\all\;M\With
  \{\Ret x \mapsto N\} \uplus \{ \ell_i\;p_i\;r_i \mapsto N_i \}_i : B}{F}}
}

{
\inferrule*[Lab=\tylab{ShallowHandler}]
{
  {D = \{\ell_i : A_i \sto B_i\}_i} \\
  \typm{\Gamma,\lockwith{\aex{D}}}{M:A}{D+F} \\\\
  \typm{\Gamma, x : \boxwith{\aex{D}} A}{N : B}{F} \\
  [\typm{\Gamma, p_i : A_i, {r_i}:\boxwith{\aex{D}}(B_i \to A)}{N_i : B}{F}]_i
}
{\typm{\Gamma}{\Handle^\dagger\;M\With
  \{\Ret x \mapsto N\} \uplus \{ \ell_i\;p_i\;r_i \mapsto N_i \}_i : B}{F}}
}

{
\inferrule*[Lab=\tylab{ShallowHandler${}^\all$}]
{
  {D = \{\ell_i : A_i \sto B_i\}_i} \\
  \typm{\Gamma,\lockwith{\aeq{D+E}_F}}{M:A}{D+E} \\\\
  \typm{\Gamma, x : \boxwith{\aeq{D+E}} A}{N : B}{F}\\
  [\typm{\Gamma, p_i : A_i, r_i : \boxwith{[D+E]}(B_i \to A)}{N_i : B}{F}]_i
  \\
  [E]_F \To \one_F
}
{\typm{\Gamma}{\Handle^{\all\dagger}\;M\With
  \{\Ret x \mapsto N\} \uplus \{ \ell_i\;p_i\;r_i \mapsto N_i \}_i : B}{F}}
}

\inferrule*[Lab=\tylab{ModAbs}]
{\typm{\Gamma,\lockwith{\aeq{}_F}}{M:A}{\cdot}}
{\typm{\Gamma}{\Box_{\aeq{}}\,M : \boxwith{\aeq{}} A}{F}}

\inferrule*[Lab=\tylab{Letmod'}]
{
  \nu_F : E\to F \\
  \typm{\Gamma,\lockwith{\nu_F},\ol{\alpha:K}}{V : \boxwith{\mu} A}{E} \\\\
  \typm{\Gamma,x\varb{\nu_F\circ\mu_E}{\forall\ol{\alpha^K}.A}}{M:B}{F}
}
{\typm{\Gamma}{\Letm{\nu}{\mu} \Lambda\ol{\alpha^K} . x = V \In M : B}{F}}
\end{mathpar}

\caption{Typing rules for \Met (only showing extensions to core \Met in \Cref{fig:typing-met}).}
\label{fig:typing-meet}
\end{figure}

\subsection{Operational Semantics}

As type application are treated as values and can reduce, we first
define value normal forms $U$ that cannot reduce further as follows.
\begin{syntax}
  \slab{Value normal forms}         &U& ::= &
  x
  \mid \lambda x^A.M
  \mid \Lambda\alpha^K.V
  \mid \Box_\mu\,U
  \mid (U_1,U_2) \mid \Inl U \mid \Inr U
  \\
\end{syntax}

Consequently, the definition for normal forms in \Cref{sec:metatheory}
is updated as follows.

\begin{definition}[Normal Forms]
  We say a term $M$ is in a normal form with respect to effect type
  $E$, if it is either in value normal form $M = U$ or of form $M =
  \EC[\Do\ell\;U]$ for $\ell\in E$ and $\free{n}{\ell}{\EC}$.
\end{definition}

We also extend the definition of evaluation contexts. The full
definition is given as follows.
Notice that we use value normal forms instead of values.
\begin{syntax}
  \slab{Evaluation contexts} &  \EC &::= & [~]
    \mid \EC\;A\mid \EC\; N \mid U\;\EC \mid
    \Do\ell\;\EC \mid \Mask_L\;\EC \mid \Handle^\delta\; \EC \With H \\
  & & \mid & \Mod_\mu\,\EC
    \mid \Letm{\nu}{\mu} x = \EC \In M
    \mid \Letm{\nu}{\mu} \Lambda \ol{\alpha^K}. x = \EC \In M
    \\
  & & \mid & (\EC,N) \mid (U,\EC) 
    \mid \Casem{\nu}\EC\Of (x, y) \mapsto M \\
  & & \mid & \Inl \EC \mid \Inr \EC
    \mid \Casem{\nu}\EC\Of \{\Inl x \mapsto M, \Inr y \mapsto N\} \\
\end{syntax}

\Cref{fig:semantics-meet} shows the operational semantics of \Met.

\begin{figure}

\begin{reductions}
\semlab{App}   & (\lambda x^A.M)\,U &\reducesto& M[U/x] \\
\semlab{TApp} & (\Lambda \alpha.V)\, A &\reducesto& V[A/\alpha] \\
\semlab{Letmod}  & \Letm{\nu}{\mu} x = \Box_\mu\,U \In M &\reducesto& M[U/x] \\
\semlab{Letmod'}  & \Letm{\nu}{\mu} \Lambda\ol{\alpha^K} . x = \Box_\mu\,U \In M
    &\reducesto& M[(\Lambda\ol{\alpha^K}.U)/x] \\
\semlab{Mask} & \Mask_L\, U &\reducesto& \Box_{\amk{L}}\, U \\
{\semlab{Pair}} &
  {\Casey_\mu\; (U_1,U_2) \Of (x,y)\mapsto N} &{\reducesto}& {N[U_1/x,U_2/y]} \\
{\semlab{Inl}} &
  {\Casey_\mu\; \Inl U \Of \{\Inl x \mapsto N_1,\cdots\}} &{\reducesto}& {N_1[U/x]} \\
{\semlab{Inr}} &
  {\Casey_\mu\; \Inr U \Of \{\Inr y \mapsto N_2,\cdots\}} &{\reducesto}& {N_2[U/y]} \\
\semlab{Ret} &
  \Handle\; U \With H &\reducesto& N[(\Box_{\aex{D}}\,U)/x], \\
\multicolumn{4}{@{}r@{}}{
  \text{where } (\Ret x \mapsto N) \in H
}
\\
\semlab{Op} &
  \Handle\; \EC[\Do\ell \; U] \With H
    &\reducesto& N[U/p, (\lambda y.\Handle\; \EC[y] \With H)/r],\\
\multicolumn{4}{@{}r@{}}{
      \text{ where } \free{0}{\ell}{\EC} \text{ and } (\ell \; p \; r \mapsto N) \in H
} \\
\semlab{Ret}^\all &
  \Handle\; U \With H &\reducesto& N[(\Box_{\aeq{D+E}}\,U)/x] \\
\multicolumn{4}{@{}r@{}}{
  \text{ where } (\Ret x \mapsto N) \in H
} \\
\semlab{Op}^\all &
  \Handle^\all\; \EC[\Do\ell \; U] \With H
    &\reducesto& \\
    & \span \span N[U/p, (\Box_{\aeq{E}}\,(\lambda y.\Handle^\all\; \EC[y] \With H))/r]\\
\multicolumn{4}{@{}r@{}}{
      \text{ where } \free{0}{\ell}{\EC} \text{ and } (\ell \; p \; r \mapsto N) \in H
} \\
\semlab{Ret}^\dagger &
  \Handle^\dagger\; U \With H &\reducesto& N[(\Box_{\aex{D}}\,U)/x] \\ \multicolumn{4}{@{}r@{}}{
  \text{ where } (\Ret x \mapsto N) \in H
} \\
\semlab{Op}^\dagger &
  \Handle^\dagger\; \EC[\Do\ell \; U] \With H
    &\reducesto&
    N[U/p, (\lambda y.\EC[y])/r]\\
\multicolumn{4}{@{}r@{}}{
      \text{ where } \free{0}{\ell}{\EC} \text{ and } (\ell \; p \; r \mapsto N) \in H
} \\
\semlab{Ret}^{\all\dagger} &
  \Handle^{\all\dagger}\; U \With H &\reducesto& N[(\Box_{\aeq{D+E}}\,U)/x] \\
\multicolumn{4}{@{}r@{}}{
  \text{ where } (\Ret x \mapsto N) \in H
} \\
\semlab{Op}^{\all\dagger} &
  \Handle^{\all\dagger}\; \EC[\Do\ell \; U] \With H
    &\reducesto&
    N[U/p, (\Box_{\aeq{D+E}}\,(\lambda y.\EC[y]))/r]\\
\multicolumn{4}{@{}r@{}}{
      \text{ where } \free{0}{\ell}{\EC} \text{ and } (\ell \; p \; r \mapsto N) \in H
} \\
\semlab{Lift} &
  \EC[M] &\reducesto& \EC[N],  \hfill\text{if } M \reducesto N \\
\end{reductions}

\caption{Operational semantics for \Met.}
\label{fig:semantics-meet}
\end{figure}

\FloatBarrier
\section{Meta Theory and Proofs for \Met}
\label{app:CalcM}

We provide meta theory and proofs for \Met in
\Cref{sec:core-calculus,sec:extensions} including all extensions.

\FloatBarrier
\subsection{The Double Category of Effects}
\label{app:double-category-effects}

\begin{figure}[htbp]
\begin{minipage}{0.45\textwidth}
\centering
\begin{tikzcd}
  E \arrow[r, bend left, "\mu_F", ""{name=U,inner sep=1pt,below}]
  \arrow[r, bend right, "\nu_F"{below}, ""{name=D,inner sep=1pt}]
  & F
  \arrow[Rightarrow, from=U, to=D]
\end{tikzcd}
\end{minipage}
\begin{minipage}{0.45\textwidth}
\centering
\begin{tikzcd}
  E \arrow[rr, "\mu_F",  ""{name=U,inner sep=1pt}] \arrow[d, "\subtype"] &  & F \arrow[d, "\subtype"] \\
  E' \arrow[rr, "\nu_{F'}"{below}, ""{name=D,inner sep=1pt}]          &  & F'
  \arrow[Rightarrow, from=U, to=D, shorten <=3pt]
\end{tikzcd}
\end{minipage}
\caption{2-cells in a 2-category compared to 2-cells in a double category.}
\label{fig:two-cells}
\end{figure}

A double category extends a 2-category with an additional kind of morphisms.
Alongside the regular morphisms, now called \textit{horizontal} morphisms,
there are also \textit{vertical} morphisms that connect the objects of the 2-category.
This makes it possible to generalise the 2-cells to transform arbitrary morphisms,
whose source and target are connected by vertical morphisms.
\Cref{fig:two-cells} shows the differences between 2-cells in a
2-category and those in a double category using syntax of \CalcM.

In \CalcM, objects/modes are given by effect contexts, the
horizontal morphisms by modalities, the vertical morphisms by the
sub-effecting relation, and 2-cells by the modality transformations.

Now we show that it indeed has the structure of a double category.

Since the sub-effecting relation is a preorder, effect contexts
(objects) $E$ and sub-effecting (vertical morphisms) $E\subtype F$
obviously form a category given by the poset.

We repeat the definition of modalities and modality composition
from \Cref{sec:modalities} here for easy reference.
We define them directly in terms of morphisms between modes.

\[\ba{rcr@{\ \ }c@{\ \ }l}
\aeq{E}_F &:& E    &\to& F   \\
\adj{L}{D}_F &:& D+(F-L) &\to& F \\
\ea\]

\[\ba{rclcll}
\aeq{E'}_F&\circ&\aeq{E}_{E'} &=& \aeq{E}_F
\\
\adj{L}{D}_F&\circ&\aeq{E}_{D+(F-L)} &=& \aeq{E}_F
\\
\aeq{E}_F&\circ&\adj{L}{D}_E &=& \aeq{D+(E-L)}_F
\\
\adj{L_1}{D_1}_F&\circ&\adj{L_2}{D_2}_{D_1+(F-L_1)} &=&
  \adj{L_1+L}{D_2+D}_F
  &\text{ where } (L,D) = L_2 \bowtie D_1
\\
\ea\]

The effect contexts (objects) and modalities (horizontal morphisms)
also form a category since modality composition possesses
associativity and identity.
We have the following lemma.

\begin{restatable}[Modes and modalities form a category]{lemma}{modCat}
  \label{lemma:mod-cat}
  Modes and modalities form a category with the identity morphism
  $\one_E = \aid{}_E : E\to E$ and the morphism composition
  $\mu_F\circ\nu_{F'}$ such that
  \begin{enumerate}
    \item Identity: $\one_F\circ\mind{\mu}{F} = \mind{\mu}{F} =
    \mind{\mu}{F}\circ\one_E$ for $\mind{\mu}{F}:E\to F$.
    \item Associativity: $(\mu_{E_1}\circ\nu_{E_2})\circ\xi_{E_3} = \mu_{E_1}\circ(\nu_{E_2}\circ\xi_{E_3})$
    for $\mind{\mu}{E_1}:E_2\to E_1$, $\mind{\nu}{E_2}:E_3\to E_2$, and $\mind{\xi}{E_3}:E\to E_3$.
  \end{enumerate}
\end{restatable}
\begin{proof}
  By inlining the definitions of modalities and checking each case.
\end{proof}

In \Cref{sec:core-calculus}, we only define the modality
transformations of shape $\mu_F\To\nu_F$ where the targets of $\mu$
and $\nu$ are required to be the same effect context $F$.
This is enough for presenting the calculus, but we can further extend
it to allow $\mu_F\To\nu_{F'}$ where $F\subtype F'$.
This is used in the meta theory for \CalcM such as the
lock weakening lemma (\Cref{lemma:structural-rules}.3).

The extended modality transformation relation is defined by the
transitive closure of the following rules.
Compared to the definition in \Cref{sec:modalities}, the only new rule is \mtylab{Mono}.
\begin{mathpar}
  \inferrule*[Lab=\mtylab{Abs}]
  {
    {\mu}_{F} : E' \to F \\\\
    E\subtype E' \\
  }
  {\aeq{E}_F\To {\mu}_{F}}

  \inferrule*[Lab=\mtylab{Upcast}]
  {
    D \subtype D'
  }
  {\adj{L}{D}_F \To \adj{L}{D'}_F}

  \inferrule*[Lab=\mtylab{Expand}]
  {
    (F-L) \equiv \ell:A\sto B,E
  }
  {\adj{L}{D}_F \Rightarrow \adj{\ell,L}{D,\ell:A\sto B}_F}

  \inferrule*[Lab=\mtylab{Shrink}]
  {
    (F-L) \equiv \ell:P,E
  }
  {\adj{\ell,L}{D,\ell:P}_F \Rightarrow \adj{L}{D}_F}

  \hl{
  \inferrule*[Lab=\mtylab{Mono}]
  {
    F \subtype F' \\
  }
  {\mu_F \To \mu_{F'}}
  }
\end{mathpar}

The following lemmas shows that the transformation $\mu_F\To\nu_{F'}$
satisfies the requirement of being 2-cells in the double category of
effects with well-defined vertical and horizontal composition.

\begin{restatable}[Modality transformations are 2-cells]{lemma}{modTransTwoCells}
  \label{lemma:modtrans-two-cells}
  If $\mu_F \To \nu_{F'}$, $\mu_F:E\to F$, and $\nu_{F'}:E'\to F'$,
  then $E\subtype E'$ and $F\subtype F'$.
  Moreover, the transformation relation is closed under vertical and
  horizontal composition as shown by the following admissible rules.
  \begin{mathpar}
  \inferrule*
  {
    \mu_{F_1} \To \nu_{F_2} \\
    \nu_{F_2} \To \xi_{F_3}
  }
  {\mu_{F_1} \To \xi_{F_3}}

  \inferrule*
  {
    \mu_F\To\mu'_{F'} \\
    \nu_E\To\nu'_{E'} \\
    \mu_F : E \to F \\
    \mu'_{F'} : E' \to F' \\
  }
  {\mu_F\circ\nu_E \To \mu'_{F'}\circ\nu'_{E'}}
  \end{mathpar}
\end{restatable}
\begin{proof}
We take the transitive closure of the modality transformation rules.
\begin{mathpar}
  \inferrule*[Lab=\mtylab{Abs}]
  {
    {\mu}_{F'} : E' \to F' \\
    E\subtype E' \\
    F\subtype F' \\
  }
  {\aeq{E}_F\To {\mu}_{F'}}

  \inferrule*[Lab=\mtylab{Rel}]
  {
    D\subtype D' \\
    (F'-L) \equiv D_1,E_1 \equiv D_2,E_2 \\
    F\subtype F' \\
    L_1 = \dom{D_1} \\
    L_2 = \dom{D_2} \\
    \meta{present}(D_2) \\
  }
  {\adj{L_1,L}{D,D_1}_F \To \adj{L_2,L}{D',D_2}_{F'}}
\end{mathpar}

  The predicate $\meta{present}(D)$ checks if all labels in $D$ are present.
  Vertical composition follows directly from the fact that we take the
  transitive closure.
  Horizontal compositions follows from a case analysis on shapes of
  modalities being composed.
  The most nontrivial case is when all of $\mu_F$, $\nu_E$,
  $\mu'_{F'}$, and $\nu'_{E'}$ are relative modalities.
  The key observation is that \mtylab{Rel} always expand or shrink the
  mask and extension of a relative modality simultaneously.
\end{proof}

\paragraph{More on Relationships between \CalcM and Multimodal Type Theory}
In addition to extending to a double category, \CalcM also differs
from MTT in the usage of morphism families.
In types and terms we use $\mu$, indexed families of morphisms between
modes, instead of concrete morphisms $\mu_F$.
We do not lose any information. Given a typing judgement
$\typm{\Gamma}{M:A}{E}$, the indexes for all modalities in $M$ and $A$
are determined by $E$.
Similarly, given a variable binding $x\varb{\mu_F}{A}$, the indexes of
all modalities in $A$ are determined by $\mu_F$.

Using indexed families of modalities in types and terms is very useful
to allow term variables to be used flexibly in different effect
contexts larger than where they are defined.
This greatly simplifies the support of subeffecting; we do not need to
update all indexes of modalities in a term or type when upcasting this
term or type to a larger effect context.
As a result, every type is always well-defined at any modes, which
means that we do not need to define the well-formedness judgement
$A\atmode{E}$ as in MTT.
Moreover, one important benefit of having types well-defined at any
modes is that when adding polymorphism for values, type quantifiers do not
need to carry the additional information about the modes at which the
type variables can be used, greatly simplifying the type system.
Otherwise, polymorphic types would need to have forms
$\forall\alpha^{K\atmode{E}}.A$, where $E$ indicates the mode of
the type variable $\alpha$.

In contexts, we still keep concrete morphisms $\mu_F$, which makes the
proof trees of terms much more structured than using morphism
families.

\subsection{Lemmas for Modes and Modalities}
\label{app:lemmas-modes}

Beyond the structure and properties of double categories shown in
\Cref{app:double-category-effects}, we have some extra properties on
modes and modalities in \Met.

The most important one is that horizontal morphisms (sub-effecting)
act functorially on vertical ones (modalities). In other words, the
action of $\mu$ on effect contexts gives a total monotone function.

\begin{restatable}[Monotone modalities]{lemma}{monoModality}
  \label{lemma:mono-mod}
  If $\mind{\mu}{F}:E\to F$ and $F\subtype F'$, then
  $\mind{\mu}{F'}:E'\to F'$ with $E\subtype E'$.
\end{restatable}
\begin{proof}
  By definition.
\end{proof}

We prove the lemma on the equivalence between syntactic and semantic
definition of modality transformation in \Cref{sec:modalities}.
This lemma can be generalised to the general form of 2-cells in a
double category $\mu_F\To\nu_{F'}$ where $F\subtype F'$.

\semanticModTrans*
\begin{proof}
  From left to right, it is obvious that the semantics is preserved
  after taking the transitive closure. We only need to show the
  transformation given by each rule satisfies the semantics.
  \begin{description}
    \item[Case] \mtylab{Abs}. Follow from \Cref{lemma:mono-mod}.
    \item[Case] \mtylab{Upcast}. Since $D\subtype D'$, we have
    $D+(F-L) \subtype D'+(F-L)$ for any $F$.
    \item[Case] \mtylab{Expand}. Since $(F-L)\equiv\ell:A\sto B,E$, for any
    $F\subtype F'$ we have $(F'-L)\equiv\ell:A\sto B,E'$ for some $E'$.
    Both sides act on $F'$ give $D,\ell:A\sto B,E'$.
    Notice that it is important for $\ell$ to not be absent here;
    otherwise, in $F'$ we could upcast the absent type of $\ell$
    to any concrete operation arrows, which then breaks the condition
    $\mu(F')\subtype \nu(F')$.
    \item[Case] \mtylab{Shrink}. Similar to the above case.
  \end{description}

  From right to left, we need to show that for all pairs $\mu_F$ and
  $\nu_F$ satisfying the semantic definition, we have $\mu_F\To\nu_F$
  in the transitive closure of the syntactic rules.
  This obviously holds for those transformation starting from absolute
  modalities.
  For those transformation starting from relative modalities, observe
  that they can only be transformed to other relative modalities
  according to the semantic definition.
  By taking the transitive closure of the transformation rules for
  relative modalities, we have
  \begin{mathpar}
    \inferrule*[Lab=\mtylab{Rel}]
    {
      D\subtype D' \\
      (F-L) \equiv D_1,E_1 \equiv D_2,E_2 \\\\
      L_1 = \dom{D_1} \\
      L_2 = \dom{D_2} \\
      \meta{present}(D_2) \\
    }
    {\adj{L_1,L}{D,D_1}_F \To \adj{L_2,L}{D',D_2}_{F}}
  \end{mathpar}
  The predicate $\meta{present}(D)$ checks if all labels in $D$ are present.
  Suppose $\adj{L_3}{D_3}_F$ and $\adj{L_4}{D_4}_F$ satisfies that
  {$D_3+(F'-L_3)\subtype D_4+(F'-L_4)$} for all $F\subtype F'$.
  We need to show that it is generated by \mtylab{Rel}.
  The key observation is that for $\adj{L_3}{D_3}_F$ and
  $\adj{L_4}{D_4}_F$ to satisfy the semantic definition, we must have
  $L_3' \equiv L_4'$ and $D_3' \subtype D_4'$ for $L_3\bowtie D_3 =
  (L_3',D_3')$ and $L_4\bowtie D_4 = (L_4', D_4')$.
  Otherwise, we can always construct counterexamples by choosing a
  appropriate $F'$.
  Also notice that when $L_3$ and $L_4$ ($D_3$ and $D_4$) are
  different, $F$ should provide enough labels to fill the gap.
  Moreover, for all those labels in $D_4$ but not in $D_3$, they
  should be present to be stable under any $F'$ with $F\subtype F'$.
  It is not hard to verify that \mtylab{Rel} covers all such pair of
  relative modalities.

\end{proof}

Our proofs for type soundness and effect safety do not use ad-hoc case
analysis on shapes of modalities or rely on any specific properties
about the definition of composition and transformation
(except for the parts about effect handlers since they specify the
required modalities in the typing rules).
As a result, it should be possible to generalise our calculus and proofs
to other mode theories satisfying certain extra properties.
We state some properties of the mode theory as the following lemmas
for easier references in proofs. Most of them directly follow from the
definition.

\begin{restatable}[Vertical composition]{lemma}{modtransVertical}
  \label{lemma:modtrans-vertical}
  If $\mu_{F_1}\To\nu_{F_2}$ and $\nu_{F_2}\To\xi_{F_3}$, then $\mu_{F_1}\To\xi_{F_3}$.
\end{restatable}
\begin{proof}
  Follow from \Cref{lemma:modtrans-two-cells}
\end{proof}

\begin{restatable}[Horizontal composition]{lemma}{modtransHorizontal}
  \label{lemma:modtrans-horizontal}
  If $\mind{\mu}{F}:E\to F$, $\mind{\mu'}{F'}:E'\to F'$, $\mind{\mu}{F}\To\mind{\mu'}{F'}$, and
  $\mind{\nu}{E}\To\nu'_{E'}$, then $\mind{\mu}{F}\circ\mind{\nu}{E}\To\mind{\mu'}{F'}\circ\mind{\nu'}{E'}$.
\end{restatable}
\begin{proof}
  Follow from \Cref{lemma:modtrans-two-cells}
\end{proof}

\begin{lemma}[Monotone modality transformation]
  \label{lemma:mono-modtrans}
  If $\mu_F \To \nu_F$ and $F\subtype F'$,
  then $\mu_{F'}\To \nu_{F'}$.
\end{lemma}
\begin{proof}
  Follow from \Cref{lemma:semantic-modtrans}
\end{proof}

\begin{lemma}[Asymmetric reflexivity of modality transformation]
  \label{lemma:self-modtrans}
  If $F\subtype F'$ and $\mu_F:E\to F$,
  then $\mu_F\To\mu_{F'}$.
\end{lemma}
\begin{proof}
  By definition.
\end{proof}

\subsection{Lemmas for the Calculus}
\label{app:CalcM-lemmas}

We prove structural and substitution lemmas for \Met as well as some
other auxiliary lemmas for proving type soundness.

\begin{lemma}[Canonical forms]~
  \label{lemma:canonical-forms}
  \begin{enumerate}[label=\arabic*.]
    \item If $\typm{\,}{U:\boxwith{\mu}A}{E}$, then $U$ is of shape $\Box_\mu\,U'$.
    \item If $\typm{\,}{U:A\to B}{E}$, then $U$ is of shape $\lambda x^A.M$.
    \item If $\typm{\,}{U:\forall\alpha. A}{E}$, then $U$ is of shape $\Lambda \alpha.V$.
    \item If $\typm{\,}{U:(A,B)}{E}$, then $U$ is of shape $(U_1,U_2)$.
    \item If $\typm{\,}{U:A+B}{E}$, then $U$ is either of shape $\Inl
    U'$ or of shape $\Inr U'$.
  \end{enumerate}
\end{lemma}
\begin{proof}
  Directly follows from the typing rules.
\end{proof}

In order to define the lock weakening lemma, we first define a context
update operation $\updlock{\Gamma}{F'}$ which gives a new context
derived from updating the indexes of all locks and variable bindings
in $\Gamma$ such that $\locks{\updlock{\Gamma}{F'}} : \_ \to F'$.

\[\ba{rcl}
\updlock{\cdot}{F} &=& \cdot \\
\updlock{\lockwith{\mind{\aeq{E}}{F'}},\Gamma'}{F} &=& \lockwith{\mind{\aeq{E}}{F}},\Gamma' \\
\updlock{\lockwith{\mind{\adj{L}{D}}{F'}},\Gamma'}{F} &=& \lockwith{\mind{\adj{L}{D}}{F}},\updlock{\Gamma'}{D+(F-L)} \\
\updlock{x\varb{\mu_{F'}}{A},\Gamma'}{F} &=& x\varb{\mu_F}{A},\updlock{\Gamma'}{F} \\
\updlock{\alpha:K,\Gamma'}{F} &=& \alpha:K,\updlock{\Gamma'}{F} \\
\ea\]

We have the following lemma showing that the index update operation
preserves the $\locks{-}$ operation except for updating the index.

\begin{lemma}[Index update preserves composition]
  \label{lemma:index-update}
  If $\mu_F = \locks{\Gamma} : E\to F$, $F \subtype F'$, and
  $\locks{\updlock{\Gamma}{F'}}:E'\to F'$, then $\locks{\updlock{\Gamma}{F'}} = \mu_{F'}$.
\end{lemma}
\begin{proof}
  By straightforward induction on the context and using the property
  that $(\mu\circ\nu)_F = \mu_F\circ\nu_E$ for $\mu_F : E\to F$.
\end{proof}

\begin{corollary}[Index update preserves transformation]
  \label{lemma:update-modtrans}
  If $\locks{\Gamma}:E\to F$, $F \subtype F'$, and
  $\locks{\updlock{\Gamma}{F'}}:E'\to F'$, then $\locks{\Gamma}\To
  \locks{\updlock{\Gamma}{F'}}$.
\end{corollary}
\begin{proof}
  Immediately follow from \Cref{lemma:index-update} and
  \Cref{lemma:self-modtrans}.
\end{proof}

We have the following structural lemmas.

\begin{restatable}[Structural rules]{lemma}{structuralRules} ~
  \label{lemma:structural-rules}
  The following structural rules are admissible.
  \begin{enumerate}[label=\arabic*.]
    \item Variable weakening.
    \begin{mathpar}
      \inferrule*
      {
        \typm{\Gamma,\Gamma'}{M:B}{E} \\
        \Gamma,x\varb{\mu_F}{A},\Gamma'\atmode{E}
      }
      {\typm{\Gamma,x\varb{\mu_F}{A},\Gamma'}{M:B}{E}}
    \end{mathpar}
    \item Variable swapping.
    \begin{mathpar}
      \inferrule*
      {
        \typm{\Gamma,x\varb{\mu_F}{A},y\varb{\nu_F}{B},\Gamma'}{M:A'}{E}
      }
      {\typm{\Gamma,y\varb{\nu_F}{B},x\varb{\mu_F}{A},\Gamma'}{M:A'}{E}}
    \end{mathpar}
    \item Lock weakening.
    \begin{mathpar}
      \inferrule*
      {
        \typm{\Gamma,\lockwith{\mind{\mu}{F}},\Gamma'}{M:A}{E} \\
        \mu_F \To \nu_F \\
        \nu_F : F'\to F \\
        \locks{\updlock{\Gamma'}{F'}}:E'\to F' \\
      }
      {\typm{\Gamma,\lockwith{\mind{\nu}{F}},\updlock{\Gamma'}{F'}}{M:A}{E'}}
    \end{mathpar}
    \item Type variable weakening.
    \begin{mathpar}
      \inferrule*
      {
        \typm{\Gamma,\Gamma'}{M:B}{E}
      }
      {\typm{\Gamma,\alpha:K,\Gamma'}{M:B}{E}}
    \end{mathpar}
    \item Type variable swapping.
    \begin{mathpar}
      \inferrule*
      {
        \typm{\Gamma_1,\Gamma_2,\alpha:K,\Gamma_3}{M:A}{E}
      }
      {\typm{\Gamma_1,\alpha:K,\Gamma_3}{M:A}{E}}

      \inferrule*
      {
        \alpha\notin\ftv{\Gamma_2} \\
        {\typm{\Gamma_1,\alpha:K,\Gamma_3}{M:A}{E}}
      }
      {\typm{\Gamma_1,\Gamma_2,\alpha:K,\Gamma_3}{M:A}{E}}
    \end{mathpar}
  \end{enumerate}
\end{restatable}
\begin{proof}
  1, 2, 4, and 5 follow from straightforward induction on the typing derivation.
  For 3, we also proceed by induction on the typing derivation.
  The most interesting case is \tylab{Var}.
  Other cases mostly follow from IHs.
\begin{description}
\item[Case]
  \begin{mathpar}
  \inferrule*[Lab=\tylab{Var}]
  {
    \nu'_{F_1}= \locks{\Gamma_2} : E\to F_1 \\
    \refa{\mu'_{F_1}\To\nu'_{F_1}} \text{ or } \Gamma\vdash A:\Pure
  }
  {\typm{\Gamma_1,x\varb{\mu'_{F_1}},\Gamma_2}{x:A}{E}}
  \end{mathpar}
  Trivial when $A$ is pure. Otherwise, case analysis on where the lock
  weakening happens.
  \begin{description}
    \item[Case] $\Gamma$. Supposing $\Gamma_1 =
    \Gamma,\lockwith{\mu_F},\Gamma_0$ and after lock weakening we have
    $\Gamma,\lockwith{\nu_F},\Gamma_0',x\varb{\mu'_{F_1'}},\Gamma_2'$
    where $\Gamma_2' = \updlock{\Gamma_2}{F_1'} : E'\to F_1'$ and
    $\Gamma_0' = \updlock{\Gamma_0}{F'}:F_1'\to F'$.
    By \Cref{lemma:index-update} on $\Gamma_0$, $F\subtype F'$, and
    \Cref{lemma:mono-mod}, we have $F_1\subtype F_1'$.
    Then by \refa{} and \Cref{lemma:mono-modtrans}, we have
    $\mu'_{F_1'}\To\nu'_{F_1'}$.
    Then by \Cref{lemma:index-update} we have $\nu'_{F_1'} =
    \locks{\Gamma_2'}$.
    Finally by \tylab{Var} we have
    \[
      {\typm{\Gamma,\lockwith{\nu_F},\Gamma_0',x\varb{\mu'_{F_1'}},\Gamma_2'}{x:A}{E'}}
    \]
    \item[Case] $\Gamma_2$. Suppose $\Gamma_2 =
    \Gamma_0,\lockwith{\mu_F},\Gamma'$.
    is weakened to $\Gamma_2'=\Gamma_0,\lockwith{\nu_F},\updlock{\Gamma'}{F'}$.
    By \Cref{lemma:update-modtrans} we have
    $\locks{\Gamma'}\To\locks{\updlock{\Gamma'}{F'}}$.
    Then by \Cref{lemma:modtrans-horizontal} we have we have
    $\locks{\Gamma_2} \To \locks{\Gamma_2'}$.
    By \Cref{lemma:modtrans-vertical} and \refa{}, we have $\mu'_{F_1}
    \To \locks{\Gamma_2'}$.
    Finally by \tylab{Var} we have
    \[
      {\typm{\Gamma,x\varb{\mu'_{F_1}},\Gamma_2'}{x:A}{E'}}
    \]
  \end{description}
\item[Case]
  \begin{mathpar}
    \inferrule*[Lab=\tylab{Mod}]
    {
      \mind{\mu'}{E} : F_1 \to E \\
      \refa{\typm{\Gamma,\lockwith{\mind{\mu}{F}},\Gamma',\lockwith{\mind{\mu'}{E}}}{V:A}{F_1}}
    }
    {\typm{\Gamma,\lockwith{\mind{\mu}{F}},\Gamma'}{\Box_{\mu'}\,V : \boxwith{\mu'} A}{E}}
  \end{mathpar}
  We have
  \[
   \updlock{\Gamma',\lockwith{\mind{\mu'}{E}}}{F'}
  = \updlock{\Gamma'}{F'}, \updlock{\lockwith{\mind{\mu'}{E}}}{E'}
  = \updlock{\Gamma'}{F'}, \lockwith{\mind{\mu'}{E'}}.
  \]
  Supposing $\mind{\mu'}{E'}:F_1' \to E'$, by
  $\locks{\updlock{\Gamma'}{F'}, \lockwith{\mind{\mu'}{E'}}}:F_1'\to F'$ and
  IH on \refa{}, we have
  \[
    \typm{\Gamma,\lockwith{\mind{\mu}{F}},\updlock{\Gamma'}{F'}, \lockwith{\mind{\mu'}{E'}}}{V:A}{F_1'}.
  \]
  Then by \tylab{Mod} we have
  \[
  \typm{\Gamma,\lockwith{\mind{\mu}{F}},\updlock{\Gamma'}{F'}}{\Box_{\mu'}\,V : \boxwith{\mu'} A}{E'}.
  \]
  \item[Case]
  \begin{mathpar}
    \inferrule*[Lab=\tylab{Letmod}]
    {
      \nu'_E : F_1\to E \\
      \refa{\typm{\Gamma,\lockwith{\mu_F},\Gamma',\lockwith{\nu'_E}}{V : \boxwith{\mu'} A}{F_1}} \\
      \refb{\typm{\Gamma,\lockwith{\mu_F},\Gamma',x\varb{\nu'_E\circ\mu'_{F_1}}{A}}{M:B}{E}}
    }
    {\typm{\Gamma,\lockwith{\mu_F},\Gamma'}{\Letm{\nu'}{\mu'} x = V \In M : B}{E}}
  \end{mathpar}
  By IH on \refa{}, we have
  \[
    \typm{\Gamma,\lockwith{\nu_F},\updlock{\Gamma'}{F'},\lockwith{\nu'_{E'}}}{V : \boxwith{\mu'} A}{F_1'}
  \]
  where $\nu'_{E'} : F_1' \to E'$.
  By IH on \refb{}, we have
  \[
    \typm{\Gamma,\lockwith{\nu_F},\updlock{\Gamma'}{F'},x\varb{\nu'_{E'}\circ\mu'_{F_1'}}{A}}{M:B}{E'}.
  \]
  Then by \tylab{Letmod}, we have
  \[
    {\typm{\Gamma,\lockwith{\mu_F},\updlock{\Gamma'}{F'}}{\Letm{\nu'}{\mu'} x = V \In M : B}{E'}}
  \]
\item[Case]
\begin{mathpar}
\inferrule*[Lab=\tylab{Letmod'}]
{
  \nu'_E : F_1\to E \\
  \refa{\typm{\Gamma,\lockwith{\mu_F},\Gamma',\lockwith{\nu'_E},\ol{\alpha:K}}{V : \boxwith{\mu'} A}{F_1}} \\
  \refb{\typm{\Gamma,\lockwith{\mu_F},\Gamma',x\varb{\nu'_E\circ\mu'_{F_1}}{\forall\ol{\alpha^K}.A}}{M:B}{E}}
}
{\typm{\Gamma,\lockwith{\mu_F},\Gamma'}{\Letm{\nu'}{\mu'} \Lambda\ol{\alpha^K} . x = V \In M : B}{E}}
\end{mathpar}
Similar to the case for \tylab{Letmod}.
BY IH on \refa{}, we have
\[
  {\typm{\Gamma,\lockwith{\nu_F},\updlock{\Gamma'}{F'},\lockwith{\nu'_{E'}},\ol{\alpha:K}}{V : \boxwith{\mu'} A}{F_1'}}
\]
where $\nu'_{E'} : F_1' \to E'$. By IH on \refb{}, we have
\[
  {\typm{\Gamma,\lockwith{\nu_F},\updlock{\Gamma'}{F'},x\varb{\nu'_{E'}\circ\mu'_{F_1'}}{\forall\ol{\alpha^K}.A}}{M:B}{E'}}.
\]
Then by \tylab{Letmod'}, we have
\[
{\typm{\Gamma,\lockwith{\nu_F},\updlock{\Gamma'}{F'}}{\Letm{\nu'}{\mu'} \Lambda\ol{\alpha^K} . x = V \In M : B}{E'}}
\]
\item[Case] \tylab{TAbs}, \tylab{Abs}, \tylab{TApp}, \tylab{App},
\tylab{Do}, \tylab{Mask}, \tylab{Handler}, \tylab{ModAbs}, other
handlers and data types. Follow from IH. Similar to other cases we
have shown.
\end{description}
\end{proof}

As a corollary of \Cref{lemma:structural-rules}.3, the following
sub-effecting rule is admissible.
\begin{restatable}[Sub-effecting]{corollary}{subEffecting}
  \label{lemma:sub-effecting}
  The following rule is admissible.
  \begin{mathpar}
    \inferrule*
    {
      \typm{\Gamma}{M:A}{E} \\
      \locks{\Gamma} : E\to F \\
      F\subtype F' \\
      \locks{\updlock{\Gamma}{F'}} : E'\to F' \\
    }
    {\typm{\updlock{\Gamma}{F'}}{M:A}{E'}}
  \end{mathpar}
\end{restatable}
\begin{proof}
  Follow from \Cref{lemma:structural-rules}.3 by adding the lock
  $\lockwith{\aeq{F}_\cdot}$ to the left of $\Gamma$ in
  $\typm{\Gamma}{M:A}{E}$, and weaken it to
  $\lockwith{\aeq{F'}_\cdot}$.
  Note that typing judgements still hold after adding a lock to or
  removing a lock from the left of the context, as long as the new
  contexts are still well-defined.
\end{proof}

The following lemma reflects the intuition that pure values can be
used in any effect context.

\begin{lemma}[Pure Promotion]
  \label{lemma:pure-promotion}
  The following promotion rule is admissible.
  \begin{mathpar}
    \inferrule*
    {
      \typm{\Gamma_1,\Gamma}{V:A}{E} \\
      \Gamma_1\vdash A : \Pure \\\\
      \locks{\Gamma}  : E\to F \\
      \locks{\Gamma'} : E'\to F \\
      \fv{V} \cap \dom{\Gamma'} = \emptyset
    }
    {\typm{\Gamma_1,\Gamma'}{V:A}{E'}}
  \end{mathpar}
\end{lemma}
\begin{proof}
By induction on the typing derivation of $V$.
\begin{description}
\item[Case] \tylab{Var}. Trivial.
\item[Case]
\begin{mathpar}
\inferrule*[Lab=\tylab{Mod}]
{
  \mind{\mu}{E} : F_1 \to E \\
  \refa{\typm{\Gamma_1,\Gamma,\lockwith{\mind{\mu}{E}}}{V:A}{F_1}}
}
{\typm{\Gamma_1,\Gamma}{\Box_\mu\,V : \boxwith{\mu} A}{E}}
\end{mathpar}
Case analysis on the shape of $\mu$.
\begin{description}
  \item[Case] $\mu$ is relative. By kinding, $A$ is also pure.
  By IH on \refa{}, we have
  \[
    \typm{\Gamma_1,\Gamma',\lockwith{\mind{\mu}{E'}}}{V:A}{F_1'}
  \]
  where $\mu_{E'}:F_1'\to E'$.
  Then by \tylab{Mod} we have
  \[
    \typm{\Gamma_1,\Gamma'}{\Box_\mu\,V:\boxwith{\mu} A}{E'}
  \]
  \item[Case] $\mu$ is absolute. We have $\mu = \aeq{F_1}$ and
  $\locks{\Gamma',\lockwith{\mu_{E'}}} = \aeq{F_1}_{F} =
  \locks{\Gamma,\lockwith{\mu_E}}$.
  Thus, replacing the context $(\Gamma,\lockwith{\mu_E})$ with
  $(\Gamma',\lockwith{\mu_{E'}})$ in \refa{} does not influence all
  usages of \tylab{Var} in the derivation tree of \refa{}. We have
  \[
    {\typm{\Gamma_1,\Gamma',\lockwith{\mind{\mu}{E'}}}{V:A}{F_1}}
  \]
  Then by \tylab{Mod} we have
  \[
    \typm{\Gamma_1,\Gamma'}{\Box_\mu\,V:\boxwith{\mu} A}{E'}
  \]
\end{description}
\item[Case] \tylab{TAbs}. Follow from IH and
\Cref{lemma:structural-rules}.5.
\item[Case] \tylab{Abs}. Impossible since function types are impure.
\item[Case] Data types. Follow from IHs.
\end{description}
\end{proof}

\begin{restatable}[Substitution]{lemma}{substitution} ~
  \label{lemma:substitution}
  The following substitution rules are admissible.
  \begin{enumerate}[label=\arabic*.]
    \item Preservation of kinds under type substitution.
    \begin{mathpar}
      \inferrule*
      {
        \Gamma\vdash A : K \\
        \Gamma,\alpha:K,\Gamma'\vdash B : K'
      }
      {\Gamma,\Gamma'\vdash B[A/\alpha] : K'}
    \end{mathpar}
    \item Preservation of types under type substitution.
    \begin{mathpar}
      \inferrule*
      {
        \Gamma\vdash A : K \\
        \typm{\Gamma,\alpha:K,\Gamma'}{M:B}{E} \\
      }
      {\typm{\Gamma,\Gamma'}{M[A/\alpha] : B[A/\alpha]}{E}}
    \end{mathpar}
    \item Preservation of types under value substitution.
    \begin{mathpar}
      \inferrule*
      {
        \typm{\Gamma,\lockwith{\mu_F}}{V:A}{F'} \\
        \typm{\Gamma,x\varb{\mu_F}{A},\Gamma'}{M:B}{E} \\
      }
      {\typm{\Gamma,\Gamma'}{M[V/x]:B}{E}}
    \end{mathpar}
  \end{enumerate}
\end{restatable}
\begin{proof} ~ \\
\noindent 1. By straightforward induction on the kinding derivation.

\noindent 2. By straightforward induction on the typing derivation of $M$.

\noindent 3. By induction on the typing derivation of $M$.
Trivial when variable $x$ is not used. In the following induction we
always assume $x$ is used.
\begin{description}
\item[Case]
\begin{mathpar}
\inferrule*[Lab=\tylab{Var}]
{
  \mind{\nu}{F}= \locks{\Gamma'} : E\to F \\
  \refa{\mind{\mu}{F}\To\nu_F} \text{ or } \Gamma\vdash A:\Pure
}
{\typm{\Gamma,x\varb{\mu_F}{A},\Gamma'}{x:A}{E}}
\end{mathpar}
Case analysis on the purity of $A$
\begin{description}
  \item[Case] Impure. By $\typm{\Gamma,\lockwith{\mu_F}}{V:A}{F'}$,
  \refa{}, and \Cref{lemma:structural-rules}.3, we have
  \[
    \typm{\Gamma,\lockwith{\nu_F}}{V:A}{E}.
  \]
  Then, by context equivalence, \Cref{lemma:structural-rules}.1, and
  \Cref{lemma:structural-rules}.4, we have
  \[
    \typm{\Gamma,\Gamma'}{V:A}{E}.
  \]
  \item[Case] Pure. By $\typm{\Gamma,\lockwith{\mu_F}}{V:A}{F'}$ and
  \Cref{lemma:pure-promotion}, we have
  \[
    \typm{\Gamma,\Gamma'}{V:A}{E}.
  \]
\end{description}
\item[Case]
\begin{mathpar}
\inferrule*[Lab=\tylab{Mod}]
{
  \mind{\mu'}{E} : F_1 \to E \\
  \refa{\typm{\Gamma,x\varb{\mu_F}{A},\Gamma',\lockwith{\mind{\mu'}{E}}}{W:B}{F_1}}
}
{\typm{\Gamma,x\varb{\mu_F}{A},\Gamma'}{\Box_{\mu'}\,W : \boxwith{\mu'} B}{E}}
\end{mathpar}
By IH on \refa{} we have
\[
  \typm{\Gamma,\Gamma',\lockwith{\mind{\mu'}{E}}}{W[V/x]:B}{F_1}.
\]
Then by \tylab{Mod} we have
\[
  \typm{\Gamma,\Gamma'}{(\Box_{\mu'}\,W)[V/x] : \boxwith{\mu'} B}{E}
\]
\item[Case]
\begin{mathpar}
\inferrule*[Lab=\tylab{Letmod}]
{
  \nu_E : F_1\to E \\
  \refa{\typm{\Gamma,x\varb{\mu_F}{A},\Gamma',\lockwith{\nu_E}}{W : \boxwith{\mu'} A'}{F_1}} \\
  \refb{\typm{\Gamma,x\varb{\mu_F}{A},\Gamma',y\varb{\nu_E\circ\mu'_{F_1}}{A'}}{M:B}{E}}
}
{\typm{\Gamma,x\varb{\mu_F}{A},\Gamma'}{\Letm{\nu}{\mu'} y = W \In M : B}{E}}
\end{mathpar}
By IH on \refa{}, we have
\[
  \typm{\Gamma,\Gamma',\lockwith{\nu_E}}{W[V/x] : \boxwith{\mu'} A'}{F_1}.
\]
By IH on \refb{}, we have
\[
  \typm{\Gamma,\Gamma',y\varb{\nu_E\circ\mu'_{F_1}}{A'}}{M[V/x]:B}{E}.
\]
Then by \tylab{Letmod}, we have
\[
  \typm{\Gamma,\Gamma'}{(\Letm{\nu}{\mu'} y = W \In M)[V/x] : B}{E}
\]
\item[Case]
\begin{mathpar}
\inferrule*[Lab=\tylab{Letmod'}]
{
  \nu_E : F_1\to E \\
  \refa{\typm{\Gamma,x\varb{\mu_F}{A},\Gamma',\lockwith{\nu_E},\ol{\alpha:K}}{V : \boxwith{\mu'} A'}{F_1}} \\
  \refb{\typm{\Gamma,x\varb{\mu_F}{A},\Gamma',y\varb{\nu_E\circ\mu'_{F_1}}{\forall\ol{\alpha^K}.A'}}{M:B}{E}}
}
{\typm{\Gamma,x\varb{\mu_F}{A},\Gamma'}{\Letm{\nu}{\mu'} \Lambda\ol{\alpha^K} . y = V \In M : B}{E}}
\end{mathpar}
Similar to the case for \tylab{Letmod}. Our goal follows from IH on
\refa{}, IH on \refb{}, and \tylab{Letmod'}.
\item[Case]
\begin{mathpar}
\inferrule*[Lab=\tylab{Mask}]
{
  \refa{\typm{\Gamma,x\varb{\mu_F}{A},\Gamma',\lockwith{\mind{\amk{L}}{E}}}{M: B}{E-L}} \\
}
{\typm{\Gamma,x\varb{\mu_F}{A},\Gamma'}{\Mask_{L}\; M: \boxwith{\amk{L}} B}{E}}
\end{mathpar}
By IH on \refa{} we have
\[
  \typm{\Gamma,\Gamma',\lockwith{\mind{\amk{L}}{E}}}{M[V/x]: B}{E-L}.
\]
Then by \tylab{Mask} we have
\[
  \typm{\Gamma,\Gamma'}{(\Mask_{L}\; M)[V/x]: \boxwith{\amk{L}} B}{E}
\]
\item[Case]
\begin{mathpar}
\inferrule*[Lab=\tylab{Handler}]
{
  {D = \{\ell_i : A_i \sto B_i\}_i} \\
  \refa{\typm{\Gamma, x\varb{\mu_F}{A},\Gamma', \lockwith{\aex{D}_E}}{M : A_0}{D+E}} \\\\
  \refb{\typm{\Gamma, x\varb{\mu_F}{A},\Gamma', y : \boxwith{\aex{D}} A_0}{N : B}{E}} \\
  [\refc{\typm{\Gamma,x\varb{\mu_F}{A},\Gamma', p_i : A_i, {r_i}: B_i \to B}{N_i : B}{E}}]_i
}
{\typm{\Gamma,x\varb{\mu_F}{A},\Gamma'}{\Handle\;M\With
  \{\Ret y \mapsto N\} \uplus \{ \ell_i\;p_i\;r_i \mapsto N_i \}_i: B}{E}}
\end{mathpar}
Follow from IH on \refa{},\refb{},\refc{}, and reapplying
\tylab{Handler}.
\item[Case] \tylab{TAbs}, \tylab{TApp}, \tylab{Abs}, \tylab{App},
\tylab{Do}. Follow from IH.
\item[Case] \tylab{ModAbs}, other handlers and data types. Follow from IH.
\end{description}
\end{proof}

\subsection{Progress}
\label{app:progress}

\progress*
\begin{proof}
By induction on the typing derivation $\typm{}{M:A}{E}$. The most
non-trivial cases are \tylab{Mask} and \tylab{Handler}. Other cases
follow from IHs and reduction rules, using
\Cref{lemma:canonical-forms}.
\begin{description}
\item[Case] $M$ is in a value normal form $U$. Trivial. Base case.
\item[Case] \tylab{Do}. Trivial. Base case. %
\item[Case] \tylab{Mod}. $\Box_\mu\, V$. By IH on $V$.
\item[Case] \tylab{Letmod}. $\Letm{\nu}{\mu} x = V \In N$. By IH on
$V$, if $V$ is reducible then $M$ is reducible; otherwise, $V$ is in a
value normal form, then by \Cref{lemma:canonical-forms} we have that
$M$ is reducible by \semlab{Letmod}.
\item[Case] \tylab{Letmod'}. Similar to the case for \tylab{Letmod}.
\item[Case] \tylab{TApp}. $M\,A$. Similarly by IH on $M$, \Cref{lemma:canonical-forms}, and \semlab{TApp}.
\item[Case] \tylab{App}. $M\,N$. Similarly by IH on $M$ and $N$, \Cref{lemma:canonical-forms}, and \semlab{App}.
\item[Case] \tylab{Mask}. $\Adapt^E\,M$. By IH on $M$.
\begin{description}
  \item[Case] $M$ is reducible. Trivial.
  \item[Case] $M$ is in a value normal form. By \semlab{Mask}.
  \item[Case] $M = \EC[\Do\ell\,U]$ with $\free{n}{\ell}{\EC}$. The
  whole term is in a normal form.
\end{description}
\item[Case] Handlers. The general form is $\Handle^\delta\;M\With H$. By IH on $M$.
\begin{description}
  \item[Case] $M$ is reducible. Trivial.
  \item[Case] $M$ is in a value normal form. By \semlab{Ret}.
  \item[Case] $M = \EC[\Do\ell\;U]$ with $\free{n}{\ell}{\EC}$. If
  $n=0$ and $\ell\in H$, then reducible by \semlab{Op}. Otherwise,
  the whole term is in a normal form.
\end{description}
\item[Case] \tylab{ModAbs}. $\Box_{\aeq{}}\, M$. If $M\reducesto N$,
follow by IH on $M$. Otherwise, $M$ must be in a value normal form
because the \tylab{ModAbs} requires $M$ to have the empty effect. In
this case, $\Box_{\aeq{}}\,M$ is also in a value normal form.
\item[Case] Other handlers and data types. Similar to other cases.
\end{description}
\end{proof}

\subsection{Subject Reduction}
\label{app:subject-reduction}

\subjectReduction*
\begin{proof}
By induction on the typing derivation $\typm{\Gamma}{M:A}{E}$.
\begin{description}
\item[Case] \tylab{Var}. Impossible as there is no further reduction.
\item[Case]
\begin{mathpar}
  \inferrule*[Lab=\tylab{Mod}]
  {
    \mind{\mu}{F} : E \to F \\
    \refa{\typm{\Gamma,\lockwith{\mind{\mu}{F}}}{V:A}{E}}
  }
  {\typm{\Gamma}{\Box_\mu\,V : \boxwith{\mu} A}{F}}
\end{mathpar}
The only way to reduce is by \semlab{Lift} and $V\reducesto W$. IH on
\refa{} gives
\[
  {\typm{\Gamma,\lockwith{\mu_F}}{W:A}{E}}.
\]
Then by \tylab{Mod} we have
\[
  {\typm{\Gamma}{\Box_\mu\,W : \boxwith{\mu} A}{F}}.
\]
\item[Case]
\begin{mathpar}
\inferrule*[Lab=\tylab{Letmod}]
{
  \nu_F : E\to F \\
  \refa{\typm{\Gamma,\lockwith{\nu_F}}{V : \boxwith{\mu} A}{E}} \\
  \refb{\typm{\Gamma,x\varb{\nu_F\circ\mu_E}{A}}{M:B}{F}}
}
{\typm{\Gamma}{\Letm{\nu}{\mu} x = V \In M : B}{F}}
\end{mathpar}
By case analysis on the reduction.
\begin{description}
  \item[Case] \semlab{Lift} with $V\reducesto W$. By IH on \refa{} and
  reapplying \tylab{Letmod}.
  \item[Case] \semlab{Letmod}. We have $V = \Box_\mu\,U$ and
  \[
    \Letm{\nu}{\mu} x = \Box_\mu\,U \In M \reducesto M[U/x].
  \]
  Inversion on \refa{} gives
  \[
    {\typm{\Gamma,\lockwith{\nu_F},\lockwith{\mu_E}}{U : A}{E'}}.
  \]
  where $\mu_E : E' \to E$.
  By context equivalence, we have
  \[
    {\typm{\Gamma,\lockwith{\nu_F\circ\mu_E}}{U : A}{E'}}
  \]
  where $\nu_F\circ\mu_E : E'\to F$. By \Cref{lemma:substitution}.3 and \refb{}, we have
  \[
    {\typm{\Gamma}{M[U/x] : B}{F}}.
  \]
\end{description}
\item[Case]
\begin{mathpar}
  \inferrule*[Lab=\tylab{Letmod'}]
  {
    \nu_F : E\to F \\
    \refa{\typm{\Gamma,\lockwith{\nu_F},\ol{\alpha:K}}{V : \boxwith{\mu} A}{E}} \\
    \refb{\typm{\Gamma,x\varb{\nu_F\circ\mu_E}{\forall\ol{\alpha^K}.A}}{M:B}{F}}
  }
  {\typm{\Gamma}{\Letm{\nu}{\mu} \Lambda\ol{\alpha^K} . x = V \In M : B}{F}}
\end{mathpar}
Similar to the case for \tylab{Letmod'}. By case analysis on the reduction.
\begin{description}
  \item[Case] \semlab{Lift} with $V\reducesto W$. By IH on \refa{} and
  reapplying \tylab{Letmod'}.
  \item[Case] \semlab{Letmod'}. We have $V = \Box_\mu\,U$ and
  \[
    \Letm{\nu}{\mu} \Lambda\ol{\alpha^K} . x = \Box_\mu\,U \In M \reducesto M[(\Lambda\ol{\alpha^K}.U)/x].
  \]
  Inversion on \refa{} gives
  \[
    {\typm{\Gamma,\lockwith{\nu_F},\ol{\alpha:K},\lockwith{\mu_E}}{U : \boxwith{\mu} A}{E'}}.
  \]
  where $\mu_E : E' \to E$.
  By \Cref{lemma:structural-rules}.5 we have
  \[
    {\typm{\Gamma,\lockwith{\nu_F},\lockwith{\mu_E},\ol{\alpha:K}}{U : A}{E'}}.
  \]
  By context equivalence, we have
  \[
    {\typm{\Gamma,\lockwith{\nu_F\circ\mu_E},\ol{\alpha:K}}{U : A}{E'}}.
  \]
  where $\nu_F\circ\mu_E : E'\to F$. By \tylab{TAbs} we have
  \[
    {\typm{\Gamma,\lockwith{\nu_F\circ\mu_E}}{\Lambda \ol{\alpha^K} . U : \forall\ol{\alpha^K}. A}{E'}}.
  \]
  By \Cref{lemma:substitution}.3 and \refb{}, we have
  \[
    {\typm{\Gamma}{M[(\Lambda \ol{\alpha^K} .U)/x] : B}{F}}.
  \]
\end{description}
\item[Case] \tylab{TAbs},\tylab{Abs}. Impossible as there is no further reduction.
\item[Case]
\begin{mathpar}
\inferrule*[Lab=\tylab{TApp}]
{
\refa{\typm{\Gamma}{M:\forall\alpha^K.B}{E}} \\
\refb{\Gamma\vdash A : K} \\
}
{\typm{\Gamma}{M\,A : B[A/\alpha]}{E}}
\end{mathpar}
By case analysis on the reduction.
\begin{description}
  \item[Case] \semlab{Lift} with $M\reducesto N$. By IH on \refa{}
  and reapplying \tylab{TApp}.
  \item[Case] \semlab{TApp}. We have $M = \Lambda\alpha^K.V$ and
  \[
    (\Lambda\alpha^K.V)\,A\reducesto V[A/\alpha].
  \]
  Inversion on \refa{} gives
  \[
    \typm{\Gamma,\alpha:K}{V:B}{E}.
  \]
  Then by \Cref{lemma:substitution}.2 on \refb{}, we have
  \[
    \typm{\Gamma}{V[A/\alpha]:B[A/\alpha]}{E}.
  \]
\end{description}
\item[Case]
\begin{mathpar}
\inferrule*[Lab=\tylab{App}]
{
  \refa{\typm{\Gamma}{M : A \to B}{E}} \\
  \refb{\typm{\Gamma}{N : A}{E}}
}
{\typm{\Gamma}{M\; N: B}{E}}
\end{mathpar}
By case analysis on the reduction.
\begin{description}
  \item[Case] \semlab{Lift} with $M\reducesto M'$. By IH on \refa{}
  and reapplying \tylab{App}.
  \item[Case] \semlab{Lift} with $N\reducesto N'$. By IH on \refb{}
  and reapplying \tylab{App}.
  \item[Case] \semlab{App}. We have $M = \lambda x^A.M'$, $N=U$, and
  \[
    M\,N\reducesto M'[U/x].
  \]
  Inversion on \refa{} gives
  \[
    \typm{\Gamma,x:A}{M':B}{E}.
  \]
  Then by \Cref{lemma:substitution}.3 we have
  \[
    \typm{\Gamma}{M'[U/x]:B}{E}.
  \]
\end{description}
\item[Case] \tylab{Do}. The only way to reduce is by \semlab{Lift}.
Follow from IH and reapplying \tylab{Do}.
\item[Case]
\begin{mathpar}
\inferrule*[Lab=\tylab{Mask}]
{
  \refa{\typm{\Gamma,\lockwith{\mind{\amk{L}}{F}}}{M: A}{F-L}} \\
}
{\typm{\Gamma}{\Mask_{L}\; M: \boxwith{\amk{L}} A}{F}}
\end{mathpar}
By case analysis on the reduction.
\begin{description}
  \item[Case] \semlab{Lift} with $M\reducesto N$. By IH on \refa{} and reapplying \tylab{Mask}.
  \item[Case] \semlab{Mask}. We have $M = U$ and
  \[
  \Mask_L\,U\reducesto \Box_{\amk{L}}\, U.
  \]
  By $\amk{L}_F: F-L\to F$ and \tylab{Mod}, we have
  \[
    \typm{\Gamma}{\Box_{\amk{L}}\,U:\boxwith{\amk{L}}A}{F}.
  \]
\end{description}
\item[Case]
\begin{mathpar}
  \inferrule*[Lab=\tylab{Handler}]
  {
    H = \{\Ret x \mapsto N\} \uplus \{ \ell_i\;p_i\;r_i \mapsto N_i \}_i \\
    \\{D = \{\ell_i : A_i \sto B_i\}_i} \\
    \refa{\typm{\Gamma, \lockwith{\aex{D}_F}}{M : A}{D+F}} \\\\
    \refb{\typm{\Gamma, x : \boxwith{\aex{D}} A}{N : B}{F}} \\
    [\refc{\typm{\Gamma, p_i : A_i, {r_i}: B_i \to B}{N_i : B}{F}}]_i
  }
  {\typm{\Gamma}{\Handle\;M\With H : B}{F}}
\end{mathpar}
By case analysis on the reduction.
\begin{description}
  \item[Case] \semlab{Lift} with $M\reducesto M'$. By IHs and reapplying \tylab{Handler}.
  \item[Case] \semlab{Ret}. We have $M = U$ and
  \[
    \Handle\;U\With H \reducesto N[(\Box_{\aex{D}}\,U)/x].
  \]
  By \refa{}, $\amk{D}_F : F\to D+F$, and \tylab{Mod}, we have
  \[
  \typm{\Gamma}{\Box_{\aex{D}}\,U:A}{F}.
  \]
  Then by \refb{} and \Cref{lemma:substitution}.3 we have
  \[
    \typm{\Gamma}{N[(\Box_{\aex{D}}\,U)/x]:B}{F}.
  \]
  \item[Case] \semlab{Op}. We have $M = \EC[\Do\ell_j\;U]$,
  $\free{0}{\ell_j}{\EC}$, $\ell_j\,p_j\,r_j\mapsto N_j$, and
  \[
    \Handle\;M\With H \reducesto N_j[U/p, (\lambda y.\Handle\;\EC[y]\With H)/r].
  \]
  Since $D$ is well-kinded, $A_j$ and $B_j$ are pure. By inversion on
  $\Do\ell_j\;U$ we have
  \[
    {\typm{\Gamma,\lockwith{\aex{D}_F}}{U:A_j}{D+F}}.
  \]
  By $A_j$ is pure and \Cref{lemma:pure-promotion}, we have
  \[
    {\typm{\Gamma,\lockwith{\aex{D}_F},\lockwith{\amk{L}_{D+F}}}{U:A_j}{F}}
  \]
  where $L = \dom{D}$.
  By context equivalence, we have
  \[
    \refd{\typm{\Gamma}{U:A_j}{F}}
  \]
  Observe that $B_j$ being pure allows $y:B_j$ to be accessed in any
  context. By \refa{} and a straightforward induction on $\EC$ we have
  \[
    \typm{\Gamma,y:B_j,\lockwith{\aex{D}_F}}{\EC[y] : A}{D+F}.
  \]
  Then by \tylab{Handler} and \tylab{Abs} we have
  \[\refe{\typm{\Gamma}{\lambda y . \Handle\;\EC[y]\With H : B_j \to B}{F}}.\]
  Finally, by \refc{}, \refd{}, \refe{}, and \Cref{lemma:substitution}.3 we have
  \[\typm{\Gamma}{N_j[U/p, (\lambda y.\Handle\;\EC[y]\With H)/r]:B}{F}.\]
\end{description}
\item[Case]
\begin{mathpar}
\inferrule*[Lab=\tylab{Handler${}^\all$}]
{
  {D = \{\ell_i : A_i \sto B_i\}_i} \\
  \refa{\typm{\Gamma,\lockwith{\aeq{D+E}_F}}{M: A}{D+E}} \\\\
  \refb{\typm{\Gamma, \lockwith{\aeq{E}_F}, x : \boxwith{\aeq{D+E}} A}{N : B}{E}} \\
  [\refc{\typm{\Gamma, \lockwith{\aeq{E}_F}, p_i : A_i, r_i :
    \boxwith{\aeq{E}}
    (B_i \to B)}{N_i : B}{E}}]_i
  \\
  {[E]_F \To \one_F}
}
{\typm{\Gamma}{\Handle^\all\;M\With
  \{\Ret x \mapsto N\} \uplus \{ \ell_i\;p_i\;r_i \mapsto N_i \}_i : B}{F}}
\end{mathpar}
By case analysis on the reduction.
\begin{description}
  \item[Case] \semlab{Lift} with $M\reducesto M'$. By IHs and reapplying \tylab{Handle${}^\all$}.
  \item[Case] \semlab{Ret${}^\all$}. We have $M=U$ and
  \[
    \Handle\;M\With H \reducesto N[(\Box_{\aeq{D+E}}U)/x].
  \]
  By \refa{} and $\aeq{D+E}_F = \aeq{E}_F \circ \aeq{D+E}_E$ we have
  \[
  \typm{\Gamma,\lockwith{\aeq{E}_F}, \lockwith{\aeq{D+E}_E}}{U:A}{E}.
  \]
  By $\aeq{D+E}_E : D+E\to E$ and \tylab{Mod}, we have
  \[
  \typm{\Gamma,\lockwith{\aeq{E}_F}}{\Box_{\aeq{D+E}}\,U:\boxwith{\aeq{D+E}}A}{E}.
  \]
  Then by \refb{} and \Cref{lemma:substitution}.3 we have
  \[
    \typm{\Gamma, \lockwith{\aeq{E}_F}}{N[(\Box_{\aeq{D+E}}\,U)/x]:B}{E}.
  \]
  By ${[E]_F \To \one_F}$ and \Cref{lemma:structural-rules}.3 we have
  \[
    \typm{\Gamma}{N[(\Box_{\aeq{D+E}}\,U)/x]:B}{F}.
  \]
  \item[Case] \semlab{Op${}^\all$}.
  We have $M = \EC[\Do\ell_j\;U]$,
  $\free{0}{\ell_j}{\EC}$, $\ell_j\,p_j\,r_j\mapsto N_j$, and
  \[\Handle^\all\;M\With H \reducesto N_j[U/p, (\Box_{\aeq{E}}\,(\lambda y.\Handle^\all\; \EC[y] \With H))/r].\]
  Since $D$ is well-kinded, $A_j$ and $B_j$ are pure.
  By inversion on $\Do\ell_j\;U$, we have
  \[
    \typm{\Gamma,\lockwith{\aeq{D+E}_F}}{U:A_j}{D+E}.
  \]
  By the fact that $A_j$ is pure, \Cref{lemma:pure-promotion}, and context equivalence, we have
  \[
    \refd{\typm{\Gamma,\lockwith{\aeq{E}_F}}{U:A_j}{E}}.
  \]
  Observe that $B_j$ being pure allows $y$ to be accessed in any
  context.
  By \refa{} and a straightforward induction on $\EC$ we have
  \[
    \typm{\Gamma,y:B_j,\lockwith{\aeq{D+E}_F}}{\EC[y] : A}{D+E}.
  \]
  By $\aeq{E}_F\circ\aeq{E}_E\circ\aeq{D+E}_E = \aeq{D+E}_F$ and context equivalence, we have
  \[
    \typm{\Gamma,y:B_j,\lockwith{\aeq{E}_F},\lockwith{\aeq{E}_E},\lockwith{\aeq{D+E}_E}}{\EC[y] : A}{D+E}.
  \]
  Since $B_j$ is pure, we can swap $y:B_j$ with locks and derive
  \[
    \typm{\Gamma,\lockwith{\aeq{E}_F},\lockwith{\aeq{E}_E},y:B_j,\lockwith{\aeq{D+E}_E}}{\EC[y] : A}{D+E}.
  \]
  By \tylab{Handler${}^\all$}, we have
  \[
    {\typm{\Gamma,\lockwith{\aeq{E}_F},\lockwith{\aeq{E}_E},y:B_j}{\Handle^\all\;\EC[y]\With H : B}{E}}.
  \]
  Notice that we can put $H$ after absolute locks
  because all clauses in $H$ have an absolute lock $\lockwith{\aeq{E}_E}$ in
  their contexts.
  Then by \tylab{Abs} and \tylab{Mod} we have
  \[\refe{\typm{\Gamma,\lockwith{\aeq{E}_F}}{\Box_{\aeq{E}}\,(\lambda y . \Handle^\all\;\EC[y]\With H) : \boxwith{\aeq{E}}(B_j\to B)}{E}}.\]
  By \refc{}, \refd{}, \refe{}, and \Cref{lemma:substitution}.3 we have
  \[\typm{\Gamma, \lockwith{\aeq{E}_F}}{N_j[U/p, (\Box_{\aeq{F}}\,(\lambda y.\Handle\;\EC[y]\With H))/r]:B}{E}.\]
  Finally, by ${[E]_F \To \one_F}$ and \Cref{lemma:structural-rules}.3 we have
  \[\typm{\Gamma}{N_j[U/p, (\Box_{\aeq{F}}\,(\lambda y.\Handle\;\EC[y]\With H))/r]:B}{F}.\]
\end{description}
\item[Case] Shallow handlers. Similar to the cases of deep handlers.
\item[Case] Data types. Nothing more special than the cases we have
already shown. Introduction rules follows from IHs and reapplying the
same typing rules. Elimination rules require to additionally consider
their corresponding reduction rules.
\end{description}
\end{proof}

\section{Specification, Proof, and Discussion of the Encoding}
\label{app:encoding}

In this section, we show the full encoding of \Leff into \Met, prove
its type preservation, and discuss its extensibility.
The definition of \Leff is in \Cref{sec:Leff} and the definition of
\Met is in \Cref{sec:core-calculus}.

\subsection{Full Encoding}
\label{app:spec-encoding}

We repeat the encoding of types and contexts here for easy reference.

\begin{minipage}[t]{0.45\textwidth}
  \[\ba[t]{r@{\ \ }c@{\ \ }l}
    \stransl{\TUnit}{E} &=& \boxwith{\one} \TUnit \\
    \stransl{\earr{A}{B}{F}}{E} &=& \boxwith{\adj{E-F}{F-E}}(\stransl{A}{F}\to\stransl{B}{F}) \\
    \stransl{\forall.A}{E} &=& \boxwith{\aeq{}} \stransl{A}{\cdot} \\
  \ea
  \]
\end{minipage}
\begin{minipage}[t]{0.55\textwidth}
  \vspace{-.3em}
  \[\ba[t]{r@{\ \ }c@{\ \ }l}
    \stransl{\cdot}{E} &=& \cdot \\
    \stransl{\Gamma,x:A}{E} &=& \stransl{\Gamma}{E}, x\varb{\mu_E}{A'} \text{ for } \boxwith{\mu} A' = \stransl{A}{E} \\
    \stransl{\Gamma,\marker{}{F}}{E} &=& \stransl{\Gamma}{F}, \lockwith{\adj{F-E}{E-F}} \\
    \stransl{\Gamma,\marker{\Lambda}{F}}{\cdot} &=& \stransl{\Gamma}{F}, \lockwith{\aeq{}} \\
  \ea
  \]
  \\
\end{minipage}

\Cref{fig:Feff-to-Met} shows the translation from \Leff terms with their
types and effect contexts to \Met terms.
We use the following syntactic sugar to simply the encoding.
\[\ba{rcl}
\Letm{}{\mu} x = M \In N &\doteq& (\lambda x . \Letm{}{\mu} x = x \In N)\; M \\
\Letm{}{\mu;\nu} x = V \In M &\doteq&
  \Letm{}{\mu} x = V \In \Letm{\mu}{\nu} x = x \In M \\
\ea\]
We use an auxiliary function $\topmod{-}$ to extract the top-level
modality.
\[\ba{rcl}
\topmod{\boxwith{\mu}A} &=& \mu
\ea\]

\begin{figure}[htb]
\raggedright
\boxed{\encode{M:A}{E}{M'}}
\hfill
\begin{mathpar}
  \inferrule*[Lab=\rowlab{Var}]
  { \mu \coloneq \topmod{\stransl{A}{E}} }
  {
    \encode{x:A}{E}{\Box_\mu\,x}
  }

  \inferrule*[Lab=\rowlab{App}]
  {
    \encode{M : \earr{A}{B}{E}}{E}{M'} \\\\
    \encode{N : A}{E}{N'} \\ x\,\meta{fresh}
  }
  {\encode{M\, N : B}{E}{\Letm{}{\one} x = M' \In x\,N'}}

  \inferrule*[Lab=\rowlab{Abs}]
  {
    \encode{M:B}{F}{M'}
    \\
    \fortrans{\nu \coloneq \adj{E-F}{F-E}} \\
    \fortrans{\mu \coloneq \topmod{\stransl{A}{F}}}
  }
  {
    \encode{\elambda{F} x^A . M : \earr{A}{B}{F}}{E}{
      \Box_{\nu}\,(\lambda x^{\stransl{A}{F}} .
      \Letm{}{\mu} x = x \In
      M')
    }
  }

  \inferrule*[Lab=\rowlab{EAbs}]
  {
    \encode{V : A}{\cdot}{V'}
  }
  {
    \encode{\Lambda . V : \forall.A}{E}{
      \Box_{\aeq{}}\, V'
    }
  }

  \inferrule*[Lab=\rowlab{EApp}]
  {
    \encode{M : \forall.A}{E}{M'} \\
    x\,\meta{fresh} \\
  }
  {
    \encode{M\,@ : A[E/]}{E}{\Letm{}{\aeq{}} x = M' \In x}
  }

  \inferrule*[Lab=\rowlab{Do}]
  {
    \encode{M : A}{\ell,E}{M'} \\
  }
  {
    \encode{\Do \ell \; M : B}{\ell,E}{\Do\ell\;M'}
  }

  \inferrule*[Lab=\rowlab{Mask}]
  {
    \encode{M: A}{E}{M'} \\
    \fortrans{\mu_1 \coloneq \topmod{\stransl{A}{E}}} \\
    \fortrans{\mu_2 \coloneq \topmod{\stransl{A}{L+E}}} \\
  }
  {
    \encode{\Mask_{L}\; M: A}{L+E}{
      \Letm{}{\amk{L};\mu_1} x = \Mask_L\;M' \In
      \Box_{\mu_2}\, x
    }
  }

  \inferrule*[Lab=\rowlab{Handler}]
  {
    \encode{M : A}{\ol{\ell_i},E}{M'} \\
    \encode{N : B}{E}{N'} \\
    [\encode{N_i : B}{E}{N_i'}]_i \\\\
    \fortrans{\mu \coloneq \topmod{\stransl{A}{\ol{\ell_i},E}}} \\
    \fortrans{\mu' \coloneq \topmod{\stransl{A}{E}}} \\\\
    \fortrans{N'' \coloneq \Letm{}{\aex{\ol{\ell_i}};\mu} x = x \In \Letm{\mu'}{\one} x = \Box_{\one}\, x \In N'} \\\\
    \fortrans{[\mu_i \coloneq \topmod{\stransl{A_i}{\cdot}} \qquad
               N_i'' \coloneq
                \Letm{}{\mu_i} p_i = p_i \In
                \Letm{}{\aid} r_i = r_i \In
                N_i']_i} \\\\
    H = \{\Ret x \mapsto N\} \uplus \{ \ell_i\;p_i\;r_i \mapsto N_i \}_i \\
    \fortrans{H' \coloneq \{\Ret x \mapsto N''\} \uplus \{ \ell_i\;p_i\;r_i \mapsto N_i'' \}_i} \\
  }
  {
    \encode{\Handle\;M\With H : B}{E}{
      \Handle\;M'\With H'
    }
  }
\end{mathpar}

\caption{Encoding of \Leff in \Met.}
\label{fig:Feff-to-Met}
\end{figure}

In the term translation, all terms are translated to boxed terms with
proper modalities consistent with those given by the type translation.
Recall that \Met uses let-style unboxing; we cannot immediately unbox
values at the place we need.
To get a systematic encoding, we \emph{greedily unbox} top-level
modalities for term variables when they are bound, and rebox them when
they are used.

Greedy unboxing happens for variable bindings such as $\lambda$-abstractions and handlers.
In the \rowlab{Abs} case, we unbox the top-level
modality of variable $x$. %
Additionally, we box the whole function with the relative modality
$\adj{E-F}{F-E}$, reflecting the effect context transition.
In the $\rowlab{Handler}$ case, we similarly unbox the bound variables
for return and operation clauses.
In the operation clauses ($N_i''$), we need only unbox the operation
argument $p_i$; the resumption function $r_i$ is introduced under the
current effect context $E$. In the return clause ($N''$), we unbox $x$
with $\aex{\ol{\ell_i}}\circ\mu$ and then transform this modality to
$\mu'$ given by $\topmod{\stransl{A}{E}}$ in order to match the
effect context $E$.

Similar to the \rowlab{Abs} case, the \rowlab{EAbs} case boxes the
translated value with the empty absolute modality.
Similar to the return clauses of the \rowlab{Handler} case, the
\rowlab{Mask} case %
transforms the modality $\amk{L}\circ\mu_1$ to $\mu_2$ in order to
match the current effect context $L+E$.

In \rowlab{Var}, we rebox the variable with the appropriate modality
given by the type translation.

As a result of translating all terms to boxed terms, we must insert
unboxing for elimination rules such as \rowlab{App} and \rowlab{EApp}.
Nothing special happens for the \rowlab{Do} case.

\subsection{Proof of Encoding}
\label{app:proof-encoding}

We prove the encoding from \Leff into \CalcM in \Cref{sec:encodings}.

\LeffWellScoped*

\begin{restatable}[Well-scopedness of Derivation Trees]{lemma}{wsDerivTree}
If the judgement at the bottom of a derivation tree is well-scoped,
then every judgement in the derivation tree is well-scoped.
\end{restatable}
\begin{proof}
  Assume the contrary.
  Let $\typl{\Gamma_1, x :_\varepsilon A, \Gamma_2}{M : B}{E}$ be the top-most judgement in the derivation tree with
  $x \in \fv{M}$ and $\marker{\Lambda}{F} \in \Gamma_2$ and $A \neq \forall. A'$ and $A\neq\TUnit$.
  By case analysis on whether $\marker{\Lambda}{F} \in \Gamma_2$ was introduced in the derivation tree.
  \begin{description}
    \item[Case] not introduced in the derivation tree: Then the judgement at the bottom
      of the derivation tree must contain both the marker and $x$ and is not well-scoped for $x$. Contradiction.
    \item[Case] introduced in the derivation tree: since we chose the top-most judgement, the judgement must have
      introduced the marker by an application of the \rowlab{EAbs} rule. Let $\varepsilon'$ be the effect variable introduced at
      this judgement. Then $\varepsilon \neq \varepsilon'$ by the side-condition of the $\rowlab{EAbs}$ rule.
      We have that $\varepsilon$ is the ambient effect at the \rowlab{Var} rule where $x$ is used as a free variable,
      since we chose the top-most judgement. By the side-condition of the $\rowlab{Var}$ rule, then $\varepsilon = \varepsilon'$
      or $A = \forall. A'$ or $A = \TUnit$. Contradiction.
  \end{description}
\end{proof}

In the special case we consider there are no absent types.
This implies that submoding on effects can only add labels to the end.
Furthermore, all labels are drawn from a global environment and thus have
the same operation arrows. This allows us to freely permute them in the effect row.
In this case, we can strengthen the statement to the following:

\begin{restatable}[Transformation from Index]{corollary}{transFromIndex}
  \label{corollary:trans-from-index}
  If $\act{\adj{L_1}{D_1}}{F}\subtype\act{\adj{L_2}{D_2}}{F}$ and
  $L_1 \subtype F$ and $L_2 \subtype F$ and $L_1\bowtie D_1 = L_2\bowtie D_2$,
  then $\adj{L_1}{D_1}_F\To\adj{L_2}{D_2}_F$.
\end{restatable}
\begin{proof}
  We show that for all $F'$ with $F\subtype F'$, we have $\act{\adj{L_1}{D_1}}{F'}\subtype\act{\adj{L_2}{D_2}}{F'}$.
  Since all operations are present in $F$, we have that $F' = F + \ol{l}$ for some collection of labels
  $\ol{l}$. Then we use that $L_1 \subtype F$:
  \begin{align*}
    \act{\adj{L_1}{D_1}}{F'} &= \act{\adj{L_1}{D_1}}{F + \ol{l}}  \\
                             &= D_1 + ((F + \ol{l}) - L_1)        \\
                             &= D_1 + ((F - L_1) + \ol{l})        \\
                             &= \act{\adj{L_1}{D_1}}{F} + \ol{l}
  \end{align*}
  and the same for $\act{\adj{L_2}{D_2}}{F'}$. Since $\act{\adj{L_1}{D_1}}{F}\subtype\act{\adj{L_2}{D_2}}{F}$
  and we can freely permute labels,
  we have that $(\act{\adj{L_1}{D_1}}{F} + \ol{l})\subtype(\act{\adj{L_2}{D_2}}{F} + \ol{l})$.
\end{proof}

The condition that $L_1\bowtie D_1 = L_2\bowtie D_2$ can be checked easily,
where for the composition of modalities we use the fact that
for $\adj{L}{D} = \adj{L_1}{D_1}\circ\adj{L_2}{D_2}$,
we have $L \bowtie D = (L_1,L_2) \bowtie (D_1,D_2)$.

\begin{restatable}[First Modality Transformation]{lemma}{firstModTrans}
\label{lemma:first-mod-trans}
For all $E_1,E_2,E_3$:
\begin{mathpar}
  \inferrule*
  {}
  {(\adj{E_1-E_2}{E_2-E_1} \circ \adj{E_2-E_3}{E_3-E_2})_{E_1} \Leftrightarrow \adj{E_1-E_3}{E_3-E_1}_{E_1}}
\end{mathpar}
\end{restatable}
\begin{proof}
We can use \Cref{corollary:trans-from-index} since $(E_1 - E_3)\subtype E_1$
and $(E_1 - E_2) + L \subtype E_1$ where $(L,D) = (E_2 - E_3) \bowtie (E_2 - E_1)$.
We have:
\begin{align*}
  \act{\adj{E_1-E_3}{E_3-E_1}}{E_1} &= (E_3-E_1) + (E_1 - (E_1 - E_3)) \\
                              &= (E_3-E_1) + (E_1 \cap E_3) \\
                              &= E_3
\end{align*}
and using this calculation:
\begin{align*}
  \act{\adj{E_1-E_2}{E_2-E_1} \circ \adj{E_2-E_3}{E_3-E_2}}{E_1}
    &= \act{\adj{E_2-E_3}{E_3-E_2}}{\act{\adj{E_1-E_2}{E_2-E_1}}{E_1}} \\
    &= \act{\adj{E_2-E_3}{E_3-E_2}}{E_2} \\
    &= E_3
\end{align*}
\end{proof}

\begin{restatable}[Second Modality Transformation]{lemma}{secondModTrans}
\label{lemma:second-mod-trans}
For all $L,E,F$:
\begin{mathpar}
  \inferrule*
  {}
  {\adj{L + (E - F)}{F - E}_{L + E} \To \adj{(L + E) - F}{F - (L + E)}_{L + E}}
\end{mathpar}
\end{restatable}
\begin{proof}
We can use \Cref{corollary:trans-from-index} since $(L + E) - F \subtype L + E$
and $L + (E - F) \subtype L + E$. We have:
\begin{align*}
  \act{\adj{(L + E) - F}{F - (L + E)}}{L + E} &= (F - (L + E)) + ((L + E) - (L + E - F)) \\
                              &= (F - (L + E)) + ((L + E) \cap F) \\
                              &= F
\end{align*}
and:
\begin{align*}
  \act{\adj{L + (E - F)}{F - E}}{L + E} &= (F - E) + ((L + E) - (L + (E - F))) \\
                              &= (F - E) + (E - (E - F)) \\
                              &= (F - E) + (E \cap F) \\
                              &= F
\end{align*}
\end{proof}

\begin{restatable}[Third Modality Transformation]{lemma}{thirdModTrans}
\label{lemma:third-mod-trans}
For all $\ol{\ell_i},E,F$:
\begin{mathpar}
  \inferrule*
  {}
  {(\aex{\ol{\ell_i}}\circ\adj{\ol{\ell_i},E - F}{F - \ol{\ell_i},E})_{E} \To \adj{E - F}{F - E}_{E}}
\end{mathpar}
\end{restatable}
\begin{proof}
We can use \Cref{corollary:trans-from-index} since
$(\aex{\ol{\ell_i}}\circ\adj{\ol{\ell_i},E - F}{F - \ol{\ell_i},E}) = \act{\adj{\ol{\ell_i},E - F}{F - \ol{\ell_i},E}}{\ol{\ell_i},E}$
and $\ol{\ell_i},E - F \subtype \ol{\ell_i},E$ and $E - F \subtype E$. We have $\act{\adj{E - F}{F - E}}{E} = F$ and:
\begin{align*}
  \act{\aex{\ol{\ell_i}}\circ\adj{\ol{\ell_i},E - F}{F - \ol{\ell_i},E}}{E}
    &= \act{\adj{\ol{\ell_i},E - F}{F - \ol{\ell_i},E}}{\act{\aex{\ol{\ell_i}}}{E}} \\
    &= \act{\adj{\ol{\ell_i},E - F}{F - \ol{\ell_i},E}}{\ol{\ell_i},E} \\
    &= F
\end{align*}
\end{proof}

\begin{restatable}[Translating Instantiated Types]{lemma}{instTypes}
\label{lemma:inst-types}
For all \Leff types $A$: $\stransl{A}{E} = \stransl{A[E'/]}{E,E'}$.
\end{restatable}
\begin{proof}
By induction on the type $A$.
\begin{description}
\item[Case] $A = \Int$. Trivial.
\item[Case] $A = \forall.A'$. Trivial.
\item[Case] $A = \earr{A'}{B'}{F}$. Then:
\begin{align*}
  \stransl{A}{E} &= \adj{E - F}{F - E}(\stransl{A'}{F} \to \stransl{B'}{F}) \\
  \stransl{A[E'/]}{E,E'} &= \adj{E,E' - F,E'}{F,E' - E,E'}(\stransl{A'[E'/]}{F,E'} \to \stransl{B'[E'/]}{F,E'})
\end{align*}
  By the induction hypothesis we have:
\begin{align*}
  \stransl{A'}{F} &= \stransl{A'[E'/]}{F,E'} \\
  \stransl{B'}{F} &= \stransl{B'[E'/]}{F,E'}
\end{align*}
  Since we can freely permute labels:
\begin{align*}
  \adj{E,E' - F,E'}{F,E' - E,E'} &= \adj{E',E - E',F}{E',F - E',E} \\
                                 &= \adj{E - F}{F - E}
\end{align*}
\end{description}
\end{proof}

\LeffToMet*
\begin{proof}
By induction on the typing derivation $\typl{\Gamma}{M:A}{E}$. We prove this for each rule of the translation.
As a visual aid, we repeat each rule where we replace the translation premises by the \CalcM judgement implied by
the induction hypothesis and the translation in the conclusion by the \CalcM judgement we need to prove.
\begin{description}

\item[Case]
\begin{mathpar}
  \inferrule*[Lab=\rowlab{Var}]
  { }
  {
    \typm{\stransl{\Gamma_1, x : A, \Gamma_2}{E}}{\rebox{x}{A}{E} : \stransl{A}{E}}{E}
  }
\end{mathpar}

We use the $\rebox{x}{A}{E}$ function defined as follows:
\[
\ba{rcl}
\rebox{x}{A}{E} &=&
\begin{cases}
\Box_{\aid{}}\,x, &\text{if } A = \Int \\
\Box_{\adj{E-F}{F-E}}\,x, &\text{if } A = \earr{A'}{B'}{F} \\
\Box_{\aeq{}}\,x, &\text{if } A = \forall. A'
\end{cases} \\
\ea
\]

This function is exactly equivalent to $\Box_\mu\,x$ where $\mu = \topmod{\stransl{A}{E}}$
We use the \tylab{Mod} rule to introduce the box. By cases on the type $A$:
\begin{description}
\item[Case] $A = \Int$. We can use the \tylab{Var} rule since $\cdot \vdash \Int : \Pure$.
\item[Case] $A = \forall.A'$. Then $\stransl{A}{F} = \boxwith{\aeq{}} \stransl{A'}{\cdot}$ for all $F$.
  By rule \mtylab{Abs}, the pure modality transforms into any other modality and so we can use the \tylab{Var} rule.
\item[Case] $A = \earr{A'}{B'}{F}$. Since the \Leff judgement is well-scoped, we have that $\locks{\Gamma_2}$ is
  the composition of transition modalities. Furthermore, $\locks{\Gamma'}\circ\adj{E-F}{F-E} : F \to F'$ for the
  context $F'$ where $x$ as introduced and $x$ is annotated by the modality $\adj{F'-F}{F-F'}_{F'} : F \to F'$.
  By \Cref{lemma:first-mod-trans}, we can use the \tylab{Var} rule.
\end{description}

\item[Case]
\begin{mathpar}
  \inferrule*[Lab=\rowlab{App}]
  {
    \typm{\stransl{\Gamma}{E}}{M' : \stransl{\earr{A}{B}{E}}{E}}{E} \\\\
    \typm{\stransl{\Gamma}{E}}{N' : \stransl{A}{E}}{E} \\
    x\,\meta{fresh}
  }
  {
    \typm{\stransl{\Gamma}{E}}{\Letm{}{\aid{}} x = M' \In x\,N' : \stransl{B}{E}}{E}
  }
\end{mathpar}

We have $\stransl{\earr{A}{B}{E}}{E} = \aid{}(\stransl{A}{E} \to \stransl{B}{E})$. The claim follows
by the \tylab{Letmod} and \tylab{App} rules.

\item[Case]
\begin{mathpar}
  \inferrule*[Lab=\rowlab{Abs}]
  {
    \typm{\stransl{\Gamma,\marker{}{E}, x : A}{F}}{M' : \stransl{B}{F}}{F}
    \\
    \fortrans{\nu \coloneq \adj{E-F}{F-E}} \\
    \fortrans{\mu \coloneq \topmod{\stransl{A}{F}}}
  }
  {
    \typm{\stransl{\Gamma}{E}}{\Box_{\nu}\,(\lambda x^{\stransl{A}{F}} .
      \Letm{}{\mu} x = x \In
      M') : \stransl{\earr{A}{B}{F}}{E}}{E}
  }
\end{mathpar}

We have $\stransl{\Gamma,\marker{}{E}, x : A}{F} = \stransl{\Gamma}{E},\lockwith{\adj{E-F}{F-E}}, x :_{\mu_F} A'$
where $\boxwith{\mu} A' = \stransl{A}{F}$. Further $\stransl{\earr{A}{B}{F}}{E} = \adj{E-F}{F-E}(\stransl{A}{F} \to \stransl{B}{F})$.
The claim follows from the \tylab{Letmod}, \tylab{Abs} and \tylab{Mod} rules.

\item[Case]
\begin{mathpar}
  \inferrule*[Lab=\rowlab{EAbs}]
  {
    \typm{\stransl{\Gamma,\marker{\Lambda}{E}}{\cdot}}{V' : \stransl{A}{\cdot}}{\cdot}
  }
  {
    \typm{\stransl{\Gamma}{E}}{\Box_{\aeq{}}\, V' : \stransl{\forall.A}{E}}{E}
  }
\end{mathpar}

We have $\stransl{\Gamma,\marker{\Lambda}{E}}{\cdot} = \stransl{\Gamma}{E},\lockwith{\aeq{}}$.
Further, $\stransl{\forall.A}{E} = \boxwith{\aeq{}} \stransl{A}{\cdot}$.
The claim follows from the \tylab{Mod} rule.

\item[Case]
\begin{mathpar}
  \inferrule*[Lab=\rowlab{EApp}]
  {
    \typm{\stransl{\Gamma}{E}}{M' : \stransl{\forall.A}{E}}{E} \\
    x\,\meta{fresh} \\
  }
  {
    \typm{\stransl{\Gamma}{E}}{\Letm{}{\aeq{}} x = M' \In x : \stransl{A[E/]}{E}}{E}
  }
\end{mathpar}

We have $\stransl{\forall.A}{E} = \boxwith{\aeq{}} \stransl{A}{\cdot}$.
By \Cref{lemma:inst-types}, $\stransl{A}{\cdot} = \stransl{A[E/]}{E}$.
The claim follows by the \tylab{Letmod} rule.

\item[Case]
\begin{mathpar}
  \inferrule*[Lab=\rowlab{Do}]
  {
    \ell : A \sto B \in \Sigma \\\\
    \typm{\stransl{\Gamma}{\ell,E}}{M' : \stransl{A}{\ell,E}}{\ell,E}
  }
  {
    \typm{\stransl{\Gamma}{\ell,E}}{\Do\ell\;M' : \stransl{B}{\ell,E}}{\ell,E}
  }
\end{mathpar}

Because we only allow pure values in the operation arrows of \Leff, we have
that $\stransl{A}{\ell,E} = \stransl{A}{\cdot}$
and $\stransl{B}{\ell,E} = \stransl{B}{\cdot}$,
where $\ell : \stransl{A}{\cdot} \sto \stransl{B}{\cdot}$ in \CalcM.
The claim follows directly by the \tylab{Do} rule.

\item[Case]
\begin{mathpar}
  \inferrule*[Lab=\rowlab{Mask}]
  {
    \typm{\stransl{\Gamma,\marker{}{L+E}}{E}}{M' : \stransl{A}{E}}{E} \\\\
    \fortrans{\mu_1 \coloneq \topmod{\stransl{A}{E}}} \\
    \fortrans{\mu_2 \coloneq \topmod{\stransl{A}{L+E}}} \\
  }
  {
    \typm{\stransl{\Gamma}{L+E}}{\Letm{}{\amk{L};\mu_1} x = \Mask_L\;M' \In \Box_{\mu_2}\, x : \stransl{A}{L+E}}{L+E}
  }
\end{mathpar}

We have $\stransl{\Gamma,\marker{}{L+E}}{E} = \stransl{\Gamma}{L+E},\lockwith{\adj{(L + E) - E}{E - (L + E)}}$.
By permuting labels, we have\\ $\adj{(L + E) - E}{E - (L + E)} = \amk{L}$. The goal follows by the
\tylab{Letmod}, \tylab{Mask} and \tylab{Mod} rules if we can show that $x$ can be used under the box.
This is clear for integers, since they are pure and otherwise we need to show that
$(\amk{L}\circ\mu_1)_{L+E} \To (\mu_2)_{L+E}$. For $A = \forall.A'$ this is clear since $\mu_1 = \mu_2 = \aeq{}$
and $\amk{L}\circ\aeq{} = \aeq{}$. For functions, this follows from \Cref{lemma:second-mod-trans}.

\item[Case]
\begin{mathpar}
  \inferrule*[Lab=\rowlab{Handler}]
  {
    \typm{\stransl{\Gamma,\marker{}{E}}{\ol{\ell_i},E}}{M' : \stransl{A}{\ol{\ell_i},E}}{\ol{\ell_i},E} \\\\
    \typm{\stransl{\Gamma, x : A}{E}}{N' : \stransl{B}{E}}{E} \\
    [\typm{\stransl{\Gamma, p_i : A_i, r_i : \earr{B_i}{B}{E}}{E}}{N_i' : \stransl{B}{E}}{E}]_i \\\\
    \fortrans{\mu \coloneq \topmod{\stransl{A}{\ol{\ell_i},E}}} \\
    \fortrans{\mu' \coloneq \topmod{\stransl{A}{E}}} \\\\
    \fortrans{N'' \coloneq \Letm{}{\aex{\ol{\ell_i}};\mu} x = x \In \Letm{\mu'}{\aid{}} x = \Box_{\aid{}}\, x \In N'} \\\\
    \fortrans{[\mu_i \coloneq \topmod{\stransl{A_i}{\cdot}} \qquad
               N_i'' \coloneq \Letm{}{\mu_i} p_i = p_i \In
                \Letm{}{\aid} r_i = r_i \In
               N_i']_i} \\\\
    H = \{\Ret x \mapsto N\} \uplus \{ \ell_i\;p_i\;r_i \mapsto N_i \}_i \\
    \fortrans{H' \coloneq \{\Ret x \mapsto N''\} \uplus \{ \ell_i\;p_i\;r_i \mapsto N_i' \}_i} \\
  }
  {
    \typm{\stransl{\Gamma}{E}}{\Handle\;M'\With H' : \stransl{B}{E}}{E}
  }
\end{mathpar}

We have $\stransl{\Gamma,\marker{}{E}}{\ol{\ell_i},E} = \stransl{\Gamma}{E},\lockwith{\adj{E - \ol{\ell_i},E}{\ol{\ell_i},E - E}}$.
By permuting labels, we have $\adj{E - \ol{\ell_i},E}{\ol{\ell_i},E - E} = \aex{\ol{\ell_i}}$.
In the operation clauses, we have that $\stransl{\earr{B_i}{B}{E}}{E} = \aid{}(\stransl{B_i}{E} \to \stransl{B}{E})$.
Because the argument and return of effects are pure, we have that $\stransl{B_i}{E} = \stransl{B_i}{\cdot}$ and
$\stransl{A_i}{E} = \stransl{A}{\cdot}$. We need to unbox the argument $p_i$ and continuation $r_i$ though.
In the return clause, \CalcM gives us $x : \aex{\ol{\ell_i}}\stransl{A}{\ol{\ell_i},E}$, but
we need $x : \stransl{A}{E}$. We achieve this by unboxing $x$ fully and then re-boxing it with the
modality $\mu'$. This is possible for integers because they are pure, for $\forall$s because
of the \mtylab{Abs} rule and for functions due to the modality transformation in \Cref{lemma:third-mod-trans}.

\end{description}

\end{proof}

\subsection{Extensibility of the Encoding}
\label{app:extensible-encoding}

Our encoding in \Cref{sec:encodings} does not consider any extensions
of \Met.
With the extension of effect polymorphism, the encoding certainly
becomes trivial.
Thus, we only consider extensions of data types and polymorphism for value types
and discuss how to extend the encoding.

Recall that in \Cref{sec:encoding-Leff} we always translate types to
modal types and perform greedy unboxing and lazy boxing for variables.
For variables of data types such as a pair $(A, B)$, we just need to
further destruct the pair before greedy unboxing, and reconstruct the
pair when lazy boxing.
This is because the translation on its components might give terms of
type $\boxwith{\mu}A'$ and $\boxwith{\nu}B'$ with different
modalities which require separate unboxing.
For variables of recursive data types, we need to destruct
only to the extent that the data type is unfolded in the function body
(where we may treat recursive invocations as opaque). While this
requires a somewhat global translation, it does not require
destructing and unboxing the recursive data type more than a small
number of times.

The essential reason for the translation being global comes from the
fact that we use let-style unboxing following MTT.
For modalities with certain structure (right adjoints), it is possible
to use Fitch-style unboxing~\citep{Clouston18} which allows terms to
be directly unboxed without binding~\citep{Gratzer23,abs-2303-02572}.
We are interested in exploring whether we could extend \Met to use
Fitch-style unboxing and thus give a compositional local encoding for
recursive data types.
Fortunately, these issues appear to only cause problems for encoding
but not in practice.
Functional programs typically use pattern-matching in a structured way
that plays nicely with automatic unboxing of bidirectional typing.

Extending the encoding with polymorphism for value types is tricky as the source
calculus \Leff is not stable under value type substitution.
For instance, the following substitution breaks the condition that
function arrows only refer to the lexically closest effect variable:
$
(\forall\gray{\evar'}.\alpha)[(\earr{\TUnit}{\TUnit}{\geffrow{}{\evar}})/\alpha]
= \forall\gray{\evar'}.\earr{\TUnit}{\TUnit}{\geffrow{}{\evar}}
$.
This exemplifies the fragility of the syntactic approach of \Frank.
It is possible to still define the encoding by forcing the substituted
type to satisfy the lexical restriction.
We leave the full development of such an encoding as future work.

\FloatBarrier
\section{Specification and Implementation of \SurfaceMet}
\label{app:surface-lang}

In this section, we provide the syntax and typing rules for
\SurfaceMet introduced in \Cref{sec:surface-lang} and the elaboration
from it to \Met.
We also briefly discuss our prototype implementation.

\subsection{Syntax and Typing Rules}
\label{app:surface-lang-spec}

The syntax of \SurfaceMet is shown in~\Cref{fig:syntax-maetel}.
We include extensions of data types and polymorphism in
\Cref{sec:extensions}.
For the latter we require explicit type application $M\;A$ (explicit
type abstraction is optional since the checking mode tells us when to
introduce type abstraction).

\begin{figure}[htbp]
  \begin{syntax}
  \slab{Types}    &A,B  &::= & A\to B \mid {\boxwith{\mu} A}
    \mid \alpha \mid \Pair{A}{B} \mid A + B \\
  \slab{Guarded Types}    &G  &::= & A\to B
    \mid \alpha \mid \Pair{A}{B} \mid A + B \\
  \slab{Masks}          &L   &::= & \cdot \mid \ell,L \\
  \slab{Extensions}     &D   &::= & \cdot \mid \ell : P, D \\
  \slab{Effect Contexts}&E,F &::= & \cdot \mid \ell : P, E \\
  \slab{Modalities}     &\mu &::= & \aeq{E} \mid \adj{L}{D} \\
  \slab{Kinds}    &K &::= & \Pure \mid \Any \mid \Effect \\
  \slab{Contexts}       &\Gamma    &::=& \cdot \mid \Gamma, x:A
                                            \mid \Gamma, \alpha:K
                                            \mid {\Gamma,\lockwith{\mu_F}} \\
  \slab{Terms}          &M,N  &::= & x \mid \lambda x.M \mid M\;N \mid {M:A} \mid M\;A
                              \mid  {\Do\ell\; M} \\
                        &     &\mid& {\Mask_L\;M}\mid \Handle\;M\With H \\
                        &     &\mid& (M, N) \mid \Case M \Of (x, y) \mapsto N \\
                        &     &\mid& \Inl M \mid \Inr M
                              \mid \Case M \Of \{\Inl x \mapsto M_1, \Inr y \mapsto M_2\} \\
  \slab{Values}         &V,W  &::= & x \mid \lambda x.M \mid {V:A} \mid V\;A
                        \mid (V, W) \mid \Inl V \mid \Inr V\\
  \slab{Handlers}      &H    &::= & \{ \Ret x \mapsto M \}
                              \mid  \{ (\ell:A\sto B) \; p \; r \mapsto M \} \uplus H
  \end{syntax}
  \caption{Syntax of \SurfaceMet.}
  \label{fig:syntax-maetel}
\end{figure}

The bidirectional typing rules for \SurfaceMet are shown
in~\Cref{fig:typing-maetel}.

We repeat the definition of $\meta{across}$ used in \btylab{Var} here
for easy reference.
\[\ba{rcl}
\meta{across}(\Gamma,A,\nu,F) &=&
\begin{cases}
  A, &\text{if } \Gamma\vdash A : \Pure \\
  \boxwith{\zeta} G, &\text{otherwise, where } A = \boxwith{\ol{\mu}}{G}
    \text{ and } \rresidual{\ol{\mu}_F}{\nu_F} = \zeta_E
\end{cases}
\ea\]
For $\mu_F : E\to F$ and $\nu_F : F' \to F$, the right residual
$\rresidual{\mu_F}{\nu_F}$ is a partial operation defined as follows.
\[\ba{r@{\,}c@{\,}lcl}
{\nu_F}
&\backslash&
{\aeq{E}_F}
&=& \aeq{E}_{F'} \\
{\adj{L'}{D'}_F}
&\backslash &
{\adj{L}{D}_F}
&=&
\begin{cases}
  \adj{\toSet{D'} + (L - L')}{D + F|_{L'-L}}_{D'+(F-L')},
  &\text{if } \meta{present}(F|_{L'-L})\\
  \text{fail}, &\text{otherwise}
\end{cases}\\
{\aeq{E}_F}
&\backslash &
{\adj{L}{D}_F}
&=& \text{fail}
\ea\]
We define $\toSet{D}$ as converting an extension to a mask by taking
the multiset of all its labels.
We define $E|_L$ as the extension derived by extracting the entries in
$E$ with labels in $L$ from the left.
The predicate $\meta{present}(D)$ checks if all labels in $D$ are present.

\begin{minipage}[t]{0.5\textwidth}
\[\ba{rcl}
\toSet{\cdot} &=& \cdot \\
\toSet{\ell:P,D} &=& \ell,\toSet{D} \\
\ea\]
\end{minipage}
\begin{minipage}[t]{0.5\textwidth}
\[\ba{rcl}
(\ell:P,E)|_{\ell,L} &=& \ell:P, E|_L \\
(\ell':P,E)|_{\ell,L} &=& E|_{\ell,L} \\
\evar |_{\ell,L} &=& \text{fail}
\ea\]
\end{minipage}
\vspace{.5\baselineskip}

\btylab{Mod} introduces a lock and \btylab{Forall} introduces a type
variable into the context, respectively.
\btylab{Annotation} is standard for bidirectional typing.
\btylab{Switch} not only switches the direction from checking to
inference, but also transforms the top-level modalities when there is
a mismatch.
\btylab{Abs} is standard. Both \btylab{App} and \btylab{TApp} unbox
the eliminated term $M$ when it has top-level modalities.
\btylab{Do} is standard. For masks and handlers, we have typing rules
in both checking and inference modes.
For \btylab{HandlerInfer}, we use a partial join operation $A\join_{\Gamma,F}
B$ to join the potentially different types of different branches.
The join operation fails when $A$ and $B$ are different types modulo
top-level modalities; otherwise, it tries to transform the top-level
modalities of $A$ and $B$ to the same one.
As a special case, if $A$ and $B$ give some absolute guarded type $G$
after removing top-level modalities, the join operation succeeds and
directly returns $G$, which is a general enough solution because an
absolute type has no restriction on its accessibility.

We define join on types as follows.
\[\ba{r@{\;}c@{\;}lcl}
\boxwith{\ol{\mu}}G &\join_{\Gamma,F}& \boxwith{\ol{\nu}}G
  &=&
\begin{cases}
  G, &\text{if } \Gamma\vdash G : \Pure \\
  \boxwith{({\ol{\mu}} \join_F {\ol{\nu}})} G, &\text{otherwise} \\
\end{cases}
\ea\]
In order to define $\mu\join_F\nu$, we first define some auxiliary join operations.

We define the join of presence types $P\join P'$ as follows, in
order to obtain the minimal type $P''$ such that $P\subtype P''$
and $P'\subtype P''$.
\[\ba{r@{\;}c@{\;}lcl}
\Abs &\join& P &=& P \\
P    &\join& \Abs &=& P \\
P    &\join& P' &=& 
\begin{cases}
P, &\text{if } P \equiv P' \\
\text{fail}, &\text{otherwise}
\end{cases}
\ea\]

We define the join of effect contexts $E\join E'$ as follows, in order
to obtain the minimal effect context $F$ such that $E\subtype F$ and
$E'\subtype F$.
\[\ba{r@{\;}c@{\;}lcl}
\cdot &\join& \cdot &=& \cdot \\
\ell:P,E &\join& \ell:P',E' &=& \ell:(P\join P'), (E\join E') \\
\ea\]

We define the join of an extension and an effect context $D\join E$ as
follows, creating an extension based on $D$ by joining the presence types
with those of corresponding labels in $E$.
\[\ba{r@{\;}c@{\;}lcl}
\cdot &\join& E &=& \cdot \\
\ell:P,D &\join& \ell:P',E' &=& \ell:(P\join P'), (D\join E') \\
\ea\]

We define the meet of extensions $D\wedge D'$ as follows.
\[\ba{r@{\;}c@{\;}lcl}
\cdot &\wedge& D &=& \cdot \\
\ell:P,D &\wedge& \ell:P',D' &=& \ell:(P\join P'), (D\wedge D') \\
\ell:P,D &\wedge& D' &=& D\wedge D', \quad\text{where } \ell\notin\dom{D'} \\
\ea\]
Note that we take the join of the presence types when the labels are
present in both extensions.

We write $L\cup L'$ and $L\cap L'$ for the standard union and
intersection of multisets.

We define the join of modalities $\mu\join_F\nu$ at mode $F$ as
follows, where
if the join operation succeeds, we have $\mu_F\To (\mu\join_F\nu)_F$
and $\nu_F\To (\mu\join_F\nu)_F$, and for any other $\zeta$ such that
$\mu_F\To\zeta_F$ and $\nu_F\To\zeta_F$, we have $(\mu\join_F\nu)_F
\To \zeta_F$.
\[\ba{r@{\;}c@{\;}lcl}
\aeq{E} &\join_F& \aeq{E'} &=& \aeq{E\join E'} \\
\aeq{E} &\join_F& \adj{L}{D} &=&
\begin{cases}
  \adj{L}{D\join E}, &\text{if } E \subtype D\join E +(F-L) \\
  \text{fail}, &\text{otherwise}
\end{cases}
\\
\adj{L}{D} &\join_F& \adj{L'}{D'} &=&
\begin{cases}
  \adj{L\cap L'}{D\wedge D'}, &
    \bl
    \text{if } \adj{L}{D}_F\To\adj{L\cap L'}{D \wedge D'}_F \\
    \text{and } \adj{L'}{D'}_F\To\adj{L\cap L'}{D \wedge D'}_F \\
    \el\\
  \text{fail},& \text{otherwise}
\end{cases}
\\

\ea\]
The third case may appear surprising since it takes the meet of
masks and extensions.
It works because the \mtylab{Shrink} rule in \Cref{sec:modalities}
allows us to remove corresponding elements from masks and extensions
simultaneously.
The side conditions guarantee that both modalities can be transformed
to their join.

The introduction rules of data types are standard.
For their elimination, we have both checking and inference version.
Both \btylab{CrispPair} and \btylab{CrispSum} extract the top-level
modalities of the data values and distribute them to the types of the
variables for their components.
We also use the join operation to unify the different branches of
\btylab{CrispSumInfer}.

\begin{figure}[htbp]
\raggedright
\boxed{\typmi{\Gamma}{M}{A}{E}{}}
\boxed{\typmc{\Gamma}{M}{A}{E}}
\hfill

\begin{mathpar}
\inferrule*[Lab=\btylab{Var}]
{
  \nu_F = \locks{\Gamma'} : E\to F\\
  B = \meta{across}(\Gamma, A, \nu, F)
}
{\typmi{\Gamma,x:A,\Gamma'}{x}{B}{E}{}}

\inferrule*[Lab=\btylab{Mod}]
{
  \typmc{\Gamma,\lockwith{\mu_F}}{V}{A}{E} \\
  \mu_F : E \to F
}
{\typmc{\Gamma}{V}{\boxwith{\mu} A}{F}}

\inferrule*[Lab=\btylab{Forall}]
{
  \typmc{\Gamma,\alpha : K}{V}{A}{E} \\
}
{\typmc{\Gamma}{V}{\forall\alpha:K . A}{E}}

\inferrule*[Lab=\btylab{Annotation}]
{
  \typmc{\Gamma}{V}{A}{E} \\
}
{\typmi{\Gamma}{V:A}{A}{E}{}}

\inferrule*[Lab=\btylab{Switch}]
{
  \typmi{\Gamma}{M}{\boxwith{\ol{\mu}} G}{E}{} \\
  \Gamma\vdash ({\ol{\mu}}, G) \To {\ol{\nu}} \atmode{E}
}
{\typmc{\Gamma}{M}{\boxwith{\ol{\nu}} G}{E}}

\inferrule*[Lab=\btylab{Abs}]
{
  \typmc{\Gamma, x : A}{M}{B}{E}
}
{\typmc{\Gamma}{\lambda x . M}{A \to B}{E}}

\inferrule*[Lab=\btylab{App}]
{
  \typmi{\Gamma}{M}{\boxwith{\ol{\mu}}(A \to B)}{E}{} \\\\
  {\ol{\mu}}_E \To \one_E \\
  \typmc{\Gamma}{N}{A}{E} \\
}
{\typmi{\Gamma}{M\; N}{B}{E}{}}

\inferrule*[Lab=\btylab{TApp}]
{
  \typmi{\Gamma}{M}{\boxwith{\ol{\mu}}(\forall\alpha:K . B)}{E}{} \\\\
  {\ol{\mu}}_E \To \one_E \\
  \Gamma\vdash A : K
}
{\typmi{\Gamma}{M\;A}{B[A/\alpha]}{E}{}}

\inferrule*[Lab=\btylab{Do}]
{
  E = \ell:A\sto B, F \\\\
  \typmc{\Gamma}{M}{A}{E} \\
}
{\typmi{\Gamma}{\Do \ell \; M}{B}{E}{}}

\inferrule*[Lab=\btylab{MaskCheck}]
{
  \typmc{\Gamma,\lockwith{\amk{L}_F}}{M}{A}{F-L}
}
{\typmc{\Gamma}{\Mask_L\; M}{\boxwith{\amk{L}}A}{F}}

\inferrule*[Lab=\btylab{MaskInfer}]
{
  \typmi{\Gamma,\lockwith{\amk{L}_F}}{M}{A}{F-L}{}
}
{\typmi{\Gamma}{\Mask_L\; M}{\boxwith{\amk{L}}A}{F}{}}

\inferrule*[Lab=\btylab{HandlerInfer}]
{
  D = \{\ell_i : A_i \sto B_i\}_i \\
  \typmi{\Gamma, \lockwith{\aex{D}_E}}{M}{A}{D+E}{} \\
  \typmi{\Gamma, x : \boxwith{\aex{D}} A}{N}{B'}{E}{} \\\\
  [\typmi{\Gamma, p_i : A_i, {r_i}: B_i \to B'}{N_i}{B_i}{E}{}]_i \\
  B = B' (\join_{\Gamma,E} B_i)_i
}
{\typmi{\Gamma}{\Handle\;M\With \{\Ret x \mapsto N\} \uplus \{ (\ell_i:A_i\sto B_i)\;p_i\;r_i \mapsto N_i \}_i}{B}{E}{}}

\inferrule*[Lab=\btylab{HandlerCheck}]
{
  D = \{\ell_i : A_i \sto B_i\}_i \\
  \typmi{\Gamma, \lockwith{\aex{D}_E}}{M}{A}{D+E}{} \\
  \typmc{\Gamma, x : \boxwith{\aex{D}} A}{N}{B}{E} \\\\
  [\typmc{\Gamma, p_i : A_i, {r_i}: B_i \to B}{N_i}{B}{E}]_i \\
}
{\typmc{\Gamma}{\Handle\;M\With \{\Ret x \mapsto N\} \uplus \{ (\ell_i:A_i\sto B_i)\;p_i\;r_i \mapsto N_i \}_i}{B}{E}}

\inferrule*[Lab=\btylab{Pair}]
{
  \typmc{\Gamma}{M}{A}{E} \\
  \typmc{\Gamma}{N}{B}{E} \\
}
{\typmc{\Gamma}{(M,N)}{\Pair{A}{B}}{E}}

\inferrule*[Lab=\btylab{Inl}]
{
  \typmc{\Gamma}{M}{A}{E} \\
}
{\typmc{\Gamma}{\Inl M}{A+B}{E}}

\inferrule*[Lab=\btylab{Inr}]
{
  \typmc{\Gamma}{M}{B}{E} \\
}
{\typmc{\Gamma}{\Inr M}{A+B}{E}}

\inferrule*[Lab=\btylab{CrispPairInfer}]
{
  \typmi{\Gamma}{V}{\boxwith{\ol{\mu}}(\Pair{A}{B})}{E}{} \\\\
  \typmi{\Gamma, x:\boxwith{\ol{\mu}} {A}, y:\boxwith{\ol{\mu}} {B}}{M}{A'}{E}{}
}
{\typmi{\Gamma}{\Case V\Of (x,y)\mapsto M}{A'}{E}{}}

\inferrule*[Lab=\btylab{CrispPairCheck}]
{
  \typmi{\Gamma}{V}{\boxwith{\ol{\mu}}(\Pair{A}{B})}{E}{} \\\\
  \typmc{\Gamma, x:\boxwith{\ol{\mu}} {A}, y:\boxwith{\ol{\mu}} {B}}{M}{A'}{E}{}
}
{\typmc{\Gamma}{\Case V\Of (x,y)\mapsto M}{A'}{E}{}}

\inferrule*[Lab=\btylab{CrispSumInfer}]
{
  \typmi{\Gamma}{V}{\boxwith{\ol{\mu}}(A+B)}{E}{} \\\\
  \typmi{\Gamma, x:\boxwith{\ol{\mu}}{A}}{M_1}{A_1}{E}{} \\\\
  \typmi{\Gamma, y:\boxwith{\ol{\mu}}{B}}{M_2}{A_2}{E}{}
}
{
  {\Gamma}\vdash {\Case V\Of \{\Inl x \mapsto M_1, \Inr y\mapsto M_2\}} \\\\
  \mathrel{\gray{\Rightarrow}} {A_1\join_{\Gamma,E} A_2}\atmode{E}{}
}

\inferrule*[Lab=\btylab{CrispSumCheck}]
{
  \typmi{\Gamma}{V}{\boxwith{\ol{\mu}}(A+B)}{E}{} \\\\
  \typmc{\Gamma, x:\boxwith{\ol{\mu}}{A}}{M_1}{A'}{E}{} \\\\
  \typmc{\Gamma, y:\boxwith{\ol{\mu}}{B}}{M_2}{A'}{E}{}
}
{
  {\Gamma}\vdash {\Case V\Of \{\Inl x \mapsto M_1, \Inr y\mapsto M_2\}} \\\\
  \mathrel{\gray{\Leftarrow}}{A'}\atmode{E}{}
}
\end{mathpar}
\caption{Bidirectional typing rules for \SurfaceMet.}
\label{fig:typing-maetel}
\end{figure}

\subsection{Elaboration and Implementation}
\label{app:surface-lang-implementation}

\renewcommand{\transto}[1]{\ \hl{\dashrightarrow #1}}
\renewcommand{\fortrans}[1]{\hl{#1}}

It is easy to elaborate \SurfaceMet terms into \Met terms in a
type-directed manner.
We formally define the elaboration alongside the typing rules in
\Cref{fig:elaboration-maetel-one,fig:elaboration-maetel-two}.
We have $\typmi{\Gamma}{M}{A}{E}{}\transto{N}$ for the inference mode
and $\typmc{\Gamma}{M}{A}{E}\transto{N}$ for the checking mode, both
of which elaborates $M$ in \SurfaceMet to $N$ in \Met.

We use the following auxiliary functions in the elaboration.

\[\ba{rcl}
\meta{unmod}(M; \ol{\mu}) &=& (\lambda y . \Letmod_{\ol{\mu}}\; y = y \In y)\; M
  \\
\meta{unvar}(x; A; M) &=& \Letmod_{\ol{\mu}}\; {x} = x \In M
  \quad\text{ where } A = \boxwith{\ol{\mu}}G\\
\\
\meta{join}_{\Gamma,E}(M_1 : \boxwith{\ol{\mu_1}}G_1, M_2 : \boxwith{\ol{\mu_2}}G_2,\cdots) &=&
  \bl
  (\Letmod_{\ol{\mu_1}}\; x = M_1 \In \Mod_{\ol\zeta}\; x,\\
  \hphantom{(}\Letmod_{\ol{\mu_2}}\; x = M_2 \In \Mod_{\ol\zeta}\; x,
  \cdots)
  \el \\
\span\span\quad\text{ where } \boxwith{\ol{\mu_1}}G_1 \join_{\Gamma,E} \boxwith{\ol{\mu_2}}G_2\join_{\Gamma,E}\cdots = \boxwith{\ol{\zeta}}G
\\
\ea\]

The core idea of the elaboration is similar to the greedy unboxing
strategy of the encoding we have in \Cref{app:spec-encoding}.
For \btylab{Abs} (and also handler rules), we immediately fully unbox
the bound variables.
For \btylab{Mod}, we insert explicit boxing. For \btylab{Var}, we
re-box them with the appropriate modality $\zeta$.
For \btylab{Switch}, we insert an unboxing followed by boxing.
For \btylab{App} and \btylab{TApp}, we unbox $\ol{\mu}$.
For \btylab{HandlerCheck} and \btylab{HandlerInfer}, we unbox the
bound variables in handler clauses. We do not need to unbox the
continuation functions since they have no top-level modality.
For \btylab{CrispPairCheck/Infer} and \btylab{CrispSumCheck/Infer}, we
unbox $V$ before case splitting.
Also, when the $\wedge_E$ operation is used (such as in
\btylab{CrispSumInfer} and \btylab{HandlerInfer}), we use $\meta{join}$ to
unbox the terms correspondingly.

\begin{figure}[htbp] \rulesize
\raggedright
\boxed{\typmi{\Gamma}{M}{A}{E}{}\transto{N}}
\boxed{\typmc{\Gamma}{M}{A}{E}\transto{N}}
\hfill

\begin{mathpar}
\inferrule*[Lab=\btylab{Var}]
{
  \nu_F = \locks{\Gamma'} : E\to F\\
  \boxwith{\ol{\zeta}}G = \meta{across}(\Gamma, A, \nu, F) \\
}
{\typmi{\Gamma,x:A,\Gamma'}{{x}}{\boxwith{\ol{\zeta}}G}{E}{} \transto{\Mod_{\ol{\zeta}}\;{x}}}

\inferrule*[Lab=\btylab{Mod}]
{
  \mu_F : E \to F \\\\
  \typmc{\Gamma,\lockwith{\mu_F}}{V}{A}{E} \transto{V'}
}
{\typmc{\Gamma}{V}{\boxwith{\mu} A}{F} \transto{\Mod_\mu\;V'}}

\inferrule*[Lab=\btylab{Forall}]
{
  \typmc{\Gamma,\alpha : K}{V}{A}{E} \transto{V'} \\
}
{\typmc{\Gamma}{V}{\forall\alpha:K . A}{E} \transto{\Lambda\alpha^K.V'}}

\inferrule*[Lab=\btylab{Annotation}]
{
  \typmc{\Gamma}{V}{A}{E} \transto{V'} \\
}
{\typmi{\Gamma}{V:A}{A}{E}{} \transto{V'}}

\inferrule*[Lab=\btylab{Switch}]
{
  \typmi{\Gamma}{M}{\boxwith{\ol{\mu}} G}{E}{} \transto{M'} \\
  (G, {\ol{\mu}}) \To {\ol{\nu}} \atmode{E} \\\\
  \fortrans{N = \Letmod_{\ol{\mu}}\; x = M' \In \Mod_{\ol{\nu}}\; x}
}
{\typmc{\Gamma}{M}{\boxwith{\ol{\nu}} G}{E} \transto{N}}

\inferrule*[Lab=\btylab{Abs}]
{
  \typmc{\Gamma, x : A}{M}{B}{E} \transto{M'} \\\\
  \fortrans{M'' = {\lambda x . \meta{unvar}(x;A;M')}}
}
{\typmc{\Gamma}{\lambda x . M}{A \to B}{E} \transto{M''}}

\inferrule*[Lab=\btylab{App}]
{
  \typmi{\Gamma}{M}{\boxwith{\ol{\mu}}(A \to B)}{E}{}
  \transto{M'} \\
  {\ol{\mu}}_E \To \one_E \\\\
  \fortrans{M'' = \meta{unmod}(M';\ol{\mu})} \\
  \typmc{\Gamma}{N}{A}{E} \transto{N'} \\
}
{\typmi{\Gamma}{M\; N}{B}{E}{} \transto{M''\;N'}}

\inferrule*[Lab=\btylab{TApp}]
{
  \typmi{\Gamma}{M}{\boxwith{\ol{\mu}}(\forall\alpha:K . B)}{E}{}
  \transto{M'} \\
  {\ol{\mu}}_E \To \one_E \\\\
  \fortrans{M'' = \meta{unmod}(M';\ol{\mu})} \\
  \Gamma\vdash A : K
}
{\typmi{\Gamma}{M\;A}{B[A/\alpha]}{E}{} \transto{M''\;A}}

\inferrule*[Lab=\btylab{Do}]
{
  E = \ell:A\sto B, F \\\\
  \typmc{\Gamma}{M}{A}{E} \transto{M'} \\
}
{\typmi{\Gamma}{\Do \ell \; M}{B}{E}{} \transto{\Do\ell\;M'}}

\inferrule*[Lab=\btylab{MaskCheck}]
{
  \typmc{\Gamma,\lockwith{\amk{L}_F}}{M}{A}{F-L} \transto{M'}
}
{\typmc{\Gamma}{\Mask_L\; M}{\boxwith{\amk{L}}A}{F} \transto{\Mask_L\;M'}}

\inferrule*[Lab=\btylab{MaskInfer}]
{
  \typmi{\Gamma,\lockwith{\amk{L}_F}}{M}{A}{F-L}{} \transto{M'}
}
{\typmi{\Gamma}{\Mask_L\; M}{\boxwith{\amk{L}}A}{F}{} \transto{\Mask_L\;M'}}

\inferrule*[Lab=\btylab{HandlerCheck}]
{
  D = \{\ell_i : A_i \sto B_i\}_i \\
  \typmi{\Gamma, \lockwith{\aex{D}_E}}{M}{A}{D+E}{} \transto{M'} \\
  \typmc{\Gamma, x : \boxwith{\aex{D}} A}{N}{B}{E} \transto{N'} \\
  [\typmc{\Gamma, p_i : A_i, {r_i}: B_i \to B}{N_i}{B}{E}\transto{N_i'}]_i \\
}
{\typmc{\Gamma}{\Handle\;M\With \{\Ret x \mapsto N\} \uplus \{ (\ell_i:A_i\sto B_i)\;p_i\;r_i \mapsto N_i \}_i}{B}{E} \\
\transto{\Handle\;M'\With \{\Ret x \mapsto \meta{unvar}(x;\boxwith{\aex{D}}A;N')\}
  \uplus \{ (\ell_i:A_i\sto B_i)\;p_i\;r_i \mapsto \meta{unvar}(p_i; A_i; N_i') \}_i}
}

\inferrule*[Lab=\btylab{HandlerInfer}]
{
  D = \{\ell_i : A_i \sto B_i\}_i \\
  \typmi{\Gamma, \lockwith{\aex{D}_E}}{M}{A}{D+E}{} \transto{M'} \\
  \typmi{\Gamma, x : \boxwith{\aex{D}} A}{N}{B'}{E}{} \transto{N'} \\
  [\typmi{\Gamma, p_i : A_i, {r_i}: B_i \to B'}{N_i}{B_i}{E}{} \transto{N_i'}]_i\\
  B = B' (\join_{\Gamma,E} B_i)_i \\
  \fortrans{N'',(N_i'')_i = \meta{join}_{\Gamma,E}(N' : B',(N_i' : B_i')_i)}
}
{\typmi{\Gamma}{\Handle\;M\With \{\Ret x \mapsto N\} \uplus \{ (\ell_i:A_i\sto B_i)\;p_i\;r_i \mapsto N_i \}_i}{B}{E}{} \\
\transto{\Handle\;M'\With \{\Ret x \mapsto \meta{unvar}(x;\boxwith{\aex{D}}A;N'')\} \uplus
  \{ (\ell_i:A_i\sto B_i)\;p_i\;r_i \mapsto \meta{unvar}(p_i; A_i; N_i'') \}_i}
}

\end{mathpar}
\caption{Elaboration from \SurfaceMet to \Met (part I).}
\label{fig:elaboration-maetel-one}
\end{figure}

\begin{figure}[htbp] \rulesize
\raggedright
\boxed{\typmi{\Gamma}{M}{A}{E}{}\transto{N}}
\boxed{\typmc{\Gamma}{M}{A}{E}\transto{N}}
\hfill

\begin{mathpar}
\inferrule*[Lab=\btylab{Pair}]
{
  \typmc{\Gamma}{M}{A}{E} \transto{M'} \\
  \typmc{\Gamma}{N}{B}{E} \transto{N'} \\
}
{\typmc{\Gamma}{(M,N)}{\Pair{A}{B}}{E} \transto{(M', N')}}

\inferrule*[Lab=\btylab{Inl}]
{
  \typmc{\Gamma}{M}{A}{E} \transto{M'} \\
}
{\typmc{\Gamma}{\Inl M}{A+B}{E} \transto{\Inl M'}}

\inferrule*[Lab=\btylab{Inr}]
{
  \typmc{\Gamma}{M}{B}{E} \transto{M'} \\
}
{\typmc{\Gamma}{\Inr M}{A+B}{E} \transto{\Inr M'}}

\inferrule*[Lab=\btylab{CrispPairInfer}]
{
  \typmi{\Gamma}{V}{\boxwith{\ol{\mu}}(\Pair{A}{B})}{E}{} \transto{V'} \\\\
  \typmi{\Gamma, x:\boxwith{\ol{\mu}} {A}, y:\boxwith{\ol{\mu}} {B}}{M}{A'}{E}{}
  \transto{M'}
}
{\typmi{\Gamma}{\Case V\Of (x,y)\mapsto M}{A'}{E}{}\\\\
\transto{
  \Letmod_{\ol{\mu}}\; x = V' \In \Casey_{\ol{\mu}}\; x \Of (x,y)\mapsto M'
}
}

\inferrule*[Lab=\btylab{CrispPairCheck}]
{
  \typmi{\Gamma}{V}{\boxwith{\ol{\mu}}(\Pair{A}{B})}{E}{} \transto{V'} \\\\
  \typmc{\Gamma, x:\boxwith{\ol{\mu}} {A}, y:\boxwith{\ol{\mu}} {B}}{M}{A'}{E}{}
  \transto{M'}
}
{\typmc{\Gamma}{\Case V\Of (x,y)\mapsto M}{A'}{E}{}\\\\
\transto{
  \Letmod_{\ol{\mu}}\; x = V' \In \Casey_{\ol{\mu}}\; x \Of (x,y)\mapsto M'
}
}

\inferrule*[Lab=\btylab{CrispSumInfer}]
{
  \typmi{\Gamma}{V}{\boxwith{\ol{\mu}}(A+B)}{E}{} \transto{V'} \\
  \typmi{\Gamma, x:\boxwith{\ol{\mu}}{A}}{M_1}{A_1}{E}{} \transto{M_1'} \\
  \typmi{\Gamma, y:\boxwith{\ol{\mu}}{B}}{M_2}{A_2}{E}{} \transto{M_2'} \\
  \fortrans{M_1'', M_2'' = \meta{join}_{\Gamma,E}(M_1:A_1, M_2:A_2)}
}
{\typmi{\Gamma}{\Case V\Of \{\Inl x \mapsto M_1, \Inr y\mapsto M_2\}}{A_1\join_{\Gamma,E} A_2}{E}{}
\\\\ \transto{
  \Letmod_{\ol{\mu}}\; x = V' \In \Casey_{\ol{\mu}}\; x \Of \{\Inl x \mapsto M_1'', \Inr y\mapsto M_2''\}
}
}

\inferrule*[Lab=\btylab{CrispSumCheck}]
{
  \typmi{\Gamma}{V}{\boxwith{\ol{\mu}}(A+B)}{E}{} \transto{V'} \\
  \typmc{\Gamma, x:\boxwith{\ol{\mu}}{A}}{M_1}{A'}{E}{} \transto{M_1'} \\
  \typmc{\Gamma, y:\boxwith{\ol{\mu}}{B}}{M_2}{A'}{E}{} \transto{M_2'} \\
}
{\typmc{\Gamma}{\Case V\Of \{\Inl x \mapsto M_1, \Inr y\mapsto M_2\}}{A'}{E}{}
\\\\ \transto{
  \Letmod_{\ol{\mu}}\; x = V' \In \Casey_{\ol{\mu}}\; x \Of \{\Inl x \mapsto M_1', \Inr y\mapsto M_2'\}
}
}
\end{mathpar}
\caption{Elaboration from \SurfaceMet to \Met (part II).}
\label{fig:elaboration-maetel-two}
\end{figure}

We implement a prototype of \SurfaceMet with all features mentioned
above as well as algebraic data types and pattern matching.
We do not encounter any challenges in generalising the pairs and sums
to algebraic data types.
In our implementation, we do not strictly follow the conventional
bidirectional typing approach, which distinguishes between the
checking and inference mode as in the above rules.
Instead, we use the form $\typmp{\Gamma}{M}{A}{E}{S}$ where each rule
has an input shape $S$ and output type $A$, similar to contextual
typing~\citep{XuO24}.
A shape $S$ is a type with some holes.
When $S$ is empty, we are in inference mode; when $S$ is a complete
type, we are in checking mode; otherwise, we can still use $S$ to pass
in partial type information which could allow us to type check more
programs.

Our implementation supports polymorphism with explicit type
instantiation.
As we have discussed in \Cref{sec:surface-lang}, it is natural to
extend it with inference for polymorphism, following the literature on
bidirectional typing~\citep{DunfieldK13,ZhaoO22,CuiJO23}.
We plan to explore this extension in the future.

\end{document}